%% file: paper.tex
\documentclass[11pt,a4paper]{scrartcl}
  \usepackage[greek,pdf]{preamble}
  \usepackage{extras}
  \usepackage{queuepage}
  \usepackage{abstract}

  \title{\vspace{-2cm}The \textbfAleph\ Calculus}
  \subtitle{A declarative model of reversible programming with relevance to Brownian computers}

\endofdump

  \DeclareSymbolFont{euler-letters}{U}{eur}{b}{n}
  \SetSymbolFont{euler-letters}{bold}{U}{eur}{b}{n}
  \DeclareMathAlphabet\mathbf{U}{eur}{b}{n}
  \DeclareMathSymbol\eulerSigma\mathord{euler-letters}{"06}
  \DeclareMathSymbol\eulermu\mathord{euler-letters}{"16}

\begin{document}

\maketitle
% \largertitlepage

\begin{figure}[!h]
  \centering
  \vspace{-1cm}
  \includegraphics[width=10cm]{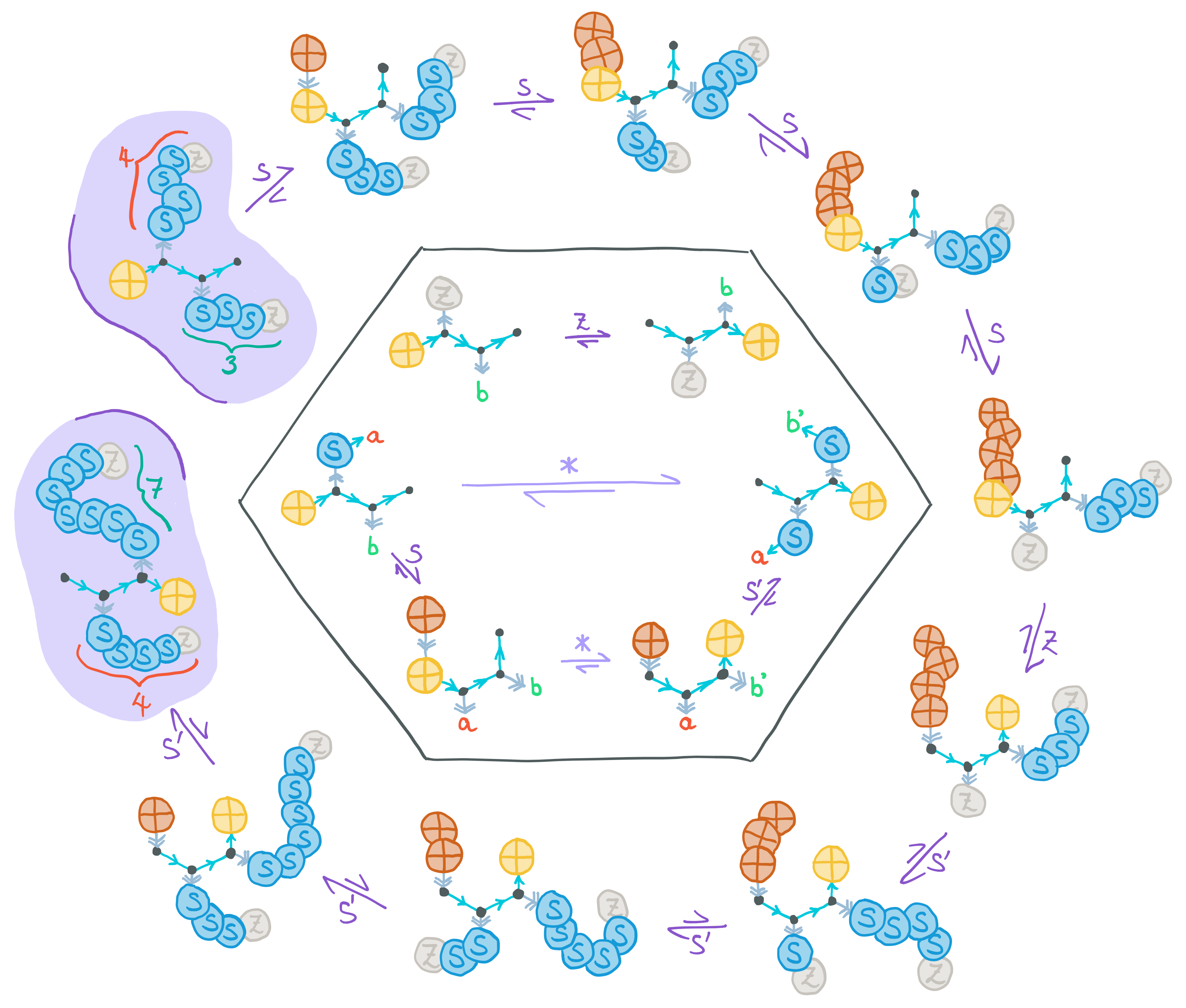}\vspace{1em}
  \caption{An abstract molecular realisation of an \textAleph\ definition of addition, here reversibly computing the sum of 4 and 3.}
  \label{fig:cover}
\end{figure}

\renewcommand{\abstractname}{Lay Summary}
\begin{abstract}
  In recent years, unconventional forms of computing ranging from molecular computers made out of DNA to quantum computers have started to be realised. Not only that, but they are becoming increasingly sophisticated and have a lot of potential to influence the future of computing. Another interesting class of unconventional computers is that of reversible computers, of which quantum computers are an example. Reversible computing---wherein state transitions must be invertible, and therefore must conserve information---is largely neglected outside of quantum computing, but is a promising avenue for realising substantial gains in computational performance and energy efficiency. We are interested in the intersection between molecular and reversible computing: although much success has been achieved in developing irreversible molecular computers driven by strong entropic forces, irreversibility can only be achieved in an approximate sense as the laws of physics are fundamentally reversible. Moreover, as molecular systems operate far closer to the underlying microscopic laws of physics, they are (arguably) influenced by this reversibility to a far greater extent. As such, we believe that reversibility is a far more natural basis for designing molecular computers.

  In this paper, we introduce a novel model of computation---the \textAleph\ calculus---that aims to be a step towards this goal. At present, \textAleph\ is a mathematical model of computation capable of concurrency, and whose objects are self-contained computational entities that are well-suited for a Brownian or molecular context. Reversibility of computation is achieved by the idea of learning and un-learning the state of variables (whence the name of the calculus, inspired by the Greek for \emph{not forgotten}, \emph{\gr ἀλήθεια}). Whilst more work needs to be done in order to realise something resembling \textAleph\ experimentally in a molecular context (for which, DNA computing is a very promising approach), we include in some of our examples what such molecular computational entities may look like---a preview of which is shown in \Cref{fig:cover}. In the meantime, computation with \textAleph\ can be realised on conventional computers via an associated programming language we also introduce, \alethe.
\end{abstract}

% \largertitlepage
\renewcommand{\abstractname}{Technical Abstract}
\begin{abstract}
  Motivated by a need for a model of reversible computation appropriate for a Brownian molecular architecture, the \textAleph\ calculus is introduced. This novel model is declarative, concurrent, and term-based---encapsulating all information about the program data and state within a single structure in order to obviate the need for a \emph{von Neumann}-style discrete computational `machine', a challenge in a molecular environment. The name is inspired by the Greek for `not forgotten', due to the emphasis on (reversibly) learning and un-learning knowledge of different variables. To demonstrate its utility for this purpose, as well as its elegance as a programming language, a number of examples are presented; two of these examples, addition/subtraction and squaring/square-rooting, are furnished with designs for abstract molecular implementations. A natural by-product of these examples and accompanying syntactic sugar is the design of a fully-fledged programming language, \alethe, which is also presented along with an interpreter. Efficiently simulating \textAleph\ on a deterministic computer necessitates some static analysis of programs within the \alethe\ interpreter in order to render the declarative programs sequential. Finally, work towards a type system appropriate for such a reversible, declarative model of computation is presented.
\end{abstract}

\clearpage

\section{Introduction}

Since the advent of thermodynamics and statistical mechanics, it has become increasingly clear that the laws of physics are intrinsically reversible at a microscopic level. The connection between statistical mechanics and information theory has also been resolved to most's satisfaction, with \textcite{szilard-engine} and \textcite{landauer-limit} providing a solution to Maxwell's eponymous thought experiment, the Maxwell D{\ae}mon (\Cref{fig:daemon}). In contrast, most models of computation are intrinsically irreversible, readily discarding information at nearly every computational step. For example, a Turing Machine~\cite{turing-machine} may overwrite or erase a square of its tape, whilst an implementation of the $\lambda$ calculus~\cite{lambda-calculus} may discard its redex. The consequence is that the heat generated by conventional computers is an unavoidable byproduct of their operation (although the amount generated by contemporary consumer processors is some 6--10 orders of magnitude greater\footnotemark\ still than the lower bound found by \textcite{landauer-limit}).
\footnotetext{It is somewhat challenging to precisely quantify the information processing capacity of a modern consumer processor due to the extensive use of pipelining, vectorisation and multiple-instruction dispatch. Nevertheless, we can obtain a reasonable estimate. arm's processors are well known for emphasising power efficiency, and are used extensively in mobile electronics as well as some laptops and desktops (notably, three of Apple's new device lineup). Considering arm's A-78 processor micro-architecture, we find~\cite{arm-a78} numbers of around \SI{1}{\watt\per{core}}, with each core running at a clock speed of \SI{3}{\giga\hertz} and executing up to \SI{6}{(macro)instructions} per clock cycle with a width of up to \SI{128}{\bit}. Taken together, these give an energy dissipation of roughly \SI{1.5e-13}{\joule\per\bit}. In contrast, at a temperature of around \SI{300}{\kelvin}, Landauer's bound gives just \SI{2.9e-21}{\joule\per\bit}, a \num{5e7} overhead for the arm processor. Intel's chips have thermal design powers typically exceeding \SI{100}{\watt}, but have compensatingly greater vectorisation widths of some instructions, such that a similar Landauer overhead can be achieved when fully exploiting the processor's capabilities.}

\begin{figure}
  \centering
  \includegraphics[width=.7\textwidth]{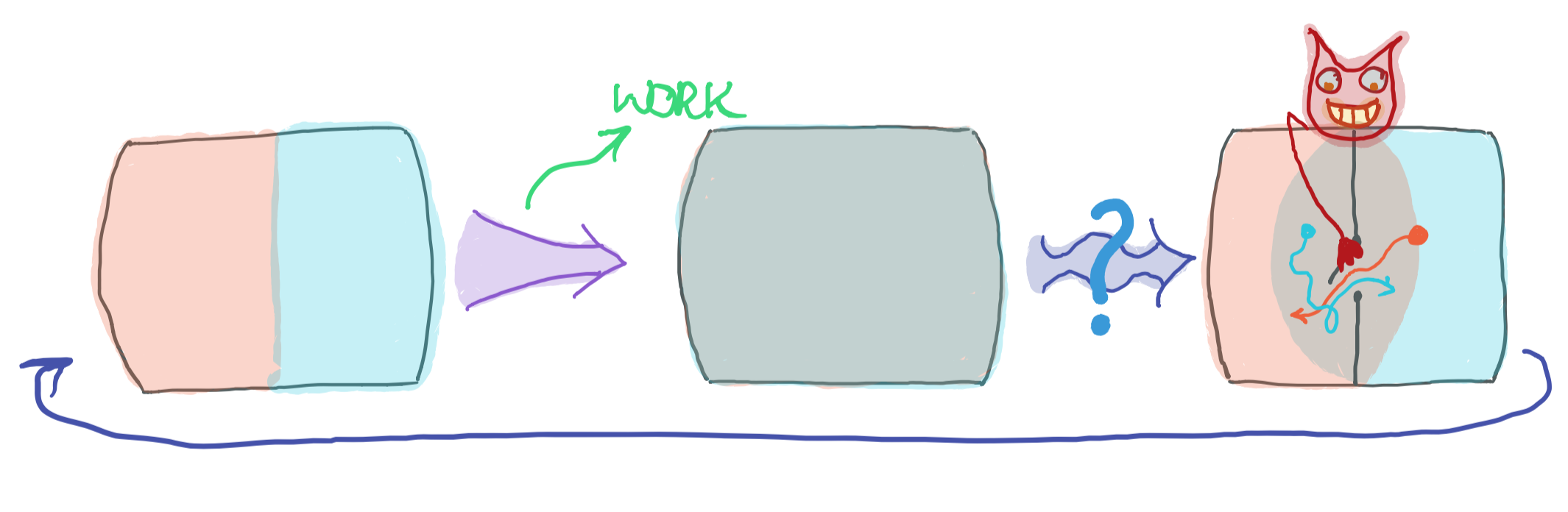}
  \caption{Maxwell's D{\ae}mon (1867) is a thought experiment which, if valid, would offer a method to violate the second law of thermodynamics. Consider a box divided in two; to the left of the divider is a red gas, and to the right a blue gas. If the divider is removed, then mechanical work can be extracted as the gases mix and equilibrate. Now suppose the divider is re-inserted, and endowed with a small window whose hinge is frictionless and can therefore be opened and closed without energetic cost. Imagine that a d{\ae}mon or other such entity, with particularly keen eyesight and nimble fingers, is positioned at the window. If she observes a red particle in the right compartment moving towards the window, or a blue particle in the left compartment moving towards the window, she briefly opens it to allow the particle's passage; otherwise she leaves the window closed. Over time, this process will bring the box back to its initial state, thereby reducing its entropy.\captionnl The problem is that this process requires the d{\ae}mon to learn information about the state of the gas particles, and this information cannot be `unlearned'; therefore her memory will eventually be filled by this stale information. Forgetting all this information requires a commensuarate increase in entropy. If this entropy manifests as heat, then Landauer showed that discarding a quantity of information $\Delta I$ requires the production of heat $\Delta Q\ge kT\Delta I$ where $k$ is Boltzmann's constant, $T$ the temperature, and equality only in the limit of a thermodynamically reversible process (i.e.\ taking infinite time).}
  \label{fig:daemon}
\end{figure}

Fortunately computation is not inherently irreversible, as \textcite{bennett-tm}---the father of reversible computing---found. He introduced a reversible analogue of the Turing Machine and even demonstrated an algorithm that could embed any irreversible computation $x\mapsto f(x)$ within a reversible computer as $x\mapsto(x,f(x))$, cleaning up any intermediate state through a reversible process rather than discarding to the environment as is standard. He later~\cite{bennett-pebbling} iterated upon this algorithm to show this could be achieved efficiently in both time and space. Nevertheless we are not satisfied with this embedding, and wish to better exploit reversible computational architectures by programming directly with reversible primitives. The benefits of doing so are that one often finds that far less temporary information need be generated than Bennett's algorithms might suggest, and also the injective embedding, $x\mapsto(x,f(x))$, retains the input which is excessive except in the trivial case of $f$ being constant; for any other function $f$ having any correlation at all between its outputs and inputs, only a partial image of $x$ need be preserved. Moreover it is often the case that one can imagine a suitable and more preferable injective or bijective embedding. For example, for the operation of addition a useful embedding might take the form $+:(x,y)\leftrightarrow(x,x+y)$, whereas Bennett's algorithm would yield $+:(x,y)\leftrightarrow(x,y,x+y)$. In defining these injections more carefully, one then often finds that the residual information of the input can in fact be made use of; for example, a Peano arithmetic implementation of $+:(x,y)\leftrightarrow(x,x+y)$ over the naturals readily begets the (domain-restricted) injection $\forall x>0.\times:(x,y)\leftrightarrow(x,xy)$ without generating any temporary data whatsoever. Moreover, when redundancy is eliminated as in these two cases one obtains the converse operation for free: that is, running $+:(x,y)\leftrightarrow(x,x+y)$ in reverse yields subtraction, and $\forall x>0.\times:(x,y)\leftrightarrow(x,xy)$ yields division, whilst their Bennett embeddings cannot do the same. A little thought shows these must fail in certain cases, and we shall have more to say about this in the following section.

A number of designs of reversible computer architectures and reversible programming languages have thus arisen. One of the first suggestive of a physical implementation was the ballistic architecture of \textcite{fredkin-conlog}, upon which \textcite{ressler} proposed a full processor design in his Master's thesis; the ballistic architecture makes use of perfectly elastic classical billiard balls projected on a frictionless surface at prescribed velocities and initial positions, which then bounce off of strategically placed walls and each other. Classical physics, being reversible, ensures that the trajectories of the balls is reversible and it was shown by Fredkin and Toffoli that these trajectories could encode arbitrary computation by devising a series of reversible logic gates. Unfortunately Bennett's analysis of this architecture showed that it could not be realised in our universe:
\def\signed #1{{\leavevmode\unskip\nobreak\hfil\penalty50\hskip2em
  \hbox{}\nobreak\hfil{#1}%
  \parfillskip=0pt \finalhyphendemerits=0 \endgraf}}
\begin{quote}
  ``Even if classical balls could be shot with perfect accuracy into a perfect apparatus, fluctuating tidal forces from turbulence in the atmospheres of nearby stars would be enough to randomise their motion within a few hundred collisions. Needless to say, the trajectory would be spoiled much sooner if stronger nearby noise sources (e.g., thermal radiation and conduction) were not eliminated.''
  \signed{\textrm{ --- \textcite{bennett-rev}}}
\end{quote}
See~\textcite{frank-thesis,earley-parsimony-i,earley-parsimony-ii,earley-parsimony-iii} for further analysis of the constraints physics puts on the performance of reversible computers. More practicable architectures have since arisen, such as Pendulum~\cite{pendulum}, as well as a host of languages~\cite{janus,psi-lisp,frank-r,kayak,theseus,lang-pi}. Furthermore, it transpires that Quantum Computers are necessarily reversible as irreversible logic can only be achieved by disrupting the well-prepared quantum state of the system, and therefore all quantum computing architectures and quantum programming languages are reversible.

In this paper we present a new model of reversible computing, the \textAleph\ Calculus, and an associated programming language \alethe\ together with interpreter, whose properties we believe to be novel. The name is inspired by the Greek meaning `not forgotten', as the semantics of \textAleph\ revolve around the transformation of knowledge: `unlearning' knowledge of one or more variables in order to `learn' knowledge about one or more other variables, but in a reversible fashion such that nothing is ever truly forgotten. \textAleph\ is declarative and concurrent and, whilst perhaps a little too abstract and high level for this purpose, is motivated by a need for a reversible model for molecular programming and DNA computing. Crucial to such applications is that almost all information about not just the program's data, but the program state too, should be encoded within a computation `term'. This is in contrast to the imperative and functional languages referenced above, wherein state information such as the instruction counter (indicating where in the program the current execution context is) is implicitly stored in a special register or other hidden state of the processor; in \textAleph, the program itself is (reversibly) mutated during the course of its execution, such that the distinction between `program' and `processor' vanishes. The reason for this design is that a \emph{von Neumann}-style architecture, in which there is a discrete processing unit that interacts with a memory unit to execute a program, is unsuitable for a molecular context as it is---in some sense---too `bulky', and is difficult to engineer. We assert that the proposed design is far more suitable to molecular implementation, or at least serves as a step towards this goal, and it is hoped the reader will be convinced of this by the examples and semantics illustrated herein.

\input{part-ex}
\input{part-aleph}
\input{part-impl}

\section{Conclusion}

The examples furnishing this paper demonstrate the utility of the \textAleph\ calculus, as well as its appropriateness as a candidate model targeting reversible molecular and Brownian computational architectures. Though a molecular implementation is as yet unrealised, an interpreter for the programming language \alethe\ serves as a useful testbed for \textAleph. Going forward, we hope to extend the interpreter to support the full concurrent language. Additionally, development of a type system appropriate for a language that is both reversible and declarative would be helpful for improving code analysis and reducing programmer errors, and work towards this is underway (see \Cref{app:typing}). Lastly, it is hoped that \textAleph\ can be realised experimentally in a molecular context.

\appendix
\section{Acknowledgements}

The author would like to acknowledge the invaluable help and support of his supervisor, Gos Micklem. This work was supported by the Engineering and Physical Sciences Research Council, project reference 1781682.

\input{part-alethe}
\input{part-type}
\input{part-sigma}

\begin{minipage}{\textwidth}%
  \vspace{-\baselineskip}
  \printbibliography
\end{minipage}

%%%% Supplementary material

\ForEachFloatingEnvironment{\SetupFloatingEnvironment{#1}{within=section}}
\setcounter{section}{18}%
\begin{listing}
  \vspace{-1em}%
  \section{Supplementary Material}
  \vspace{-2em}%
  \centering%
  \input{polish}
  \caption{The full \alethe\ program implementing interconversion between trees and Polish notation, following the \texttt{haskell} snippet in \Cref{lst:serial-polish-hs}. In fact, a little thought and optimisation will reveal that the $\ell$ outputs of the read functions can be eliminated, which then simplifies the definition of \atom{Polish} to that given in \Cref{lst:serial-polish-nice} and obviates the need for the auxiliary functions \atom{Length}, \atom{Sum}, \atom{Map}, etc. Nevertheless, the point of this is to show that it is not unreasonable to arrive at the program listed here, and it may not be entirely obvious that a simpler implementation with less garbage is possible.}
  \label{lst:supp-polish}
\end{listing}

\end{document}

%% file: part-ex.tex
\section{\textbfAleph\ by Example}
\label{sec:ex1}

Before expositing the formal semantics of \textAleph, it is illuminating to introduce a number of examples of \textAleph\ in order to gain an intutition for its syntax and features. We begin with the reversible definition of natural addition from the introduction, but first there is a point to be made regarding injectivity and bijectivity. Recall the Bennett-style reversible embedding of addition, $+_{\text{Bennett}}:(x,y)\leftrightarrow(x,y,x+y)$; it is easy to see that the forward direction of this program maps, e.g., $(3,4)$ to $(3,4,7)$, and likewise the reverse direction maps $(3,4,7)$ to $(3,4)$. Clearly the forward direction is an injection, but what about the reverse? Suppose we attempt to feed in $(3,4,5)$: as $3+4\ne5$, and as the forward direction is injective, there can be no pair $(a,b)$ that maps to $(3,4,5)$. In fact, even our less redundant embedding $+:(x,y)\leftrightarrow(x,x+y)$ suffers from non-injectivity of its inverse, for example there is no value $y\in\mathbb N$ satisfying $+:(5,y)\leftrightarrow(5,2)$. Yet another example occurs for the function $+':(x,y)\leftrightarrow(x-y,x+y)$ defined over the integers, in which there is no pair $(x,y)$ satisfying $+':(x,y)\leftrightarrow(3,6)$ (there is in $\mathbb Q$, however; namely, $(\frac92,\frac32)$). Therefore, whilst we might expect reversible computers to compute bijective operations, this appears to be violated in even the simple example of addition. In fact, no violation of reversibility has occurred, and the `primitive' steps of any reversible computer \emph{are} bijective. What we have encountered here is simply a (co)domain error, the same as if we were to ask a conventional computer to evaluate $1/0$. That is, whereas conventional computers compute \emph{partial} functions, reversible computers compute \emph{partial} bijections (or partial isomorphisms). Exactly what happens when such an error is encountered depends on both the algorithm and architecture in question; for example, attempting to divide 1 by 0 may cause a computer to immediately complain, or it may enter an infinite loop if the algorithm is `repeatedly subtract the divisor until the dividend vanishes'. We shall have more to say about how \textAleph\ handles such errors in due course.

\para{Addition}
\input{examples}

Our examples will mostly concern natural numbers because they are particularly amenable to inductive and recursive definitions of both their structure and operations over them, such as addition and multiplication. The standard approach to this is the Peano axiomatic formulation, in which a natural number is defined to either be $\Z\equiv 0$, or $\S n\equiv n+1$ whenever $n$ is a natural number---i.e.\ the \emph{successor} of $n$. For example, 4 is constructed as $\S(\S(\S(\S\Z)))$. In fact there are seven more axioms in order to clarify such subtle points as uniqueness of representation, conditions for equality, and non-negativity. Addition can then be defined recursively by the base case $\Z+b=b$ and the inductive step $\S a+b=\S(a+b)$, and multiplication by $a\cdot\Z=\Z$ and $a\cdot\S b=a+a\cdot b$. To render addition reversible, we write $+:(\Z,b)\leftrightarrow(\Z,b)$ and $+:(\S a,b)\leftrightarrow(\S a,\S(a+b))$, realising the proposed embedding $+:(x,y)\leftrightarrow(x,x+y)$. Multiplication is embedded similarly, but there is a domain-restriction imposed in that $a$ must not be zero (or else $b$ cannot be uniquely recovered); $b$ may, however, be zero. In \textAleph, this is written thus
  \begin{align*}
    & {+}~{\Z}~{b}~{\unit} = {\unit}~{\Z}~{b}~{+}{;} && \ruleName{add--base} \\
    & {+}~({\S} {a})~{b}~{\unit} = {\unit}~({\S} {a})~({\S} {b'})~{+}{:} && \ruleName{add--step} \\
    &\qquad  {+}~{a}~{b}~{\unit} = {\unit}~{a}~{b'}~{+}{.} && \ruleName{add--step--sub}
  \end{align*}
and is perhaps best understood by example. In \Cref{lst:ex-add}, we perform the example addition of $4$ and $3$ and, as promised, an example failure mode in which we attempt to subtract $5$ from $2$.

\textAleph\ can thus be seen, in a loose sense, to be a term-rewriting system. It is `loose' in the sense that its `sub-terms' exist `separate' from their parent term. In the example of \Cref{lst:ex-add}, an addition term is written ${+}~{a}~{b}~{\unit}$ and is mapped by the transition rules to ${\unit}~{a}~{c}~{+}$ where $c\equiv a+b$; here $+$ is an `atom', $a$, $b$ and $c$ are terms representing natural numbers (as composite terms formed from nested applications of the atoms $\S$ and $\Z$), and $\unit$, or `unit', is the empty term and is used by convention in \textAleph\ to avoid certain ambiguities. A program corresponds to a series of definitions of transition rules which pattern match on terms, and then substitute the variables into an output pattern. This matching process is subject to certain constraints that ensure reversibility. In addition, a rule may specify sub-rules that indicate how to transform knowledge of some variable, e.g.~$b$, to knowledge of another, e.g.~$b'$, and is the primary mechanism of composition in \textAleph. Inherent to the semantics of the calculus is a secondary composition mechanism, which was the only mechanism available in an earlier iteration of this calculus (see \Cref{app:sigma}): composite terms at any level are all subject to the same transition rules. This behaviour is more reminiscent of functional programming languages, but is somewhat clumsy in practice; nevertheless it is not without utility in \textAleph---in particular, it is well suited for when a continuation-passing style approach is favoured. The sub-rule mechanism, on the other hand, is more reminiscent of declarative programming languages. 

When there is no matched rule, such as in \Cref{lst:ex-add-52}, this simply means that there is no successor state and so computation cannot continue. It is in fact very similar in nature to the case when computation succeeds, as then we obtain a term which has a predecessor state via some rule, but lacks a rule to generate a successor state. To properly distinguish these two conditions, we must explicitly mark `true' halting states; for addition, this is written as
\begin{align*}
  & {!}~{+}~{\blank}~{\blank}~{\unit}{;}\quad {!}~{\unit}~{\blank}~{\blank}~{+}{;}
\end{align*}
where $!$ (`bang') indicates that any term of the specified forms is a valid computational output. The former corresponds to the output of a subtraction, and the latter to the output of an addition. This subtlety is further contextualised by \Cref{fig:rev-classes}.

We conclude this first example by alluding to a possible molecular implementation, as depicted in \Cref{fig:ex-plus}. Here we see the importance of \textAleph\ being interpretable as a term-rewriting system, as the entirety of the state of the computation must be encoded within a single macromolecule (or molecular complex). Strictly speaking, this is not a requirement as DNA Strand Displacement systems~\cite{dsd} achieve computation without this requirement; in payment for this, though, the entire reaction volume is dedicated to the same (typically analogue) program. The precise mechanism of the reactions is omitted from the figure, instead expressing the model as an abstract chemical reaction network. Finding possible reaction mechanisms to implement \textAleph, or similar calculi, in real chemical systems will be the subject of future work. Whilst the abstract molecular formalism is attractive form an explanatory perspective, \textAleph\ provides a more concise formalism that is also easier to manipulate and to explicate its semantics.

\para{Squaring}

Eliminating redundancy in the definition of addition yielded its inverse for free, but one may still object that the additional output is `garbage' and not of any utility. What will happen if one uses this addition subroutine many times in a larger program? Naively we may expect this garbage to accumulate, requiring either active dissipation of the additional entropy or the application of Bennett's algorithm to clean it up. In fact, by retraining one's thought process from the irreversible programming paradigm to the reversible paradigm, it is often possible to make use of this garbage data. We demonstrate this with the example of finding the square of a natural number. The candidate function, $\square:n\leftrightarrow n^2$, is an injection and so clearly meets our requirement of partial bijection. Therefore, we have good reason to believe that it is possible to implement it. The obvious approach of using multiplication will not work because its reversible embedding will take a form not dissimilar to $\times:(m,n)\leftrightarrow(m,m\cdot n)$, and so would yield $\square:n\leftrightarrow(n,n^2)$. Whilst this suffices for realising the square of a number, it retains too much redundancy in its outputs.

Often a helpful tactic is to consider an inductive approach. For the square numbers this is encapsulated by $(n+1)^2-n^2=2n+1$, from which can be obtained the identity $n^2\equiv\sum_{k=0}^{n-1}2k+1$. Again, we need to be clever: in order to achieve our desired (partial) bijection, we need to completely consume our input value of $n$. This can be done by evaluating the sum in reverse. Instantiate a new variable, $m=0$; as $m$ is set to a known value, this is reversible. Then, perform the following loop until $n$ reaches 0: decrement $n$, add $2n+1$ to $m$, repeat. At the end of this loop, $n$ will have reached a unique value (and can thus be reversibly destroyed) and $m$ will have been set to the square of the original value of $n$. In addition, this can be implemented with our addition subroutine in two steps, by adding $n$ twice to $m$ (retaining the value of $n$) and finally incrementing $m$. This is implemented in \textAleph\ in \Cref{lst:ex-square-def}, and the example of squaring $3$ (equivalently, taking the square root of $9$) is presented in \Cref{lst:ex-square}.

\section{\textbfAleph\ by Example 2: Parallelism \& Concurrency}
\label{sec:ex2}

Having introduced the essence of \textAleph\ in the previous section, we now dive deeper into some more advanced features and examples of \textAleph.

\para{Sugar}

To reduce boilerplate and increase clarity in longer programs, it is helpful to introduce some syntactic sugar (shorthands). More sugar will be introduced later for the definition of the programming language \alethe\ (which is really just sugared \textAleph), but for now only a light sprinkling is required.

Many rules take the rote form
\begin{align*}
  & {!}~{f}~{x_1}~{x_2}~\cdots~{x_m}~{\unit}{;}\quad
    {!}~{\unit}~{y_1}~{y_2}~\cdots~{y_n}~{f}{;} \\
  & {f}~{x_1}~{x_2}~\cdots~{x_m}~{\unit}={\unit}~{y_1}~{y_2}~\cdots~{y_n}~{f}{:} \\
  &\qquad \cdots
\end{align*}
which we can abbreviate with an infix form as
\begin{align*}
  & {x_1}~{x_2}~\cdots~{x_m}~\infix{f}~{y_1}~{y_2}~\cdots~{y_n}{:} \\
  &\qquad \cdots
\end{align*}
where the halting patterns are implied. In the special case of $f$ a single symbol, such as $+$, $\times$ or $\square$, we omit the backticks and write, e.g., ${4}~{3}~{+}~{4}~{7}$. Note that, if $f$ is a composite term, then we additionally assert ${!}~{f}$.

We have already seen sugar for numeric data, e.g.\ $4\equiv\S(\S(\S(\S\Z)))$. Another common data type is that of lists. As is standard in the functional world, we opt for singly linked lists implemented as composite pairs. If $\Cons~{x}~{y}$ represents the pair $(x,y)$ and $\Nil$ the empty list, then the list $[{5}~{2}~{4}~{3}]$ has corresponding representation $\Cons~{5}~(\Cons~{2}~(\Cons~{4}~(\Cons~{3}~\Nil)))$. We also introduce sugar for matching on a partial prefix of a list, i.e.\ $[{x}~{y}~{\cdot}~\cursivezS]$ corresponds to $\Cons~{x}\ (\Cons~{y}\ \cursivezS)$.

\para{Parallelism}

With this sugar thus defined, we can rewrite the somewhat clumsy definition of recursively mapping a function $f$ over a list,
\begin{align*}
  & {!}~{\atom{Map}}~{\blank}{;}\quad {!}~{({\atom{Map}}~{\blank})}~{\blank}~{\unit}{;}\quad {!}~{\unit}~{\blank}~{({\atom{Map}}~{\blank})}{;} \\
  & {({\atom{Map}}~{f})}~{\Nil}~{\unit} = {\unit}~{\Nil}~{({\atom{Map}}~{f})}{;} \\
  & {({\atom{Map}}~{f})}~{({\Cons}~{x}~\cursivexS)}~{\unit} = {\unit}~{({\Cons}~{y}~\cursiveyS)}~{({\atom{Map}}~{f})}{:} \\
  &\qquad {f}~{x}~{\unit} = {\unit}~{y}~{f}{.} \\
  &\qquad {({\atom{Map}}~{f})}~\cursivexS~{\unit} = {\unit}~\cursiveyS~{({\atom{Map}}~{f})}{.}
\end{align*}
more concisely and clearly as 
\begin{align*}
  & []~\infix{\atom{Map}~{f}}~[]{;} \\
  & [{x}~{\cdot}~\cursivexS]~\infix{\atom{Map}~{f}}~[{y}~{\cdot}~\cursiveyS]{:} \\
  &\qquad {x}~\infix{f}~{y}{.} \\
  &\qquad \cursivexS~\infix{\atom{Map}~{f}}~\cursiveyS{.}
\end{align*}
Notice that the order in which the sub-rules are executed does not make a difference to the final result, due to referential transparency; in fact \textAleph, being declarative, does not ascribe any importance to the ordering of the statements. Moreover, should a rule not be necessary for the final computation (or if there are multiple routes to the answer) then that rule will not necessarily be executed. A rule may even be evaluated more than once; for example, Bennett's algorithm may be implemented as
\begin{align*}
  & {x}~\infix{\atom{Bennett}~{f}}~{x}~{y}{:} \\
  &\qquad {x}~\infix{f}~\textit{garbage}~{y}{.}
\end{align*}
wherein the function $f$ will be evaluated once in the forward direction, its output $y$ will be copied, and then $f$ will be evaluated again in the reverse direction to consume the garbage. This arbitrarity in rule ordering and execution is important to its ability to operate in a stochastic system such as a molecular context, although in practice a compilation pass that chooses and enforces an optimal execution plan is important for efficiency.

The implicit duplication of variables that may occur means that one possible interpretation of \atom{Map} is
\begin{align*}
  & []~\infix{\atom{Map}~{f}}~[]{;} \\
  & [{x}~{\cdot}~\cursivexS]~\infix{\atom{Map}~{f}}~[{y}~{\cdot}~\cursiveyS]{:} \\
  &\qquad \infix{\atom{Dup}~{f}}~{f'}{.} \\
  &\qquad {x}~\infix{f}~{y}{.} \\
  &\qquad \cursivexS~\infix{\atom{Map}~{f'}}~\cursiveyS{.}
\end{align*}
from which we can see that not only is the order of the sub-rules arbitrary, but that they can be evaluated in parallel as the following example makes clear:
\input{aleph-ex-map}

There remains a subtle point to be made: variables can be implicitly duplicated if they are used by multiple rules, but there is then a contract made with these rules that they really do return the variable unchanged. It is not possible to ensure this statically, however, and so it is entirely possible that the copies of some variable may diverge. In this case, running \atom{Dup} in reverse will fail, and hence computation will stall. If duplication is not used, then computation will occur linearly and the changed value may be fed into subsequent rules unnoticed. In this case, computation may run to completion but yield an incorrect result due to the logic error.

\para{Sorting}

A larger example program, which is capable of sorting a list using an arbitrary comparison function, is presented in \Cref{lst:ex-sort}. Whilst of poor algorithmic complexity, insertion sort is employed for simplicity. A more efficient implementation using merge sort is available in the accompanying standard library\footnotemark. Clearly sorting is an irreversible process, as its purpose is to discard information regarding the original ordering of the list. The presented sorting algorithm puts a little effort towards increasing the utility of the garbage, in that the garbage data is a list saying where the original list item was placed into the sorted list \emph{at the moment of insertion}. The insertion sort included in the standard library (\Cref{lst:ex-sort2}) applies some additional processing to make this garbage data correspond to the permutation that maps the original list to the sorted list.
\footnotetext{The \alethe\ standard library is available from \url{https://github.com/hannah-earley/alethe-examples}.}

\para{Concurrency}

In addition to automatic parallelisation of independent sub-rules, \textAleph\ also supports defining transitions between separate terms. The motivation for this is for \textAleph\ to be able to fully exploit molecular architectures, but its utility is general and it is intended that a future version of the \alethe\ interpreter will support concurrency.

\subparagraph*{Biasing Computation}

In \Cref{fig:ex-plus} and \Cref{lst:ex-square-def-mol}, abstract molecular implementations of addition and squaring were introduced respectively. Notably, the reaction arrows---whilst reversible---were biased in the forward direction. From a thermodynamic perspective, the direction in which a reaction occurs cannot be specified in an absolute sense, and depends on the conditions of the reaction volume at a given point in time. Moreover, at equilibrium the net direction of each reaction is null: they make no net progress. These concerns are discussed in more detail in \textcite{earley-parsimony-i}, but suffice it to say that trying to arrange the conditions of the system such that the computational terms themselves are inherently biased is impracticable, and will also limit the number of transition steps that can be executed. Moreover, it makes it difficult to run sub-rules in a reverse direction (such as is needed for Bennett's algorithm).

The solution is to separate the concerns of computation and biasing said computation. In biochemical systems, this is (mostly) achieved by using a common free energy carrier in the form of $\ce{ATP}$ and $\ce{ADP + P_{\textrm{i}}}$, which are held in disequilibrium such that the favourable reaction direction is $\ce{ATP + H2O <=>> ADP + P_{\textrm{i}}}$. Other biochemical reactions are then coupled to one or more copies of this hydrolysis reaction. For our purposes, we generalise this to assume that free energy is available in the form of two terms $\oplus$ and $\ominus$, where the concentration of $\oplus$ in the reaction volume is greater than that of $\ominus$. The `computational bias' quantifying the average net proportion of computational transitions that are successful is found to be $b=([\oplus]-[\ominus])/([\oplus]+[\ominus])$.

With a bias system thus defined, we can couple any computational transition to it by simply having the input side consume a $\oplus$ term and the output side release a $\ominus$ term. The more transitions that couple to the bias source, the faster and more robustly the computation proceeds. For example, the square definition can be amended like so:
\begin{align*}
  & {!}~{\square}~{\blank}~{\unit}{;}\quad {!}~{\unit}~{\blank}~{\square}{;} \\
  & \begin{parties}&\oplus \\ &{\square}~{n}~{\unit}\end{parties} = \begin{parties}&\ominus \\ &{\square}~{n}~{\Z}~{\square}\end{parties}{;} && \ruleName{square--init} \\
  & \begin{parties}&\oplus \\ &{\square}~{\Z}~{m}~{\square}\end{parties} = \begin{parties}&\ominus \\ &{\unit}~{m}~{\square}\end{parties}{;} && \ruleName{square--term} \\
  &\begin{parties}&\oplus \\ &{\square}~({\S}{n})~{m}~{\square}\end{parties} =
   \begin{parties}&\ominus \\ &{\square}~n~({\S}{m''})~{\square}\end{parties}{:} && \ruleName{square--step}\\
  &\qquad {+}~{n}~{m}~{\unit} = {\unit}~{n}~{m'}~{+}{.} && \ruleName{square--step--sub$_1$}\\
  &\qquad {+}~{n}~{m'}~{\unit} = {\unit}~{n}~{m''}~{+}{.} && \ruleName{square--step--sub$_2$}
\end{align*}
The braces indicate that these definitions are concurrent, in that the transition will draw (release) two terms from (into) the reaction volume. As written above, the scheme is fairly basic in that it only drives computation in one direction. An improved scheme may make use of a hidden variable in each term to indicate its preferred direction of computation. If this direction is reversed, then the transition will instead draw a $\ominus$ term and release a $\oplus$ term. If a sub-rule then needs to be run in the opposite direction to its parent, then the hidden variable need simply be inverted. This scheme can be made even more sophisticated by allowing the hidden variable to refer to alternate bias sources in case there are multiple to which the transition could couple (as is recommended by \textcite{earley-parsimony-ii,earley-parsimony-iii}).

\subparagraph*{Communication}

A more canonical application of concurrency is that of communication between different computers. There are many possible communication schemes, and an overview and analysis of schemes appropriate for Brownian computers is presented by \textcite{earley-parsimony-ii}, but we shall consider the simple example of an open communication channel between two computers, \atom{Alice} and \atom{Bob}. \atom{Alice} has a list, the contents of which she wishes to send to \atom{Bob}. A naive approach to this may resemble the \textAleph\ definition
\begin{align*}
  & \atom{Alice}~[{x}~{\cdot}~\cursivexS] = \begin{parties}&\atom{Alice}~\cursivexS\\&\atom{Courier}~{x}\end{parties}{;} &
  & \begin{parties}&\atom{Bob}~\cursiveyS\\&\atom{Courier}~{y}\end{parties} = \atom{Bob}~[{y}~{\cdot}~\cursiveyS]{;}
\end{align*}
where the \atom{Courier} terms are used to convey list items between the two computers. Unfortunately, this does not quite realise the desired behaviour in a Brownian context: Suppose \atom{Alice} starts with the list $[1~2~3~4~5]$. She will generate the courier terms $(\atom{Courier}~1)$, $(\atom{Courier}~2)$, \ldots, $(\atom{Courier}~5)$, as expected, but these terms will be delievered to \atom{Bob} via a diffusive approach. Except for the case of a one-dimensional system, \atom{Bob} will almost certainly receive these couriers in a random order, completely uncorrelated from the original order. This may be avoided if delivery over the channel is substantially faster than the rates of term fission/fusion reactions, but this is not likely in a chemical system. Moreover, this being a reversible Brownian system, \atom{Alice} will from time to time re-absorb a courier that she previously dispatched, thereby increasing the opportunity for the list order to be shuffled.

The take-home message, here, is that whilst the local dynamics of a concurrent reversible computation system---such as \textAleph---may well be reversible, the global dynamics are not necessarily reversible. In fact, this is a common property of microscopically reversibly systems, and is the basis of thermodynamics: the manifestation of macroscopically irreversibly dynamics from microscopically reversible physics. It is doubtful that a model of concurrent reversible programming could preclude the possibility of such increases in entropy without severely restricting the system: likely any attempt to do so would effectively prevent the use of concurrency. Nevertheless, the programmer has a high degree of control here and so can, in principle, avoid entropy generation except where desired. In the above case of sending an ordered list, one could for example explicitly sequence the couriers. This suggests that a system of reversible molecular computers could have exceptionally low entropy generation and could realise very strange behaviours by getting as close to the implementation of a Maxwell D{\ae}mon as physics allows. Conversely, one is also free to exploit the thermodynamic properties of the system; for example, one could realise (to the extent the laws of physics permit) a true random number generator.

\subparagraph*{Resources}

Our last example of concurrency concerns the distribution of conserved resources amongst computers. If our computers are to respect the reversibility of the laws of physics, they should also respect mass conservation and so will need to contend with their finity. Amorphous computing presents one approach to achieving powerful computation from limited computational subunits, but it does so by making extensive use of communication which precludes the ability to perform meaningfully reversible computation. We instead suppose that the computers can exchange resources, such as memory units and structural building blocks, with the environment. Whilst this also has thermodynamic consequences~\cite{earley-parsimony-iii}, it at least separates the concerns of computation from resource access and thus preserves the ability to perform reversible computation. 

A host of resource distribution schemes, and thermodynamic analyses thereof, are presented by \textcite{earley-parsimony-iii}, but we content ourselves here with demonstrating how computers may interact with resources free in solution. In particular, we amend the definition of addition to be mass-conserving:
\begin{align*}
  & {+}~{\Z}~{b}~{\unit} = {\unit}~{\Z}~{b}~{+}{;} && \ruleName{add--base} \\
  & \begin{parties}&\S\\&{+}~({\S} {a})~{b}~{\unit}\end{parties} = {\unit}~({\S} {a})~({\S} {b'})~{+}{:} && \ruleName{add--step} \\
  &\qquad  {+}~{a}~{b}~{\unit} = {\unit}~{a}~{b'}~{+}{.} && \ruleName{add--step--sub}
\end{align*}

\para{Effects and Contexts}

As briefly mentioned, composite terms are subject to the same transition rules. Suppose \atom{Lena} is learning \textAleph\ and experiments with this, creating the following trivial wrapper around $\square$:
\begin{align*}
  & {!}~{\blank}~\infix{\atom{MySquare}}~{\blank}{;} \\
  & {\atom{MySquare}}~{n}~{\unit} = {\atom{MySquare}}~{({\square}~{n}~{\unit})}~{\atom{MySquare}}{;} \\
  & {\atom{MySquare}}~{({\unit}~{m}~{\square})}~{\atom{MySquare}} = {\atom{MySquare}}~{m}~{\unit}{;}
\end{align*}
This definition, if a little contrived, will function as intended. Suppose, now, that she wants to inspect what happens within the loop of $\square$, and so writes the following:
\begin{align*}
  & {!}~{\blank}~\infix{\atom{MySquare'}}~{\blank}{;} \\
  & {\atom{MySquare'}}~{n}~{m}~{\unit} = {\atom{MySquare'}}~{({\square}~{n}~{m}~{\square})}~{\atom{MySquare'}}{;} \\
  & {\atom{MySquare'}}~{({\square}~{n'}~{m'}~{\square})}~{\atom{MySquare'}} = {\atom{MySquare'}}~{n'}~{m'}~{\unit}{;}
\end{align*}
Now, what happens if we instantiate the term ${\atom{MySquare'}}~{3}~{\Z}~{\unit}$? Well, the final term could be any of ${\unit}~{3}~{\Z}~{\atom{MySquare'}}$, ${\unit}~{2}~{5}~{\atom{MySquare'}}$, ${\unit}~{1}~{8}~{\atom{MySquare'}}$ or ${\unit}~{\Z}~{9}~{\atom{MySquare'}}$. Whilst this arguably achieves \atom{Lena}'s aim of inspecting the execution of $\square$, there is a problem in that \textAleph\ claims all of these terms are halting states yet \Cref{fig:rev-classes} asserts there should be a maximum of two halting states. Despite being contrived, this example shows that we need to introduce another constraint into \textAleph\ to prevent unexpected entropy generation: Any composite term or sub-term can only be created in a halting state, and can only be consumed in a halting state. Moreover, the creation and consumption patterns must be unambiguous as to which of the term's two halting states they refer (if, indeed, there are two). Otherwise, each instance of this ambiguity would multiply the state space and thus there would be an exponential increase in the size of the state space over time.

With the constraint on composite terms clarified, one can then ask whether all terms must follow the same set of rules. It turns out there is an important case where distinguishing between composite terms and top-level terms is useful; specifically, some cases of effectful computation. Suppose that the reaction volume has been endowed with a spatial lattice along which terms can travel, and imagine a term $(\atom{Cons}~\atom{Charlie}~\atom{Dan})$ with such a wanderlust. If the composite terms \atom{Charlie} and \atom{Dan} want to travel to different locations, then where will the term end up? It is likely that a tug-of-war will occur, and so it makes sense to restrict the effecting of translocation to top-level terms.

This distinction is achieved by introducing term contexts: a top-level term has a top-level context, which is some label (itself a term) that may contain information, whilst composite terms have `one-hole contexts'. For example, a term attached to a lattice may have context \atom{Lattice} while a term free in solution might have context \atom{Free}. Concretely, the aforementioned tug-of-war may be resolved by giving \atom{Charlie} control by making him the top-level term, rendered as $\atom{Lattice}:\atom{Charlie}~{(\atom{Dan})}$. Meanwhile, \atom{Dan} is rendered as $(\atom{Lattice}:(\atom{Charlie}~{\bullet})):\atom{Dan}$ where $\bullet$ indicates that this is a one-hole context. One-hole contexts can be defined for many data structures, but for a tree they correspond to removing some sub-tree of interest, replacing it with a `hole'; see \textcite{ohc-diff} for some interesting properties of one-hole contexts. Rules may pattern match on top-level contexts and thus consume the information held within, or even create and destroy top-level terms, but they cannot match on one-hole contexts as this would risk altering their structure; instead, one-hole contexts can only be matched by `opaque variables'. An opaque variable is a special variable found only in context patterns, which can match against a top-level context or a one-hole context, whilst regular variables in a context pattern can only match against top-level contexts. That is, the following are allowed
\begin{align*}
  & \begin{parties}\atom{Lattice}&:\atom{Charlie}~{x}\end{parties} = 
  \begin{parties}\atom{Lattice}&:\atom{Charlie}~{\unit} \\ \atom{Free}&:{x}\end{parties}{;} \\
  & \begin{parties}{\gamma}&:\atom{Dan}~{[\atom{Dan's Stuff}]}\end{parties} =
  \begin{parties}{\gamma}&:\atom{Dan}\\\atom{Lattice}&:{[\atom{Dan's Stuff}]}\end{parties}{;}
\end{align*}
whilst this is not
\begin{align*}
  & \begin{parties}(\atom{Lattice}:(c~{\bullet}))&:\atom{Dan}\end{parties} =
  \begin{parties}(\atom{Lattice}:({\bullet}~\atom{Dan}))&:c\end{parties}{;}
\end{align*}
This is not too onerous a restriction, as if one wishes to manipulate the structure of the one-hole context one can simply match against a higher level context.

In the earlier \textAleph\ definitions, contexts were missing from the rule patterns. This can be seen as another example of syntactic sugar, with a missing context implying an opaque variable context (i.e.\ $\gamma:$) such that the rule can match against any term at any level. For concurrent definitions, however, all participating terms must be contextualised as otherwise it is unclear how to assign the results to the bound one-hole contexts.

It is as yet unclear how translocation along a lattice, or other effects, is actually achieved. Typically effects will be introduced as additional computational primitives, as by definition their actions are not `computational'. As these computational primitives must be instantiated as top-level terms, we require a way to interface between computational terms and effector terms and this warrants a continuation-passing-style approach. For example, walking along a lattice from coordinate $\alpha$ to coordinate $\beta$ may be realised thus:
\begin{align*}
  & {!}~\atom{Charlie}'~{\blank}{;} \quad {!}~\atom{Charlie}''~{\blank}{;} \\
  & \begin{parties}\atom{Latt}&:\atom{Charlie}~{\beta}~{x}\end{parties} =
    \begin{parties}\atom{Latt}&:\atom{Walk}~{\beta}~{(\atom{Charlie}'~x)}\end{parties}{;} \\
  \mbox{\emph{(primitive)}}\quad &
    \begin{parties}\atom{Latt}&:\atom{Walk}~{\beta}~c\end{parties} \leftrightarrow
    \begin{parties}\atom{Latt}&:\atom{Walk}'~{\alpha}~c\end{parties}{;} \\
  & \begin{parties}\atom{Latt}&:\atom{Walk}'~{\alpha}~{(\atom{Charlie}''~x)}\end{parties} =
    \begin{parties}\atom{Latt}&:\atom{Charlie}'''~{\alpha}~{x}\end{parties}{;}
\end{align*}
That is, \atom{Walk} is given a coordinate to travel to as well as a continuation (which is free to perform additional computation during the walk, if desired). \atom{Walk} then replaces the destination coordinate with the origin coordinate to ensure reversibility. Finally, we define a transition from $\atom{Walk}'$ in order to return control to the continuation. This reveals a subtle point, that the effectful primitives are intentionally not marked as halting so that control can be transferred to and fro' them.

\flushpagequeue

%% file: examples.tex
\queuefloat{listing}
  \def\monus{\mathbin{\ooalign{\hss\raisebox{0.5ex}{$\cdot$}\hss\cr\phantom{$+$}\cr$-$}}}
  \centering
  \begin{sublisting}{\textwidth}
    \centering
    \begin{align*}
      & {+}~{\Z}~{b}~{\unit} = {\unit}~{\Z}~{b}~{+}{;} && \ruleName{add--base} \\
      & {+}~({\S} {a})~{b}~{\unit} = {\unit}~({\S} {a})~({\S} {b'})~{+}{:} && \ruleName{add--step} \\
      &\qquad  {+}~{a}~{b}~{\unit} = {\unit}~{a}~{b'}~{+}{.} && \ruleName{add--step--sub}
    \end{align*}
    \caption{The definition of reversible natural addition in \textAleph.}
    \label{lst:ex-add-def}
  \end{sublisting}
  \begin{sublisting}{\textwidth}
    \centering
    \input{alephex-add}
    \caption{The \textAleph\ execution path when reversibly adding $4$ to $3$ or, in reverse, subtracting $4$ from $7$. The $\leftrightsquigarrow$~arrows refer to pattern matching/substitution, whilst the solid arrows refer to instantiation/consumption of `sub-terms'.}
    \label{lst:ex-add-43}
  \end{sublisting}
  \begin{sublisting}{\textwidth}
    \centering
    \input{alephex-add-err}
    \caption{The \textAleph\ execution path when attempting to (erroneously) subtract $5$ from $2$. The recursive algorithm identifies that $2-5\equiv 0-3$, but there is no matching definition for this and therefore computation `stalls' on this sub-term. This is usually addressed in the natural numbers by employing the saturating option of `monus', i.e.\ $2\monus5=0$, but it is not reversible.}
    \label{lst:ex-add-52}
  \end{sublisting}
  \caption{The definition of, and example applications of, reversible addition in \textAleph.}
  \label{lst:ex-add}
\endqueuefloat

\queuefloat
  \def\rmsp{\nobreak\hspace{-0.1ex}}
  \caption{An overview of the different classes of (deterministic) computations, both irreversible and reversible.}
  \label{fig:rev-classes}
  \begin{subfigure}{\textwidth}
    \centering
    \includegraphics[width=.6\textwidth]{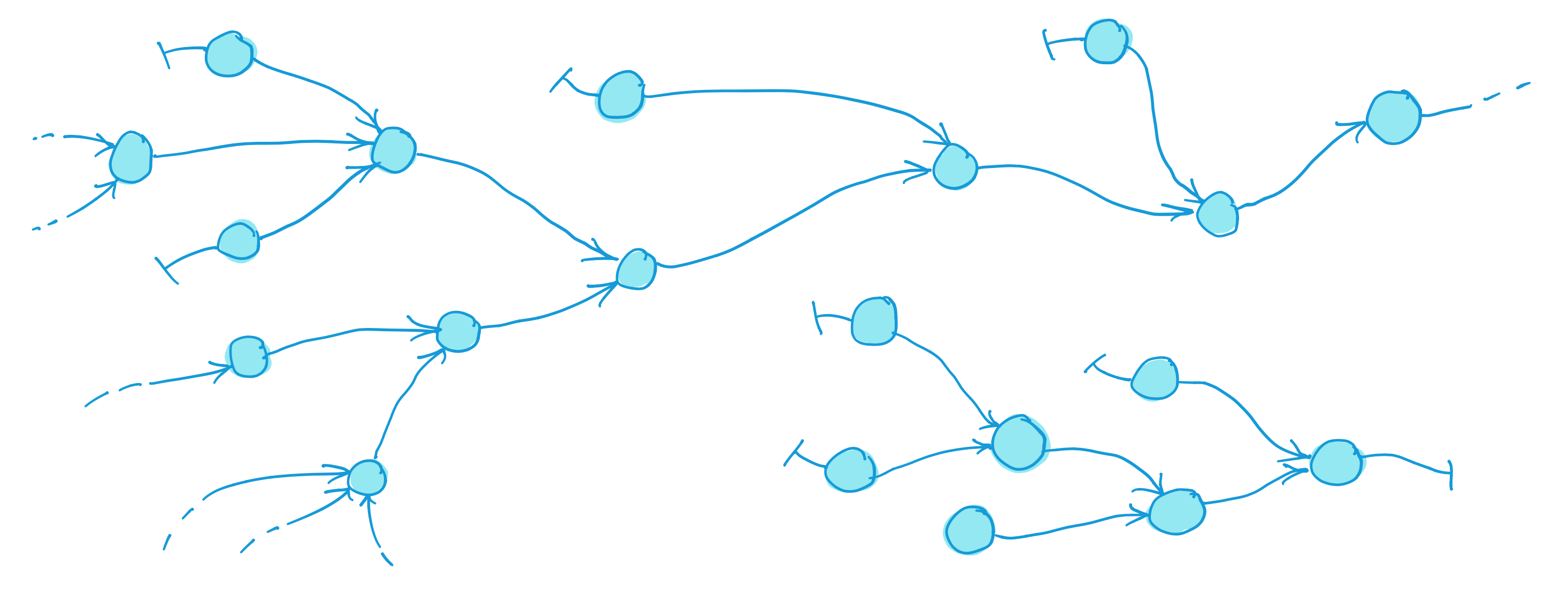}
    \caption{In an irreversible computation, each state (represented by a node) may have multiple predecessors. For example, if two Boolean variables are consumed and replaced with their logical conjunction, and if this is \atom{False}, then there are three possible predecessors: $(\atom{False},\atom{False})$, $(\atom{False},\atom{True})$, and $(\atom{True},\atom{False})$. This is loss of information, and it is impossible to determine which path was taken to reach the current state.}
  \end{subfigure}
  \begin{subfigure}{\textwidth}
    \centering
    \includegraphics[width=.6\textwidth]{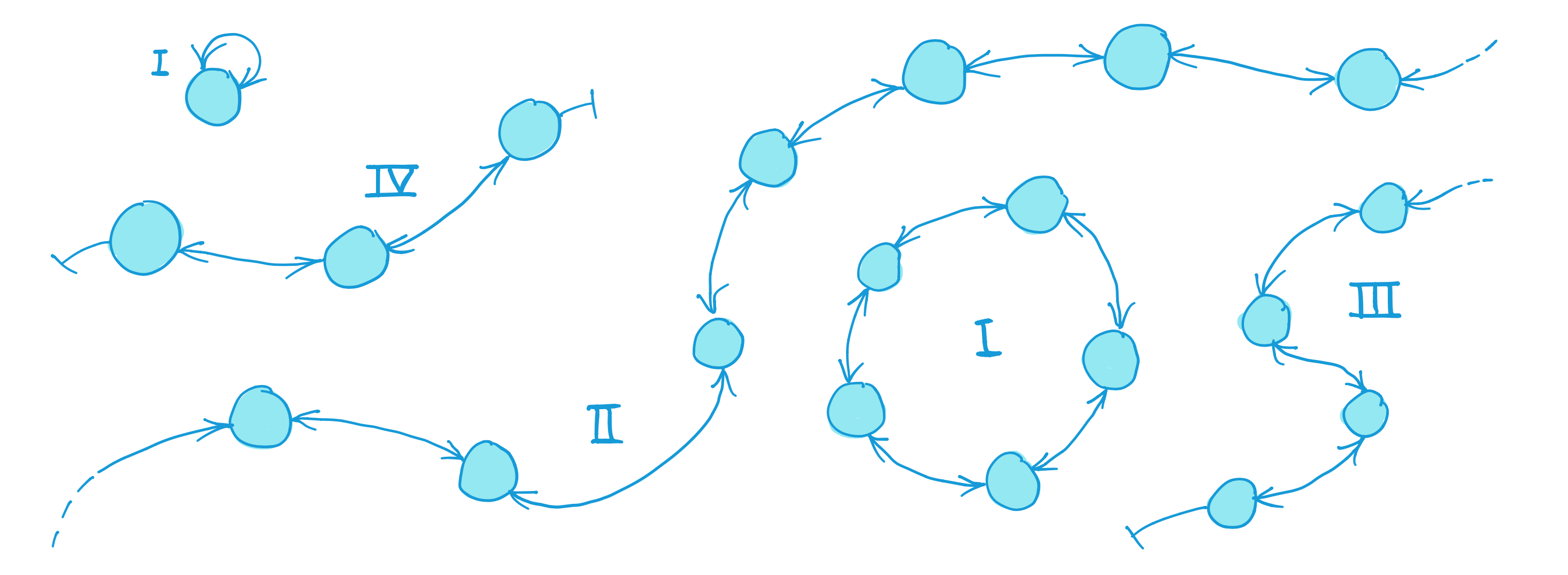}
    \caption{In contrast, a reversible computation can never lose information, and so it is always possible to retrace one's steps. As a result, every computational state has at most one successor and one predecessor. Denoting a halting state by $\vdash$, there are four kinds of reversible computation: (I)~loops without any halting state, (I\rmsp{}I)~bi-infinite computations without any halting state, (I\rmsp{}I\rmsp{}I)~semi-infinite computations which have one halting state (i.e.\ the computation has an initial state but no final state), and (I\rmsp{}V)~finite computations bounded by two halting states. The fourth class corresponds to programs of interest, as they have a well defined initial and final state (and can also be meaningfully reversed); the third corresponds to classical non-halting programs in irreversible models of computation, as does the second as the concept of an initial state is less well-defined in irreversible computation.}
  \end{subfigure}
  %%% stuff goes wrong when subfigures follow the caption :(
  \addtocounter{figure}{-1}
\endqueuefloat

\queuefloat
  \centering
  \includegraphics[width=\textwidth]{aleph-molex-add-43.png}
  \caption{An abstract molecular translation of the reversible addition definition given in \Cref{lst:ex-add-def}, as well as the example addition of $4$ and $3$ following \Cref{lst:ex-add-43}. The abstract molecular model espoused here is defined by a fixed set of atoms (e.g.\ $\{+,\S,\Z,\bullet\}$) connected by two kinds of bond. Atoms joined by single-headed bonds are analogous to \textAleph\ terms, whilst double-headed bonds corresponds to nesting of composite terms. The bonds are represented by arrows because there is an intrinsic polarity/directionality to the molecules; this is not necessary, and can be replaced by auxiliary atoms such as \atom{L} and \atom{R}, but it does simplify our representation. Atoms are rendered by circles, whilst the reaction definitions also use variables written as un-circled letters. The $\bullet$ `atom' is special in that it is a placeholder for a nested composite term. Reaction arrows for `elementary' reactions are labelled by rule names, and starred reaction arrows indicate an effective reaction composed of multiple elementary steps. The arrows are drawn biased in anticipation of a biasing mechanism to be discussed later on within this section.}
  \label{fig:ex-plus}
\endqueuefloat

\queuefloat{listing}
  \centering
  \begin{sublisting}{\textwidth}
    \centering
    \begin{align*}
      & {!}~{\square}~{\blank}~{\unit}{;}\quad {!}~{\unit}~{\blank}~{\square}{;} \\
      & {\square}~{n}~{\unit} = {\square}~{n}~{\Z}~{\square}{;} && \ruleName{square--init} \\
      & {\square}~{\Z}~{m}~{\square} = {\unit}~{m}~{\square}{;} && \ruleName{square--term} \\
      & {\square}~({\S}{n})~{m}~{\square} = {\square}~n~({\S}{m''})~{\square}{:} && \ruleName{square--step} \\
      &\qquad {+}~{n}~{m}~{\unit} = {\unit}~{n}~{m'}~{+}{.} && \ruleName{square--step--sub$_1$}\\
      &\qquad {+}~{n}~{m'}~{\unit} = {\unit}~{n}~{m''}~{+}{.} && \ruleName{square--step--sub$_2$}
    \end{align*}
    \caption{The definition of $\square$ in \textAleph.}
    \label{lst:ex-square-def-aleph}
  \end{sublisting}
  \begin{sublisting}{\textwidth}
    \centering
    \includegraphics[width=.7\textwidth]{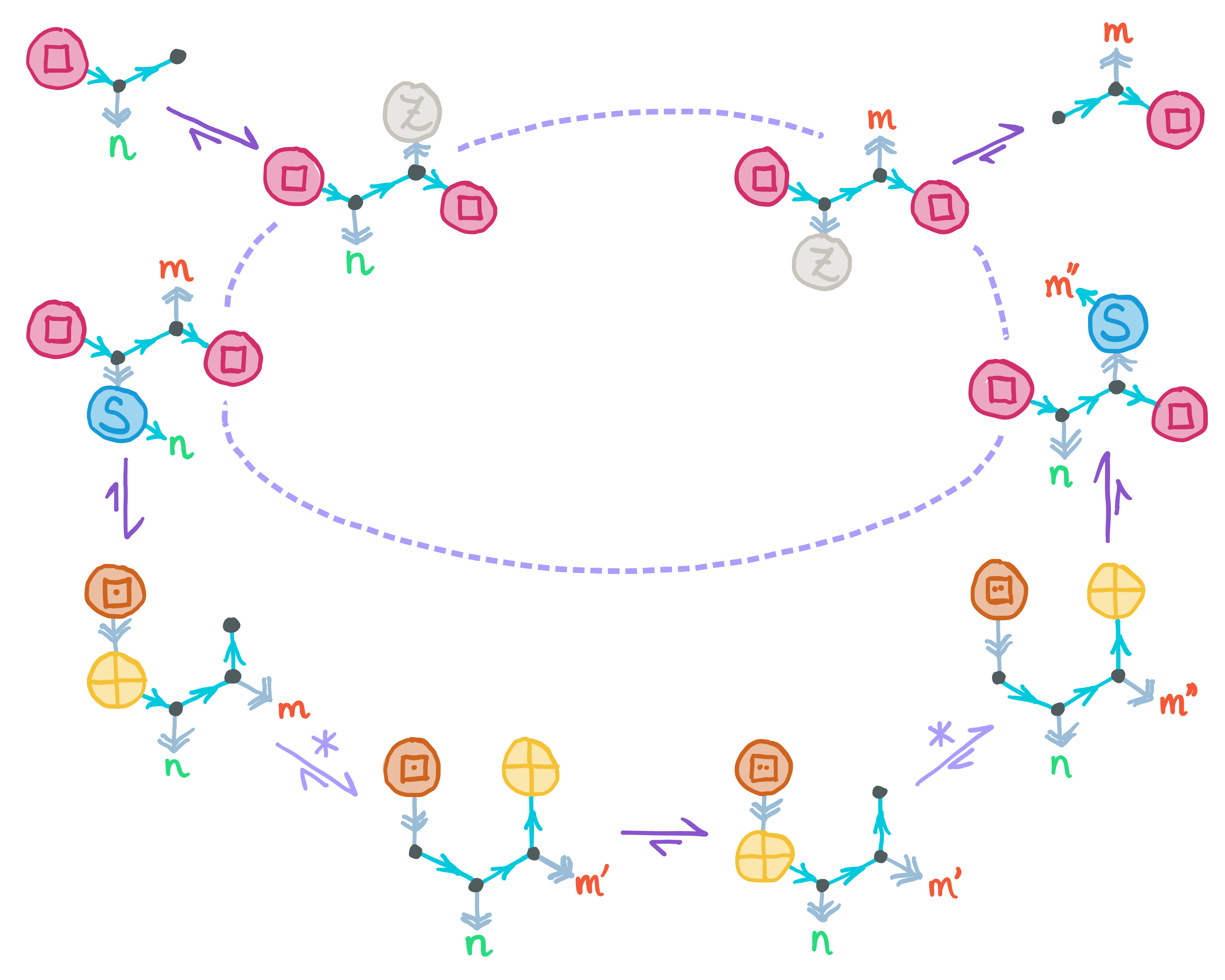}
    \caption{The definition of $\square$ in the abstract molecular formalism introduced in \Cref{fig:ex-plus}. The dashed lines indicate that there exists a term matching both patterns, and this is characteristic of conditionals and looping in \textAleph. Note that computation remains unambiguous and deterministic: by \Cref{fig:rev-classes}, each intermediate state has both a predecessor and a successor, and of the two matching patterns one will correspond to the reverse direction and one to the forward direction. For example, in the middle of a loop one may encounter the intermediate species ${\square}~{2}~{5}~{\square}$ and this term matches both the patterns ${\square}~({\S}{n})~{m}~{\square}$ and ${\square}~n~({\S}{m''})~{\square}$, but if it matches the former the computation will proceed forward whilst if the latter the computation will proceed backward. It will be shown later how the semantics keep track of this, and how we ensure the program really is unambiguous.}
    \label{lst:ex-square-def-mol}
  \end{sublisting}
  \caption{A reversible definition of squaring of natural numbers, $\square:n\leftrightarrow n^2$, both in \textAleph\ and in an abstract molecular formalism.}
  \label{lst:ex-square-def}
\endqueuefloat

\queuefloat{listing}
  \centering
  \input{alephex-square}
  \caption{An example application of the $\square$ definition from \Cref{lst:ex-square-def}.}
  \label{lst:ex-square}
\endqueuefloat

\queuefloat{listing}
  \centering
  \input{alephex-sort}
  \caption{An \textAleph\ implementation of the comparison operators $<$, $\le$, $>$ and $\ge$, and of insertion sort that can make use of these comparison operators. Lastly, the list $[{3}~{2}~{0}~{7}~{6}~{4}~{5}~{1}]$ is sorted into both ascending and descending order.}
  \label{lst:ex-sort}
\endqueuefloat

%% file: alephex-add.tex
% BEWARE, MAGIC NUMBERS LIE WITHIN
\begingroup%
\newcommand{\subLeft}[1][1]{\smash{\raisebox{-0.65em}{\scaleto{\begin{tikzpicture}%
    \draw[<->] (0,0) to [out=270,in=90, looseness=1] (-#1,-0.5);%
  \end{tikzpicture}}{1.78em}}}\hspace{4.3em}}%
\newcommand{\subRight}[1][1]{\hspace{7.3em}\smash{\raisebox{-0.65em}{\scaleto{\begin{tikzpicture}%
    \draw[<->] (#1,-0.5) to [out=90,in=270, looseness=1] (0,0);%
  \end{tikzpicture}}{1.78em}}}}%
\def\bindings#1{\{#1\}}%
\def\bindingsSub#1{\smash{\underline{\bindings{#1}}}}%
\def\subterm#1{\smash{\overline{#1}}}%
\def\midsp{\phantom{\bindings{b:3}}}%
\begin{align*}
  {+}~{4}~{3}~{\unit} \leftrightsquigarrow \bindingsSub{a:3, b:3} &\midsp \bindingsSub{a:3, b':6} \leftrightsquigarrow {\unit}~{4}~{7}~{+} && \ruleName{add--step} \\
  \subLeft &\subRight && \ruleName{add--step--sub} \\
  \subterm{{+}~{3}~{3}~{\unit}} \leftrightsquigarrow \bindingsSub{a:2, b:3} &\midsp \bindingsSub{a:2, b':5} \leftrightsquigarrow \subterm{{\unit}~{3}~{6}~{+}} && \ruleName{add--step} \\
  \subLeft &\subRight && \ruleName{add--step--sub} \\
  \subterm{{+}~{2}~{3}~{\unit}} \leftrightsquigarrow \bindingsSub{a:1, b:3} &\midsp \bindingsSub{a:1, b':4} \leftrightsquigarrow \subterm{{\unit}~{2}~{5}~{+}} && \ruleName{add--step} \\
  \subLeft &\subRight && \ruleName{add--step--sub} \\
  \subterm{{+}~{1}~{3}~{\unit}} \leftrightsquigarrow \bindingsSub{a:0, b:3} &\midsp \bindingsSub{a:0, b':3} \leftrightsquigarrow \subterm{{\unit}~{1}~{4}~{+}} && \ruleName{add--step} \\
  \subLeft[0.25]\hspace{-1.8em} & \hspace{-1.95em}\subRight[0.28] && \ruleName{add--step--sub} \\
  \subterm{{+}~{\Z}~{3}~{\unit}} \leftrightsquigarrow{} &\bindings{b:3} \leftrightsquigarrow \subterm{{\unit}~{\Z}~{3}~{+}} && \ruleName{add--base}
\end{align*}%
\endgroup

%% file: alephex-add-err.tex
% BEWARE, MAGIC NUMBERS LIE WITHIN
\begingroup%
\newcommand{\subLeft}[1][1]{\smash{\raisebox{-0.65em}{\scaleto{\begin{tikzpicture}%
    \draw[<->] (0,0) to [out=270,in=90, looseness=1] (-#1,-0.5);%
  \end{tikzpicture}}{1.78em}}}\hspace{4.3em}}%
\def\bindings#1{\{#1\}}%
\def\bindingsSub#1{\smash{\underline{\bindings{#1}}}}%
\def\subterm#1{\smash{\overline{#1}}}%
\def\nomatch{\mathrel{\reflectbox{$\rightsquigarrow$}\hspace{-0.7ex}/}}%
\def\nomatchlong{\mathrel{\rlap{$\nomatch$}\phantom{\leftrightsquigarrow}}}%
\begin{align*}
  {\unit}~{5}~{2}~{+} \leftrightsquigarrow \bindingsSub{a:4, b':1} & \phantom{\bindings{b:3} \bindingsSub{a:3, b':6} \leftrightsquigarrow {\unit}~{4}~{7}~{+}} && \ruleName{add--step} \\
  \subLeft &&& \ruleName{add--step--sub} \\
  \subterm{{\unit}~{4}~{1}~{+}} \leftrightsquigarrow \bindingsSub{a:3, b':0} &&& \ruleName{add--step} \\
  \subLeft &&& \ruleName{add--step--sub} \\
  \subterm{{\unit}~{3}~{\Z}~{+}} \nomatchlong \phantom{\bindingsSub{a:1, b':3}} &&& \text{\small\emph{no matching rule}}
\end{align*}%
\endgroup

%% file: alephex-square.tex
% BEWARE, MAGIC NUMBERS LIE WITHIN
\begingroup%
\newcommand{\subb}{\smash{\raisebox{-0.65em}{\scaleto{\begin{tikzpicture}%
    \draw[<->] (0,0) to [out=270,in=90, looseness=1.2] (-0.4,-0.22);%
  \end{tikzpicture}}{0.78em}}}}% 1.78
\def\bindings#1{\{#1\}}%
\def\over#1{\smash{\overline{#1}}}%
\def\under#1{\smash{\underline{#1}}}%
\def\overunder#1{\smash{\overline{\underline{#1}}}}%
\def\hl{\\[-1em]}%
\def\nl{\\[1.5em]}%
\def\subl{\hl&\hspace{0.8em}\subb\\}%
\def\subr{\hl&\hspace{4em}\subb\\}%
\begin{align*}
  {!}\quad{\square}~{3}~{\unit} \leftrightsquigarrow{} &\bindings{n:3} \leftrightsquigarrow {\square}~{3}~{\Z}~{\square} \quad\cdots && \ruleName{square--init} \nl
  %%%%%%%%%%%%%%%%%%%%%%%%%%%%%%
  \cdots\quad{\square}~{3}~{\Z}~{\square} \leftrightsquigarrow{} &\under{\bindings{n:2,m:\Z}} && \ruleName{square--step} \subl
  \multicolumn{2}{c}{$\over{{+}~{2}~{\Z}~{\unit}}\leftrightarrow\under{{\unit}~{2}~{2}~{+}}$} && \ruleName{square--step--sub$_1$} \subr
  &\overunder{\bindings{n:2,m':2}} && \text{\small\emph{checkpoint}} \subl
  \multicolumn{2}{c}{$\over{{+}~{2}~{2}~{\unit}}\leftrightarrow\under{{\unit}~{2}~{4}~{+}}$} && \ruleName{square--step--sub$_2$} \subr
  &\over{\bindings{n:2,m'':4}} \leftrightsquigarrow {\square}~{2}~{5}~{\square}\quad\cdots && \ruleName{square--step} \nl
  %%%%%%%%%%%%%%%%%%%%%%%%%%%%%%
  \cdots\quad{\square}~{2}~{5}~{\square} \leftrightsquigarrow{} &\under{\bindings{n:1,m:5}} && \ruleName{square--step} \subl
  \multicolumn{2}{c}{$\over{{+}~{1}~{5}~{\unit}}\leftrightarrow\under{{\unit}~{1}~{6}~{+}}$} && \ruleName{square--step--sub$_1$} \subr
  &\overunder{\bindings{n:1,m':6}} && \text{\small\emph{checkpoint}} \subl
  \multicolumn{2}{c}{$\over{{+}~{1}~{6}~{\unit}}\leftrightarrow\under{{\unit}~{1}~{7}~{+}}$} && \ruleName{square--step--sub$_2$} \subr
  &\over{\bindings{n:1,m'':7}} \leftrightsquigarrow {\square}~{1}~{8}~{\square}\quad\cdots && \ruleName{square--step} \nl
  %%%%%%%%%%%%%%%%%%%%%%%%%%%%%%
  \cdots\quad{\square}~{1}~{8}~{\square} \leftrightsquigarrow{} &\under{\bindings{n:\Z,m:8}} && \ruleName{square--step} \subl
  \multicolumn{2}{c}{$\over{{+}~{\Z}~{8}~{\unit}}\leftrightarrow\under{{\unit}~{\Z}~{8}~{+}}$} && \ruleName{square--step--sub$_1$} \subr
  &\overunder{\bindings{n:\Z,m':8}} && \text{\small\emph{checkpoint}} \subl
  \multicolumn{2}{c}{$\over{{+}~{\Z}~{8}~{\unit}}\leftrightarrow\under{{\unit}~{\Z}~{8}~{+}}$} && \ruleName{square--step--sub$_2$} \subr
  &\over{\bindings{n:\Z,m'':8}} \leftrightsquigarrow {\square}~{\Z}~{9}~{\square}\quad\cdots && \ruleName{square--step} \nl
  %%%%%%%%%%%%%%%%%%%%%%%%%%%%%%
  \cdots\quad {\square}~{\Z}~{9}~{\square} \leftrightsquigarrow{} &\bindings{m:9} \leftrightsquigarrow {\unit}~{9}~{\square}\quad{!} && \ruleName{square--term}
\end{align*}%
\begin{center}\fbox{\begin{minipage}{0.8\textwidth}\vspace{-0.5em}\begin{align*}
  {!}\quad{\square}~{3}~{\unit}\leftrightarrow
  {\square}~{3}~{\Z}~{\square}\leftrightarrow
  {\square}~{2}~{5}~{\square}\leftrightarrow
  {\square}~{1}~{8}~{\square}\leftrightarrow
  {\square}~{\Z}~{9}~{\square}\leftrightarrow
  {\unit}~{9}~{\square}\quad{!}
\end{align*}\vspace{-1.2em}\end{minipage}}\end{center}%
\endgroup

%% file: alephex-sort.tex
\begingroup%
\def\nl{\\[1.5em]}%
  \begin{align*}
    & \atom{False}~\infix{\atom{Not}}~\atom{True}{;} \\
    & \atom{True}~\infix{\atom{Not}}~\atom{False}{;}
  \nl
    & \infix{{<}~{m}~{\Z}}~\atom{False}{;} \\
    & \infix{{<}~{\Z}~{({\S}~{n})}}~\atom{True}{;} \\
    & \infix{{<}~{({\S}~{m})}~{({\S}~{n})}}~{b}{:} \\
    &\qquad \infix{{<}~{m}~{n}}~{b}{.}
  \nl
    & \infix{{\le}~{m}~{n}}~{b'}{:}
      \quad\infix{{<}~{n}~{m}}~{b}{.}
      \quad{b}~\infix{\atom{Not}}~{b'}{.} \\
    & \infix{{>}~{m}~{n}}~\mathrlap{{b}{:}}\phantom{{b'}{:}}
      \quad\infix{{<}~{n}~{m}}~{b}{.} \\
    & \infix{{\ge}~{m}~{n}}~{b'}{:}
      \quad \infix{{<}~{m}~{n}}~{b}{.}
      \quad {b}~\infix{\atom{Not}}~{b'}{.}
  \nl
    & {!}~{\blank}~\infix{\atom{InsertionSort}~{p}}~{\blank}~{\blank}{;} \\
    & {({\atom{InsertionSort}}~{p})}~\cursivexS~{\unit} = \atom{IS}~{p}~\cursivexS~[]~[]~\atom{IS}{;} \\
    & \atom{IS}~{p}~[{x}~{\cdot}~\cursivexS]~\cursivenS~\cursiveyS~\atom{IS} = \atom{IS}~{p}~\cursivexS~[{n}~{\cdot}~\cursivenS]~\cursivezS~\atom{IS}{:} \\
    &\qquad {x}~\cursiveyS~\infix{\atom{Insert}~{p}}~{n}~\cursivezS{.} \\
    & \atom{IS}~{p}~[]~\cursivenS~\cursiveyS~\atom{IS} = {\unit}~\cursivenS~\cursiveyS~{(\atom{InsertionSort}~{p})}{;}
  \nl
    & {x}~[]~\infix{\atom{Insert}~{p}}~{\Z}~[{x}]{;} \\
    & {x}~[{y}~{\cdot}~\cursiveyS]~\infix{\atom{Insert}~{p}}~{n}~[{z}~{z'}~{\cdot}~\cursivezS]{:} \\
    &\qquad \infix{{p}~{x}~{y}}~{b}{.} \\
    &\qquad {x}~[{y}~{\cdot}~\cursiveyS]~{b}~\infix{\atom{Insert'}~{p}}~n~[{z}~{z'}~{\cdot}~\cursivezS]{.}
  \nl
    & {x}~[{y}~{\cdot}~\cursiveyS]~\atom{True}~\infix{\atom{Insert'}~{p}}~{\Z}~[{x}~{y}~{\cdot}~\cursiveyS]{;} \\
    & {x}~[{y}~{\cdot}~\cursiveyS]~\atom{False}~\infix{\atom{Insert'}~{p}}~{({\S}~n)}~[{y}~{\cdot}~\cursivezS]{:} \\
    &\qquad {x}~\cursiveyS~\infix{\atom{Insert}~{p}}~{n}~\cursivezS{.} \\
    & \phantom{{[3~2~0~7~6~4~5~1]}~\infix{\atom{InsertionSort}~{<}}~{[1~4~3~3~3~0~0~0]}~{[0~1~2~3~4~5~6~7]}{.}}
  \end{align*}
\fbox{\begin{minipage}{.8\textwidth}\vspace{-0.2em}\begin{align*}
  & {[3~2~0~7~6~4~5~1]}~\infix{\atom{InsertionSort}~{<}}~{[1~4~3~3~3~0~0~0]}~{[0~1~2~3~4~5~6~7]}{.} \\
  & {[3~2~0~7~6~4~5~1]}~\infix{\atom{InsertionSort}~{\ge}}~\mathrlap{\underbrace{\phantom{[6~2~2~1~0~2~1~0]}}_{\text{garbage}}}{[6~2~2~1~0~2~1~0]}~{[7~6~5~4~3~2~1~0]}{.}
\end{align*}\vspace{-1.2em}\end{minipage}}
\endgroup

%% file: aleph-ex-map.tex
% BEWARE, MAGIC NUMBERS LIE WITHIN
\begingroup%
\def\bindings#1{\{#1\}}%
\def\over#1{\smash{\overline{#1}}}%
\def\under#1{\smash{\underline{#1}}}%
\newcommand{\sub}[2][1]{\smash{\raisebox{-0.65em}{\scaleto{\begin{tikzpicture}%
    \draw[<->] (0,0) to [out=270,in=90, looseness=#1] (#2,-0.5);%
  \end{tikzpicture}}{1.78em}}}}%
\begin{align*}
  {!}\quad{({\atom{Map}}~{\square})}~{[3~5~8]}~{\unit} \leftrightsquigarrow \bindings{ \under{f:\square, x:3\vphantom{[]}} &, \under{f':\square, \cursivexS:{[5~8]}} } \\
  \sub{-1}\hspace{2.2em} & \hspace{3em}\sub{0.5} \\
  \over{{\square}~{3}~{\unit}}=\under{{\unit}~{9}~{\square}} \quad&\quad
  \over{{(\atom{Map}~{\square})}~{[5~8]}~{\unit}}=\under{{\unit}~{[25~64]}~{(\atom{Map}~{\square})}} \\
  \sub{0.2}\hspace{1.5em} & \hspace{4.8em}\sub[0.5]{-2.3} \\
  \bindings{ \over{f:\square, y:9\vphantom{[]}} &, \over{f':\square, \cursiveyS:{[25~64]}} } \leftrightsquigarrow {\unit}~{[9~25~64]}~{({\atom{Map}}~{\square})}\quad{!}
\end{align*}%
\endgroup

%% file: part-aleph.tex
\section{The Calculus}

The \textAleph\ calculus thus introduced has a very simple definition. In BNF notation, it is
\begin{equation*}\begin{aligned}
  \ruleName{pattern term} && \tau &::= \bnfPrim{atom} ~|~ \bnfPrim{var} ~|~ (\,\tau^\ast\,) \\
  \ruleName{party} && \pi &::= \tau:\tau^\ast ~|~ \bnfPrim{var}':\tau^\ast \\
  \ruleName{definition} && \delta &::= \{\pi^\ast\}=\{\pi^\ast\}:\pi\rlap.^\ast ~|~ {!}~{\tau}
\end{aligned}\end{equation*}
where \bnfPrim{atom} is an infinite set of atomic symbols (conventionally starting with an uppercase letter or a symbol), \bnfPrim{var} an infinite set of variables (conventionally starting with a lowercase letter), and $\bnfPrim{var}'$ is an orthogonal infinite set of variables (conventionally rendered in Greek) used for opaquely matching one-hole-contexts. A program is a series of definitions, $\delta^\ast$, and a physical term is simply a pattern term without variables. Notice that the form of the sub-rules differs from the examples: a sub-rule can be separated into the instantiation of a term according to an input pattern, the evolution of that term, followed finally by the consumption of its final halting state according to the output pattern. In this view, these sub-terms are identified by one-hole-contexts---specifically, opaque variables. This formulation not only allows the semantics to represent automatic parallelisation and a non-deterministic sub-rule ordering, but also begets an additional feature whereby sub-rules can instantiate top-level terms. The sub-rule forms in the examples can then be seen as sugar, i.e.\ $s=t.$ is equivalent to $\lambda:s.~\lambda:t.$ where $\lambda\in\bnfPrim{var}'$ is a fresh opaque variable. To illustrate, recall the definition of natural addition from \Cref{lst:ex-add-def}; desugared, this takes the form
  {\def\ruleLabel#1{\llap{\tiny$#1.$}~}\begin{align*}
    \ruleLabel{1}& {!}~{+}~{a}~{b}~{\unit}{;} \\
    \ruleLabel{2}& {!}~{\unit}~{a}~{b}~{+}{;} \\
    \ruleLabel{3}& \{\alpha:{+}~{\Z}~{b}~{\unit}\} = \{\alpha:{\unit}~{\Z}~{b}~{+}\}{;} && \ruleName{add--base} \\
    \ruleLabel{4}& \{\alpha:{+}~({\S} {a})~{b}~{\unit}\} = \{\alpha:{\unit}~({\S} {a})~({\S} {b'})~{+}\}{:} && \ruleName{add--step} \\
    \raisebox{0pt}[1.9em]{}%
      {\begin{aligned}\ruleLabel{4a}&\\\ruleLabel{4b}&\end{aligned}}&%
      {\begin{aligned}
        &\qquad  \beta: {+}~{a}~{b}~{\unit}{.} \\
        &\qquad  \beta: {\unit}~{a}~{b'}~{+}{.}
      \end{aligned}} && \ruleName{add--step--sub}
  \end{align*}}%
That is, there are four definitions: two `halting', and two `computational'. On the left side of definition 4, we have a bag of one party. This party has as context pattern the opaque variable $\alpha$, and its pattern term is a composite pattern term of length 4, consisting of the atom $+$, the composite pattern term of length 2 consisting of the atom $\S$ and the variable $a$, the variable $b$, and the empty composite pattern term (`unit'). There is also an opaque variable $\beta$, confined to the inner scope of definition 4. During execution of the sub-rule in the forward direction, the variables $a$ and $b$ will first be consumed in order to generate the sub-term ${+}~{a}~{b}~{\unit}$, which will be bound to the opaque variable $\beta$. This sub-term will then evolve to its other halting state, whereupon it will match sub-rule $4b$ and thus $\beta$ will be consumed and the variables $a$ and $b'$ obtained.

\para{Semantics}

To formalise the semantics embodied in the preceding examples, we define a transition relation $\leftrightarrow$ that maps a bag of terms to another bag of terms according to the rules defined for the current program. By a bag of terms, we aim to evoke the notion of a concoction of computational molecules in solution; a bag is a multiset, which can contain multiple copies of elements, but unlike a physical solution it lacks a concept of space. If desired, spatial locations can be readily simulated by appropriately chosen top-level contexts. Recalling that the global dynamics of the system are irreversible, we should expect that $\leftrightarrow$ is a non-deterministic relation. That is, application of the relation to a set of possible bags of terms may increase the size of this set, and the logarithm of the set size is identified with the entropy of the system. Concretely, the set of bags is the macrostate of the system, and each bag within the set is a microstate.

The calculus being reversible, the relation should share properties with equivalence relations. Namely, if $S,T,U$ are bags of terms, then we have
\begin{equation*}\begin{aligned}
  && S &\leftrightarrow S && \ruleName{refl} \\
  S \leftrightarrow T \implies&& T &\leftrightarrow S && \ruleName{symm} \\
  S \leftrightarrow T \land T \leftrightarrow U \implies&& S &\leftrightarrow U && \ruleName{trans}
\end{aligned}\end{equation*} 
but we also make clear that these transitions can occur within an environment of other non-participating terms, $\Gamma$,
\begin{equation*}\begin{aligned}
  S \leftrightarrow T \implies&& \forall\Gamma.~ \Gamma \cup S &\leftrightarrow \Gamma \cup T && \ruleName{ext} \\
  s \leftrightarrow t \implies&& \forall\Gamma.~ \Gamma \cup \{s\} &\leftrightarrow \Gamma \cup \{t\} && \ruleName{ext$'$}
\end{aligned}\end{equation*}
where the second rule, in which $s$ and $t$ are single terms rather than bags of terms, is introduced as a convenient abuse of notation for later definitions.

In order to physically effect rule transitions, we introduce `mediator terms' delimited by angle brackets, which interact with computational terms and can represent each of the intermediate states. If the program is given by $\mathcal P$, then the mechanism by which these mediator terms come into and out of existence is as follows,
\begin{equation*}\begin{aligned}
  \mathcal P\vdash (I=O:R) \implies&& \{\}&\leftrightarrow\{\langle I\varnothing\varnothing R\varnothing O\rangle\} && \ruleName{inst--comp} \\
  \mathcal P\vdash \rlap{$({!}~{\tau})$}\phantom{(I=O:R)} \implies&& \{\}&\leftrightarrow\{\langle \tau\rangle\} && \ruleName{inst--halt}
\end{aligned}\end{equation*}
from which we see that the environment can contain an arbitrary number of copies of each\footnote{This complicates the aforementioned entropic interpretation of sets of term-bags, as the entropy will tend to diverge as the number of mediator terms tends to infinity. A more realistic implementation would leave the mean number of extant mediator terms finite and bounded, and perhaps even fix the number. This would have a further consequence on the kinetics and thermodynamics of the system, with the well-characterised Mich{\ae}lis-Menten kinetics a good starting point to the analysis thereof.}. The mediator terms for computational rules are sextuples $\langle II'BR\Gamma O\rangle$ consisting of, respectively, the bag of unmatched input patterns $I$, the bag of matched input patterns $I'$, the bag of resultant bindings $B$, the bag of sub-rules $R$, the local environment for internal sub-rules $\Gamma$, and the bag of output patterns $O$. The mediator terms for halting rules are trivial singletons $\langle\tau\rangle$ containing the relevant pattern, $\tau$.

To assist rules in binding composite terms, we permit the current focus of a term to vary over time,
\begin{equation*}\begin{aligned}
  c:(\vec l~{\ooalign{\hss$t$\hss\cr\phantom{$\bullet$}}}~\vec r) &\leftrightarrow (c:(\vec l~{\bullet}~\vec r)):t && \ruleName{ohc$_1$} \\
  c:[\vec l~{\ooalign{\hss$t$\hss\cr\phantom{$\bullet$}}}~\vec r] &\leftrightarrow (c:[\vec l~{\bullet}~\vec r]):t && \ruleName{ohc$_2$}
\end{aligned}\end{equation*}
where $\vec l$ and $\vec r$ are (possibly empty) strings of terms and $t$ is the term focus. The bracketed terms in the second rule will be explained shortly. Note that these one-hole-context rules may not be needed for all architectures, being implicitly true in a molecular architecture for example.

The action of halting mediators is simply to mark terms which are in a known halting state,
\begin{equation*}\begin{aligned}
  \exists b'.~t \overset{\tau}{\sim}b' \implies&& \{x:t,\langle\tau\rangle\} &\leftrightarrow \{x:[t],\langle\tau\rangle\} && \ruleName{term}
\end{aligned}\end{equation*}
with $[t]$ serving as an indicator of a halting state and where $t\overset\tau\sim b'$ means that $t$ unifies against $\tau$ with bindings $b'$ (see \Cref{lst:aleph-sem-unif}).

Computational mediators are somewhat more complicated. We first render their temporal symmetry manifest by two reversibility rules,
\begin{equation*}\begin{aligned}
  \langle I\varnothing\varnothing R\varnothing O\rangle &\leftrightarrow \langle O\varnothing\varnothing R\varnothing I\rangle && \ruleName{rev$_1$} \\
  \langle\varnothing IBR\Gamma O\rangle &\leftrightarrow \langle\varnothing OBR\Gamma I\rangle && \ruleName{rev$_2$}
\end{aligned}\end{equation*}
where these rules will be seen to be necessary even for computation in a single direction. Applying a computational rule consists first of binding against a matching term for each input pattern, followed by substituting the bindings into sub-terms per the sub-rules. The sub-terms may either be top-level or local. These transitions may all occur in parallel, i.e.\ a sub-term may be instantiated if all of its requisite variables are bound, even if not all the input patterns have matched a term.
{\def\ruleHeader#1#2{\llap{$\displaystyle #1 \hspace{5em}$}\rlap{\hspace{9.5em}\ruleName{#2}}}
\begin{equation*}\begin{aligned}
  &\ruleHeader{x:t\overset\pi\sim b \implies}{inp} \\
  \{x:t, \langle(I\cup\{\pi\})I'BR\Gamma O\rangle &\leftrightarrow
  \{\langle I(I'\cup\{\pi\})(B\cup b)R\Gamma O\rangle\} \\
  &\ruleHeader{\pi\in R \land x:t\overset\pi\sim b \implies}{sub$_1$} \\
  \langle II'(B\cup b)R\Gamma O\rangle &\leftrightarrow
  \langle II'BR(\Gamma\cup\{x:[t]\})O\rangle \\
  &\ruleHeader{x \notin \bnfPrim{var}' \implies}{sub$_2$} \\
  \{\langle II'BR(\Gamma\cup\{x:[t]\})O\rangle\} &\leftrightarrow
  \{x:[t],\langle II' BS\Gamma O\rangle\}
\end{aligned}\end{equation*}}
where $x$ is a context. The \ruleName{sub$_2$} transition enables top-level terms to escape the locally scoped environment. Completion of computation occurs by application of \ruleName{rev$_2$} followed by the \ruleName{inp} and \ruleName{sub$_{1,2}$} in reverse; clearly if there is no route from the set of input variables to the set of output variables then the computation will stall. It may be desirable to augment the transition rules with implicit variable duplication,
\begin{equation*}\begin{aligned}
  \langle II'(B\cup\{b\})R\Gamma O\rangle &\leftrightarrow\langle II'(B\cup\{b,b\})R\Gamma O\rangle && \ruleName{dup}
\end{aligned}\end{equation*}
otherwise rules which wish to increase their parallelisability should explicitly duplicate variables as needed. We shall also need to enable the sub-environment to evolve,
\begin{equation*}\begin{aligned}
  \Gamma\leftrightarrow\Gamma' \implies && \langle II'BR\Gamma O\rangle &\leftrightarrow\langle II'BR\Gamma'O\rangle && \ruleName{sub$_3$}
\end{aligned}\end{equation*}

These semantics, encapsulated by the $\leftrightarrow$ transition rule, are summarised in \Cref{lst:aleph-sem}. An example of their application is provided in \Cref{lst:semex-add}. It remains to describe the operation of binding/unification. This operation is defined as one would expect: a pattern matches if it is a variable, if it is an atom and the term is the same atom, or if it is a composite pattern and the term is a \emph{halting} composite term of the same length and if the pattern and term match pair-wise. This is summarised in \Cref{lst:aleph-sem-unif}.

\begin{listing}\centering
\fbox{\begin{minipage}{0.9\textwidth}\vspace{\baselineskip}\begin{equation*}\begin{aligned}
    && S &\leftrightarrow S && \ruleName{refl} \\
    S \leftrightarrow T \implies&& T &\leftrightarrow S && \ruleName{symm} \\
    S \leftrightarrow T \land T \leftrightarrow U \implies&& S &\leftrightarrow U && \ruleName{trans} \\
    S \leftrightarrow T \implies&& \forall\Gamma.~ \Gamma \cup S &\leftrightarrow \Gamma \cup T && \ruleName{ext} \\
    s \leftrightarrow t \implies&& \forall\Gamma.~ \Gamma \cup \{s\} &\leftrightarrow \Gamma \cup \{t\} && \ruleName{ext$'$} \\
    \mathcal P\vdash (I=O:R) \implies&& \{\}&\leftrightarrow\{\langle I\varnothing\varnothing R\varnothing O\rangle\} && \ruleName{inst--comp} \\
    \mathcal P\vdash \rlap{$({!}~{\tau})$}\phantom{(I=O:R)} \implies&& \{\}&\leftrightarrow\{\langle \tau\rangle\} && \ruleName{inst--halt} \\
    && c:(\vec l~{\ooalign{\hss$t$\hss\cr\phantom{$\bullet$}}}~\vec r) &\leftrightarrow (c:(\vec l~{\bullet}~\vec r)):t && \ruleName{ohc$_1$} \\
    && c:[\vec l~{\ooalign{\hss$t$\hss\cr\phantom{$\bullet$}}}~\vec r] &\leftrightarrow (c:[\vec l~{\bullet}~\vec r]):t && \ruleName{ohc$_2$} \\
    \exists b'.~t \overset{\tau}{\sim}b' \implies&& \{x:t,\langle\tau\rangle\} &\leftrightarrow \{x:[t],\langle\tau\rangle\} && \ruleName{term} \\
    && \langle I\varnothing\varnothing R\varnothing O\rangle &\leftrightarrow \langle O\varnothing\varnothing R\varnothing I\rangle && \ruleName{rev$_1$} \\
    && \langle\varnothing IBR\Gamma O\rangle &\leftrightarrow \langle\varnothing OBR\Gamma I\rangle && \ruleName{rev$_2$} \\
    x:t\overset\pi\sim b \implies &&&&& \ruleName{inp} \\%&&
    &&& \omit{\llap{$\displaystyle\{x:t, \langle(I\cup\{\pi\})I'BR\Gamma O\rangle\}$}\rlap{$\displaystyle{}\leftrightarrow\{\langle I(I'\cup\{\pi\})(B\cup b)R\Gamma O\rangle\}$}} \\
    \pi\in R \land x:t\overset\pi\sim b \implies &&&&& \ruleName{sub$_1$} \\
    &&& \omit{\llap{$\displaystyle\langle II'(B\cup b)R\Gamma O\rangle$}\rlap{$\displaystyle{}\leftrightarrow\langle II'BR(\Gamma\cup\{x:[t]\})O\rangle$}} \\
    x \notin \bnfPrim{var}' \implies &&&&& \ruleName{sub$_2$} \\
    &&& \omit{\llap{$\displaystyle\{\langle II'BR(\Gamma\cup\{x:[t]\})O\rangle$}\rlap{$\displaystyle{}\leftrightarrow\{x:[t],\langle II' BR\Gamma O\rangle\}$}} \\
    \Gamma\leftrightarrow\Gamma' \implies && \langle II'BR\Gamma O\rangle &\leftrightarrow\langle II'BR\Gamma'O\rangle && \ruleName{sub$_3$}
\end{aligned}\end{equation*}\vspace{\baselineskip}\end{minipage}}
  \caption{Summary of \textAleph\ semantics.}
  \label{lst:aleph-sem}
\end{listing}

\begin{listing}\centering
\fbox{\begin{minipage}{0.9\textwidth}\vspace{\baselineskip}\begin{equation*}\begin{array}{clcl}
    \cfrac{\alpha\in\bnfPrim{atom}}{\alpha \overset{\alpha}{\sim} \varnothing} & \ruleName{unif--atom} &
    \cfrac{v\in\bnfPrim{var}}{t \overset{v}{\sim} \{v\mapsto t\}} & \ruleName{unif--var} \\
    \cfrac{\bigwedge_i t_i \overset{\pi_i}{\sim} b_i}{[t_1\cdots t_n] \mathrel{\overset{\pi_1\cdots\pi_n}{\scalebox{1.75}[1]{$\sim$}}} \bigcup_i b_i} & \ruleName{unif--sub} &
    \cfrac{\lambda\in\bnfPrim{var}'\quad t\overset{\pi}{\sim}b}{\gamma:t \overset{\lambda:\pi}{\sim} \{\lambda\mapsto\gamma\}\cup b} & \ruleName{unif--ctxt$_1$} \\
    \cfrac{\gamma\overset{\pi_1}{\sim}b_1\quad t\overset{\pi_2}{\sim}b_2}{\gamma:t \mathrel{\overset{\pi_1:\pi_2}{\scalebox{1.5}[1]{$\sim$}}} b_1\cup b_2} & \ruleName{unif--ctxt$_2$}
\end{array}\end{equation*}\vspace{\baselineskip}\end{minipage}}
  \caption{Summary of unification semantics for \textAleph.}
  \label{lst:aleph-sem-unif}
\end{listing}

\begin{listing}
  \centering
  \input{aleph-semex-add}
  \caption{}
  \label{lst:semex-add}
\end{listing}

\para{Reversibility}

The \ruleName{symm} transition renders the semantics trivially reversible, but this is fairly weak. We conclude the discussion of the semantics of \textAleph\ by proving that the semantics are \emph{microscopically} reversible.

\begin{dfn}[Microscopic Reversibility]A transition is primitively microscopically reversible if it is a uniquely invertible structural rearrangement. A transition is microscopically reversible if it is decomposable into a series of primitively microscopically reversible transitions.\end{dfn}

\begin{thm}\label{thm:mrev}The semantics of\kern0.1ex\ \textslAleph\ are microscopically reversible.\end{thm}
\begin{proof}
  The rules \ruleName{refl,symm,trans,ext,ext$'$,term,rev$_1$,rev$_2$,sub$_3$} are microscopically reversible either trivially or inductively.

  The instantiation rules \ruleName{inst--comp,inst--halt} are less obvious, but can be realised in a microscopically reversible fashion in much the same way that cells translate an mRNA template to a polypeptide product. We shall first need a microscopically reversible way to duplicate a term, for which we assume that the environment contains an excess of structural building blocks (i.e.\ free atoms, variables, units $\unit$, etc). A term $t$ to be duplicated is first marked as such i.e. $t \mapsto \overline{t}$. These modified terms deviate from the definitions introduced at the beginning of this section, and are instead intermediate transitional structures used for the small-step microscopically reversible semantics described here. This marking is then propagated throughout the structure to all atoms, variables, and units by the following two microscopically reversible transitions:
  \begin{align*}
    \overline{\gamma:t} &\leftrightarrow \overline{\gamma}:\overline{t} && \ruleName{dup--prop$_1$} \\
    \overline{({x}~\vec\cursivexS)} &\leftrightarrow (\hat{x}~\vec\cursivexS) && \ruleName{dup--prop$_2$} \\
    (\vec\cursivexS~\hat{x}~{y}~\vec\cursiveyS) &\leftrightarrow (\vec\cursivexS~\overline{x}~\hat{y}~\vec\cursiveyS) && \ruleName{dup--prop$_3$}
  \end{align*}
  where rules \ruleName{dup--prop$_{2,3}$} are used to distribute the marking throughout composite terms. Note the use of the auxiliary `hat' marker $\hat{\,\cdot\,}$ to sequence this propagation in a microscopically reversible manner.
  Elementary terms thus marked recruit fresh copies of themselves from the free structural building blocks,
  \begin{align*}
    \frac{}{a} &\leftrightarrow \frac{a}{a} &
    \frac{}{u} &\leftrightarrow \frac{u}{u} &
    \frac{}{v} &\leftrightarrow \frac{v}{v} &
    \frac{}{\unit} &\leftrightarrow \frac{\unit}{\unit} &
    &\ruleName{dup--fresh}
  \end{align*}
  where $a\in\bnfPrim{atom}$, $u\in\bnfPrim{var}$ and $v\in\bnfPrim{var}'$, and where the drawing of a fresh copy from the environment is implicit. These are then ligated together in parallel, and the composite structure disentangled,
  \begin{align*}
    \left(\vec{s}~\frac{\vec t}{\vec t}~\frac{\vec u}{\vec u}~\vec{v}\right) &\leftrightarrow
    \left(\vec{s}~\frac{\vec t~\vec u}{\vec t~\vec u}~\vec{v}\right) & 
    &\ruleName{dup--lig$_1$} &
    \left(\frac{\vec t}{\vec t}\right) &\leftrightarrow\frac{(\vec t)}{(\vec t)} &
    &\ruleName{dup--topo$_1$} \\
    \left\{\vec{s}~\frac{\vec t}{\vec t}~\frac{\vec u}{\vec u}~\vec{v}\right\} &\leftrightarrow
    \left\{\vec{s}~\frac{\vec t~\vec u}{\vec t~\vec u}~\vec{v}\right\} & 
    &\ruleName{dup--lig$_2$} & 
    \left\{\frac{\vec\pi}{\vec\pi}\right\} &\leftrightarrow\frac{\{\vec\pi\}}{\{\vec\pi\}} &
    &\ruleName{dup--topo$_2$} \\
    &&&& \frac{\gamma}{\gamma} : \frac{t}{t} &\leftrightarrow \frac{\gamma:t}{\gamma:t} &
    &\ruleName{dup--topo$_3$}
  \end{align*}
  finally resulting in a fully duplicated and disentangled structure $\frac tt$ from $\overline{t}$ for $t$ any term, party, or party-bag. Finally, the instantiation rules can be realised microscopically reversible thus,
  \begin{align*}
    (I=O:S) &\leftrightarrow (\overline{I}=\overline{O}:\overline{S}) &&\ruleName{inst--comp$_1$} \\
    \left\{\left(\frac II=\frac OO:\frac SS\right)\right\} &\leftrightarrow \{(I=O:S), \langle I\varnothing\varnothing S\varnothing O\rangle\} &&\ruleName{inst--comp$_2$} \\
    ({!}~{\tau}) &\leftrightarrow ({!}~\overline{\tau}) &&\ruleName{inst--halt$_1$} \\
    \left\{\left({!}~\frac \tau\tau\right)\right\} &\leftrightarrow \{({!}~{\tau}), \langle\tau\rangle\} &&\ruleName{inst--halt$_2$}
  \end{align*}

  The one-hole-context rules are nearly trivially microscopically reversible, except for the choice of partition. This can be achieved by a movable marker like so,
  {\def\mark#1{\underset{\hat{}}{#1}}
  \begin{align*}
    c:({x}~\vec\cursivexS) &\leftrightarrow c:\llparenthesis\mark{x}~\vec\cursivexS\rrparenthesis &&\ruleName{ohc$_{11}$} &
    c:[{x}~\vec\cursivexS] &\leftrightarrow c:\llbracket\mark{x}~\vec\cursivexS\rrbracket &&\ruleName{ohc$_{21}$} \\
    c:\llparenthesis\vec\cursivexS~\mark{x}~{y}~\vec\cursiveyS\rrparenthesis &\leftrightarrow c:\llparenthesis\vec\cursivexS~{x}~\mark{y}~\vec\cursiveyS\rrparenthesis &&\ruleName{ohc$_{12}$} &
    c:\llbracket\vec\cursivexS~\mark{x}~{y}~\vec\cursiveyS\rrbracket &\leftrightarrow c:\llbracket\vec\cursivexS~{x}~\mark{y}~\vec\cursiveyS\rrbracket &&\ruleName{ohc$_{22}$} \\
    c:\llparenthesis\vec{\ell}~\mark{t}~\vec{r}\rrparenthesis &\leftrightarrow (c:(\vec{\ell}~{\bullet}~\vec{r})):t &&\ruleName{ohc$_{13}$} &
    c:\llbracket\vec{\ell}~\mark{t}~\vec{r}\rrbracket &\leftrightarrow (c:[\vec{\ell}~{\bullet}~\vec{r}]):t &&\ruleName{ohc$_{23}$}
  \end{align*}}
  
  The rules \ruleName{inp,sub$_1$,sub$_2$} require a microscopically reversible realisation of a bag with random draws. In order to be microscopically reversible, the draw needs to be deterministic at any given time, even if draws at different times are uncorrelated. This is achieved by representing a bag as a string that can be shuffled via swap operations, i.e.
  \begin{align*}
    \{\vec\cursivexS~{x}~{y}~\vec\cursiveyS\} &\leftrightarrow \{\vec\cursivexS~{y}~{x}~\vec\cursiveyS\} && \ruleName{bag--shuff}
  \end{align*}
  Random draws are then implemented by simply picking the first element of the string.
  
  The last rules to demonstrate microscopic reversibility for are the unification semantics. We leave this as an exercise for the reader, with the hint that its realisation is similar to that of the instantiation rules.
\end{proof}

\para{Computability}

The earlier examples give reasonable assurance that \textAleph\ is Turing complete and that programming in it is `easy' in the sense that programs can be readily composed. For avoidance of doubt, however, we prove that \textAleph\ is Turing complete in two ways: the first proves a stronger claim, that \textAleph\ is Reversible-Turing complete, meaning that it can efficiently and faithfully simulate a Reversible Turing Machine without generating garbage; the second proves Turing completeness in a more conventional fashion in order to demonstrate its high level of composability.

\begin{listing}
  \centering
  \input{aleph-tc-tape}
  \caption{This \textAleph\ program defines all the ingredients necessary to simulate any Reversible Turing Machine. We represent a bi-infinite tape as its one-hole-context; that is, a tape is given by the square in the current position, $x$, as well as `all' the squares to the left of it, $\ell$, and `all' the squares to the right of it, $r$. Obviously we can't actually represent \emph{all} the squares to the left and the right. Instead we make use of the fact that, at any one time, only a finite bounded region of the tape is non-blank. As such, $\ell$ contains all the squares to the left up to the last symbol, and similarly for $r$. If we keep going left, past the final symbol, then $\ell$ will be the empty list, and \atom{Blank} squares will be created as needed. The rules \atom{Left} and \atom{Right} are convenience functions for changing the focus of a tape.}
  \label{lst:aleph-tape}
\end{listing}
\begin{listing}
  \centering
  \input{aleph-tc-mu}
  \caption{An \textAleph\ program$^\ast$ implementing the $\mu$-recursive functions. The $\mu$-recursive functions are generated by the three functions, $C_n$, $S$ and $P_i$, and the three operators, $\circ$, $\rho$ and $\mu$, whose definitions are given in the comments above. The \textAleph\ values corresponding to these functions can be directly composed using the operators (e.g.\ $(\atom{Rec}~(\atom{Const}~0)~(\atom{Proj}~1))$), and then evaluated with the \atom{Mu} atom.}
  ~\\\parbox{\textwidth}{\hangindent=0.3cm $^\ast$\scriptsize Some additional sugar has been used in this program, mainly in the form of nested definitions and the more concise expression of looping constructs. See \Cref{sec:alethe} for an in-depth definition of these sugared forms.}
  \label{lst:aleph-murec}
\end{listing}

\begin{thm}The \textslAleph\ calculus is Reversible-Turing Complete, i.e.\ it can simulate any Reversible Turing Machine (RTM) as defined by \textcite{bennett-tm}.\end{thm}
\begin{proof}
  An RTM is a collection of one or more bi-infinite tapes divided into squares, each of which can be blank or can contain a symbol, as well as a machine head with an internal state. Symbols are drawn from a finite alphabet, and internal states from a distinct finite alphabet. At any time, the RTM head looks at a single square on each tape and executes one of a finite set of reversible rules. The rule chosen is uniquely determined by the current state of the system. Each rule depends on the current internal state of the machine head, which it may alter, and for each tape it can either ignore its value, possibly moving the tape one square to the left or right, or it can depend on the square taking a certain value $u$, which it can replace with a certain other value $v$.

  For example, a 6-tape RTM with symbol alphabet $\{A,B,C,D,E,F\}$ and state alphabet $\{\atom{Start},S_1,S_2,\atom{Stop}\}$ might have a rule $S_1~C~\varnothing~/~D~/~/~\rightarrow~B~A~-~F~0~+~S_2$. This rule applies if and only if the current internal state is $S_1$, the values of tapes 1 and 4 are the symbols $C$ and $D$ respectively, and tape 2 is blank. If so, it will change the internal state to $S_2$, replace the values of tapes 1, 2 and 4 with the symbols $B$, $A$ and $F$ respectively, and move tapes 3 and 6 one square to the left and right respectively, leaving tape 5 alone. The reverse of the rule is given by $S_2~B~A~/~F~/~/~\rightarrow~C~\varnothing~+~D~0~-~S_1$. The necessary conditions for the ruleset to be deterministic and reversible are that all the domains of the forward rules are mutually exclusive, as are all the domains of the reversed rules. Of course, the domains of the forward and reverse rules will intersect in any useful program.

  It is easy to translate any description of an RTM into \textAleph. Bi-infinite tapes can be easily represented and manipulated, per the program in \Cref{lst:aleph-tape}, and any rule (such as the above) can be mechanically translated as so,
  \begin{align*}
    &\phantom{{}={}} {S_1}~%
        (\atom{Tape}~{\ell_1}~(\atom{Sym}~C)~{r_1})~%
        (\atom{Tape}~{\ell_2}~\atom{Blank}~{r_2})~{t_3}~%
        (\atom{Tape}~{\ell_4}~(\atom{Sym}~D)~{r_4})~{t_5}~{t_6} \\
    &= {S_2}~%
        (\atom{Tape}~{\ell_1}~(\atom{Sym}~B)~{r_1})~%
        (\atom{Tape}~{\ell_2}~(\atom{Sym}~A)~{r_2})~{t_3'}~%
        (\atom{Tape}~{\ell_4}~(\atom{Sym}~F)~{r_4})~{t_5}~{t_6'}{:} \\
    &\phantom{{}={}}\qquad {t_3}~\infix{\atom{Left}}~{t_3'}{.} \\
    &\phantom{{}={}}\qquad {t_6}~\infix{\atom{Right}}~{t_6'}{.}
  \end{align*}
  whilst the special \atom{Start} and \atom{Stop} states are marked thus,
  \begin{align*}
    & {!}~\atom{Start}~{t_1}~{t_2}~{t_3}~{t_4}~{t_5}~{t_6}{;} \\
    & {!}~\atom{Stop}~{t_1}~{t_2}~{t_3}~{t_4}~{t_5}~{t_6}{;}
  \end{align*}
  That this translation faithfully reproduces the operation of the RTM is obvious from the high level semantics of \textAleph.
\end{proof}
\begin{crl}The \textslAleph\ calculus is Turing Complete.\end{crl}
\begin{proof}
  \textcite{bennett-tm} proved that a Reversible Turing Machine can simulate any Turing Machine (and vice-versa), and so the corollary immediately follows immediately from \textAleph's Reversible-Turing completeness. To drive the point home, however, we also implement the $\mu$-Recursive Functions (\Cref{lst:aleph-murec})---three functions, and three composition operators, which together are capable of representing any computable function over the naturals and hence are Turing complete. For example, addition, multiplication and factorials can be defined in terms of each other (in our \textAleph\ realisation) as
  \begin{align*}
    \textit{add} &\mapsto (\atom{Rec}~(\atom{Proj}~1)~(\atom{Sub}~\atom{Succ}~[(\atom{Proj}~2)])) \\
    \textit{mul} &\mapsto (\atom{Rec}~(\atom{Const}~0)~(\atom{Sub}~\textit{add}~[(\atom{Proj}~2)~(\atom{Proj}~3)])) \\
    \textit{fac} &\mapsto (\atom{Rec}~(\atom{Const}~1)~(\atom{Sub}~\textit{mul}~[(\atom{Proj}~2)~(\atom{Sub}~\atom{Succ}~[(\atom{Proj}~1)])]))
  \end{align*}
  and then one can evaluate, e.g., $(\atom{Mu}~\textit{fac})~{[7]}~{\unit}$, obtaining $\unit~5040~\textit{garbage}~(\atom{Mu}~\textit{fac})$.
\end{proof}

%% file: aleph-semex-add.tex
\begingroup
\begin{sublisting}{\linewidth}
  \def\ruleLabel#1{\llap{\tiny$#1.$}~}
  \begin{align*}
    \ruleLabel{1}& {!}~{+}~{a}~{b}~{\unit}{;} \\
    \ruleLabel{2}& {!}~{\unit}~{a}~{b}~{+}{;} \\
    \ruleLabel{3}& \{\alpha:{+}~{\Z}~{b}~{\unit}\} = \{\alpha:{\unit}~{\Z}~{b}~{+}\}{;} && \ruleName{add--base} \\
    \ruleLabel{4}& \{\alpha:{+}~({\S} {a})~{b}~{\unit}\} = \{\alpha:{\unit}~({\S} {a})~({\S} {b'})~{+}\}{:} && \ruleName{add--step} \\
    \raisebox{0pt}[1.9em]{}%
      {\begin{aligned}\ruleLabel{4a}&\\\ruleLabel{4b}&\end{aligned}}&%
      {\begin{aligned}
        &\qquad  \beta: {+}~{a}~{b}~{\unit}{.} \\
        &\qquad  \beta: {\unit}~{a}~{b'}~{+}{.}
      \end{aligned}} && \ruleName{add--step--sub}
  \end{align*}
  \caption{Desugared definition of reversible natural addition in \textAleph. It will be convenient to make the definitions $I_4=\{\alpha:{+}~({\S} {a})~{b}~{\unit}\}$, $O_4=\{\alpha:{\unit}~({\S} {a})~({\S} {b'})~{+}\}$ and $R_4=\{\beta: {+}~{a}~{b}~{\unit}, \beta: {\unit}~{a}~{b'}~{+}\}$.}
  \label{lst:semex-add-def}
\end{sublisting}

\begin{sublisting}{\linewidth}
  \def\adj{~}
  \def\smS{\textsc{\scriptsize S}}
  \begin{align*}
  &\phantom{{}\leftrightarrow{}} \{ {\unit}{:}[{+}{3}{4}{\unit}] \} && \\
  \mathcal P \vdash {!}{+}{a}{b}{\unit} \implies\adj& \leftrightarrow \{\langle{+}{a}{b}{\unit}\rangle,{\unit}{:}[{+}{3}{4}{\unit}]\} && \ruleName{inst--halt} \\
  & \leftrightarrow \{\langle{+}{a}{b}{\unit}\rangle,{\unit}{:}({+}{3}{4}{\unit})\} && \ruleName{term} \\
  \mathcal P \vdash {!}{+}{a}{b}{\unit} \implies\adj& \leftrightarrow \{{\unit}{:}({+}{3}{4}{\unit})\} && \ruleName{inst--halt} \\
  \mathcal P \vdash (I_4=O_4:R_4) \implies\adj& \leftrightarrow \{\langle I_4\varnothing\varnothing R_4\varnothing O_4\rangle,{\unit}{:}({+}{3}{4}{\unit})\} && \ruleName{inst--comp} \\
  \smash{{\unit}{:}({+}{3}{4}{\unit})\overset{{\alpha}{:}{+}({\smS}{a}){b}{\unit}}\sim\{{\cdots}\}} \implies\adj& \leftrightarrow \{\langle\varnothing I_4\{\alpha{\mapsto}{\unit},a{\mapsto}2,b{\mapsto}4\}R_4\varnothing O_4\rangle\} && \ruleName{inp} \\
  {\cdots} \implies\adj& \leftrightarrow \{\varnothing I_4\{\alpha{\mapsto}{\unit}\}R_4\{\beta{:}[{+}{2}{4}{\unit}]\}O_4\} && \ruleName{sub$_1$} \\
  \{\beta{:}[{+}{2}{4}{\unit}]\}\leftrightarrow\{\beta{:}[{\unit}{2}{6}{+}]\} \implies\adj& \leftrightarrow \{\varnothing I_4\{\alpha{\mapsto}{\unit}\}R_4\{\beta{:}[{\unit}{2}{6}{+}]\}O_4\} && \ruleName{sub$_3$} \\
  {\cdots} \implies\adj& \leftrightarrow \{\langle\varnothing I_4\{\alpha{\mapsto}{\unit},a{\mapsto}2,b'{\mapsto}6\}R_4\varnothing O_4\rangle\} && \ruleName{sub$_1$} \\
  & \leftrightarrow \{\langle\varnothing O_4\{\alpha{\mapsto}{\unit},a{\mapsto}2,b'{\mapsto}6\}R_4\varnothing I_4\rangle\} && \ruleName{rev$_2$} \\
  \smash{{\unit}{:}({\unit}{3}{7}{+})\overset{{\alpha}{:}{\unit}({\smS}{a})({\smS}{b'}){+}}\sim\{{\cdots}\}} \implies\adj& \leftrightarrow \{\langle O_4\varnothing\varnothing R_4\varnothing I_4\rangle,{\unit}{:}({\unit}{3}{7}{+})\} && \ruleName{inp} \\
  & \leftrightarrow \{\langle I_4\varnothing\varnothing R_4\varnothing O_4\rangle,{\unit}{:}({\unit}{3}{7}{+})\} && \ruleName{rev$_1$} \\
  \mathcal P \vdash (I_4=O_4:R_4) \implies\adj& \leftrightarrow \{{\unit}{:}({\unit}{3}{7}{+})\} && \ruleName{inst--comp}\\
  \mathcal P \vdash {!}{\unit}{a}{b}{+} \implies\adj& \leftrightarrow \{\langle{\unit}{a}{b}{+}\rangle,{\unit}{:}({\unit}{3}{7}{+})\} && \ruleName{inst--halt} \\
  & \leftrightarrow \{\langle{\unit}{a}{b}{+}\rangle,{\unit}{:}[{\unit}{3}{7}{+}]\} && \ruleName{term} \\
  \mathcal P \vdash {!}{\unit}{a}{b}{+} \implies\adj& \leftrightarrow \{{\unit}{:}[{\unit}{3}{7}{+}]\} && \ruleName{inst--halt}
  \end{align*}
  \caption{One possible derivation of ${\unit}:[{+}~{3}~{4}~{\unit}] \leftrightarrow {\unit}:[{\unit}~{3}~{7}~{+}]$ using the semantics for $\textAleph$.}
  \label{lst:semex-add-sem}
\end{sublisting}
\endgroup

%% file: aleph-tc-tape.tex
\begingroup%
\def\nl{\\[1.5em]}%
\begin{align*}
  & {!}~\atom{Tape}~{\ell}~{x}~{r}{;} \\
  & {!}~\atom{Sym}~{x}{;} \\
  & {!}~\atom{Blank}{;}
  \nl
  & {[\atom{Blank}~{x}~{\cdot}~\cursivexS]}~\infix{\atom{Pop}}~\atom{Blank}~{[{x}~{\cdot}~\cursivexS]}{;} \\
  & {[{(\atom{Sym}~{x})}~{\cdot}~\cursivexS]}~\infix{\atom{Pop}}~\atom{Blank}~\cursivexS{;} \\
  & {[]}~\infix{\atom{Pop}}~\atom{Blank}~{[]}{;}
  \nl
  & {(\atom{Tape}~{\ell}~{x}~{r})}~\infix{\atom{Left}}~{(\atom{Tape}~{\ell'}~{x'}~{r'})}{:} \\
  &\qquad {\ell}~\infix{\atom{Pop}}~{x'}~{\ell'}{.} \\
  &\qquad {r'}~\infix{\atom{Pop}}~{x}~{r}{.} \\
  & {t}~\infix{\atom{Right}}~{t'}{:} \\
  &\qquad {t'}~\infix{\atom{Left}}~{t}{.}
  % \nl 
  % & \cursivexS~\infix{\atom{List$\leftrightarrow$Tape}}~{(\atom{Tape}~{[]}~\atom{Blank}~{\cursivexS'})}{:} \\
  % &\qquad {x}~{\atomLocal{}}~{(\atom{Sym}~{x})}{;} \\
  % &\qquad \cursivexS~\infix{\atom{Map}~{\atomLocal{}}}~{\cursivexS'}{.}
\end{align*}%
\endgroup

%% file: aleph-tc-mu.tex
\begingroup%
\def\nl{\\[.8em]}%
\def\Mu#1{\infix{\atom{Mu}~{#1}}}%
\def\Garb#1{(\atom{Garbage}~{#1})}%
\def\Dup#1#2{\infix{\atom{Dup}~{#1}}~{#2}{.}}%
\def\xyz{x\hspace{-0.3ex}y\hspace{-0.3ex}z}%
\begin{minipage}[t]{.68\linewidth}\begin{align*}
  & \cursivexS~\Mu{(\atom{Const}~{n})}~{n'}~\Garb\cursivexS{:} \\
  &\qquad \Dup{n}{n'}
  \nl
  & {[x]}~\Mu{\atom{Succ}}~{(\S x)}~{(\atom{Garbage})}{;}
  \nl
  & \cursivexS~\Mu{(\atom{Proj}~{i})}~{y}~\Garb{\cursiveyS~\cursivezS}{:} \\
  &\qquad \Dup{i}{i'} \\
  % &\qquad\comment{\cmtSLopen~split $\cursivexS$ at position $i'$, with the prefix in reverse order} \\
  &\qquad \llap{{!}~}\atomLocal{Go}~{i'}~{[]}~\cursivexS = \atomLocal{Go}~{\Z}~{[{y}~{\cdot}~\cursiveyS]}~\cursivezS{.} \\
  &\qquad \atomLocal{Go}~{(\S n)}~\cursiveyS~{[{x}~{\cdot}~\cursivexS]} = \atomLocal{Go}~{n}~{[{y}~{\cdot}~\cursiveyS]}~\cursivexS{;}
  \nl
  & \cursivexS~\Mu{(\atom{Sub}~{g}~\cursivefS)}~{z}~\Garb{\cursivexS~{h}~\cursivehS}{:} \\
  &\qquad \cursiveyS~\Mu{g}~{z}~{h}{.} \\
  &\qquad \infix{\atomLocal{Go}~\cursivefS~\cursivexS}~\cursiveyS~\cursivehS{.} \\
  &\qquad \infix{\atomLocal{Go}~{[]}~\cursivexS}~{[]}~{[]}{;} \\
  &\qquad \infix{\atomLocal{Go}~{[{f}~{\cdot}~\cursivefS~\cursivexS]}}~{[{y}~{\cdot}~\cursiveyS]}~{[{h}~{\cdot}~\cursivehS]}{:} \\
  &\qquad\qquad \Dup\cursivexS{\cursivexS'} \\
  &\qquad\qquad {\cursivexS'}~\Mu{f}~{y}~{h}{.} \\
  &\qquad\qquad \infix{\atomLocal[2]{Go}~\cursivefS~\cursivexS}~\cursiveyS~\cursivehS{.}
  \nl
  & {[{n}~{\cdot}~\cursivexS]}~\Mu{(\atom{Rec}~{f}~{g})}~{y}~\Garb{{n'}~\cursivexS~\cursivehS}{:} \\
  &\qquad \Dup\cursivexS{\cursivexS'} \\
  &\qquad {\cursivexS'}~\Mu{f}~{x}~{h}{.} \\
  &\qquad \llap{{!}~}\atomLocal{Go}~{n}~{\Z}~{g}~{x}~\cursivexS~{[h]} =
                     \atomLocal{Go}~{\Z}~{n'}~{g}~{y}~\cursivexS~\cursivehS{.} \\
  &\qquad \atomLocal{Go}~{(\S n)}~{m}~{g}~{x}~\cursivexS~{hs} =
          \atomLocal{Go}~{n}~{(\S m)}~{g}~{y}~\cursivexS~{[{h}~{\cdot}~\cursivehS]}{:} \\
  &\qquad\qquad \Dup{m}{m'}~\Dup\cursivexS{\cursivexS'} \\
  &\qquad\qquad [{m'}~{x}~{\cdot}~{\cursivexS'}]~\Mu{g}~{y}~{h}{.}
  \nl
  & \cursivexS~\Mu{(\atom{Minim}~{f})}~{i}~\Garb{\cursivexS~\cursivehS}{:} \\
  &\qquad \llap{{!}~}\atomLocal{Go}~{f}~{\Z}~\cursivexS~{(\S\Z)}~{[]} =
                     \atomLocal{Go}~{f}~{(\S i)}~\cursivexS~{\Z}~\cursivehS{.} \\
  &\qquad \atomLocal{Go}~{f}~{i}~\cursivexS~{(\S n)}~\cursivehS =
          \atomLocal{Go}~{f}~{(\S i)}~\cursivexS~{y}~{[{n}~{h}~{\cdot}~\cursivehS]}{:} \\
  &\qquad\qquad \Dup{i}{i'}~\Dup\cursivexS{\cursivexS'} \\
  &\qquad\qquad {[{i'}~{\cdot}~{\cursivexS'}]}~\Mu{f}~{y}~{h}{.}
\end{align*}\end{minipage}%
\def\expln#1{\comment{\cmtSLopen~$#1$}}
\begin{minipage}[t]{.3\linewidth}\vspace{-.5em}%
\begin{equation*}\begin{array}{|rl}
  \qquad&\\[-.5em]
  & \dataDef{\atom{Const}~{n}}{;} \\
  & \expln{C_n(\vec\cursivexS) = n} \nl
  & \dataDef{\atom{Succ}}{;} \\
  & \expln{S(x) = x + 1} \nl
  & \dataDef{\atom{Proj}~{i}}{;} \\
  & \expln{P_i(\vec\cursivexS) = x_i} \\
  & \comment{\cmtSLopen~N.B:\ $\vec\cursivexS$\, is 1-indexed.}\nl
  & \dataDef{\atom{Sub}~{g}~\cursivefS}{;} \\
  & \expln{(g\circ\vec\cursivefS)(\vec\cursivexS) = } \\
  & \expln{\quad g(f_1(\vec\cursivexS)\cdots f_n(\vec\cursivexS))} \nl
  & \dataDef{\atom{Rec}~{f}~{g}}{;} \\
  & \expln{\rho(f,g)(0,\vec\cursivexS) = f(\vec\cursivexS)} \\
  & \expln{\rho(f,g)(n+1,\vec\cursivexS) = } \\
  & \expln{\quad g(n,\rho(f,g)(n,\vec\cursivexS),\vec\cursivexS)} \nl
  & \dataDef{\atom{Minim}~{f}}{;} \\
  & \expln{\mu(f)(\vec\cursivexS) = } \\
  & \expln{\quad\min\{n:f(n,\vec\cursivexS)=0\}}
  \\[-.5em]&
\end{array}\end{equation*}
\end{minipage}%
\endgroup

%% file: part-impl.tex
\section{Implementation Concerns}
\label{sec:impl}

Although the \textAleph\ calculus is microscopically reversible by \Cref{thm:mrev}, \textAleph\ has many degrees of freedom wherein macroscopic reversibility can be violated. There are two sources of entropy generation; the first is ambiguity in rule application, and the second is intrinsic to any useful implementation of concurrency. We already saw in \Cref{sec:ex2} how concurrency leads to entropy generation, but in this section we will expand on the other mechanism of entropy generation and how it can be avoided at the compiler-level. This is important because the primary motivation for reversible programming is to avoid unexpected entropy leaks, and so generally the appearance of ambiguity is a bug rather than intentional. In addition to ambiguity, we will also discuss other aspects of a realistic implementation to avoid or minimise the presence of random walks: optimising the evaluation order of sub-rules (`serialisation') and inferring the direction of computation. A reference implementation of these algorithms (as well as additional sugar) is made available as a language and interpreter, \alethe\ (see \Cref{sec:alethe}).

Nevertheless, in some cases the appearance of ambiguity/non-determinism may be intentional and so we should provide the programmer a mechanism to \emph{explicitly} weaken the ambiguity checker when desired. An example of this might be in implementing a fair coin toss for generating randomness,
\begin{align*}
  & \infix{\atom{Coin}}~\atom{Tails}{;} \\
  & \infix{\atom{Coin}}~\atom{Heads}{;}
\end{align*}
If the computational architecture is Brownian, such as a molecular computer, then this can be used to exploit the thermal noise of the environment to get as close as classical physics allows to a true random number generator (RNG). The way this RNG is used is by constructing a term $\atom{Coin}~{\unit}$. This term then matches both rules, and so depending on which is chosen the term may evolve to ${\unit}~\atom{Tails}~\atom{Coin}$ or to ${\unit}~\atom{Heads}~\atom{Coin}$. Moreover, this process will depend on the probability distribution realised by the implementation; if we assume this is uniform\footnotemark, then the coin toss will be fair with each term observed with probability $\frac12$. If the probability distribution is non-uniform, the situation is less clear and depends on the exact rates of the forward and reverse transitions for each rule. The analysis is beyond the scope of this paper, but if the dynamics takes the form\footnotetext{A more sophisticated implementation might reify and expose control over the distribution to the programmer, allowing arbitrary probabilities to be assigned. Going further, an interesting research direction would be to what extent \textAleph\ can be augmented quantum mechanically.}
\begin{align*}
  \ce{ ${\unit}~\atom{Tails}~\atom{Coin}$ <=>[$\lambda_t'$][$\lambda_t$] $\atom{Coin}~{\unit}$ <=>[$\lambda_h$][$\lambda_h'$] ${\unit}~\atom{Heads}~\atom{Coin}$ },
\end{align*}
then the steady-state ratio of heads to tails will be $\frac{\lambda_h/\lambda_h'}{\lambda_t/\lambda_t'}$. By microscopic reversibility, we expect $\lambda_h=\lambda_h'$ and similarly for $\lambda_t$ (unless the system is biased away from equilibrium), and therefore even a non-uniform distribution will have a uniformly distributed steady-state. Note the interesting property that, in reverse, an ambiguous/non-deterministic reversible rule is indistinguishable from irreversibility. Here, for example, the inverse of a fair coin toss is a process that consumes\footnotemark\ a bit.\footnotetext{There is a subtle point to make here: if the bit being consumed has no computational content and is uniformly distributed, then this consumption does not generate any entropy. That is, drawing a random bit from the environment and then replacing it is a completely isentropic process. For non-uniform distributions, the process can also be isentropic provided that the producer/consumer has a matching distribution. The problem arises when the distribution of the bit does not fit that of the producer/consumer, with deterministic computation being an extreme case wherein the bit can only take on a prescribed value (this is true even if the bit is pseudorandom).}

\para{Ambiguity Checker}

To exclude these sorts of process, we thus introduce an ambiguity checker. Not only does the introduction of an ambiguity checker reassure the programmer that they are writing information-preserving code, but the algorithm also serves as an invaluable debugging tool. In a large codebase, it can be hard to keep track of all the different patterns employed (the standard library has nearly two thousand), let alone ensure that there are no ambiguities. The ambiguity checker in \alethe\ performs this check for the programmer and, most importantly, is able to tell the programmer where the ambiguities lie.

In an unambiguous program, there are only a few valid scenarios for each term. It can match two computational rules (where we think of the forward and reverse directions of each rule as distinct rules for the current purpose), in which case it is an intermediate state of a computation. It can match one computational rule and one or more halting rules, in which case it is a terminal state of the computation. It can match one computational rule, in which case the program has entered an erroneous state and stalled. It can match one or more halting rules, in which case the term has no computational capacity, and most likely represents a data value. Lastly, it can match no rules, in which case it is an invalid term that cannot be constructed under normal conditions.

The remaining possibilities can be summarised two-fold: a term can match two computational rules and one or more halting rules, in which case the computational path from the terminus is ambiguous, or a term can match more than two computational rules (and zero or more halting rules), which is the more obvious ambiguity scenario.

The above cases can be simplified: for a given term, consider all the rules it matches but coalesce all halting rules, if any, into a single rule. Then, a term is unambiguous if it matches at most two such rules, and ambiguous if it matches more than two.

Clearly it is impractical to test this for all possible terms, there being a countable infinity of them; instead, one can determine whether it is possible to construct an ambiguous term. To do so, build a graph $G$ whose nodes are all the patterns occurring on either side of a computational rule---which we will colour white for reasons that will become apparent---and all the patterns occurring in a halting rule---which we will colour black. An edge is drawn between two nodes if it is possible to construct a term satisfying both patterns. This condition can be determined inductively: if either pattern is variable, then draw an edge. If neither pattern is variable, then draw an edge only if they are both the same atom, or they are both composite terms of the same length and their respective sub-patterns matches pairwise. In order to proceed, we first prove a lemma about this graph.

\begin{lem}There exists a term that simultaneously satisfies each of a set of patterns $\Pi$ if and only if the subgraph of $G$ formed from the nodes $\Pi$ is complete.\label{lem:amb-subgraph}\end{lem}
\begin{proof}
  The $(\Leftarrow)$ case is obvious. We prove the converse by structural induction. As base cases, we note that the statement is trivially true if $|\Pi|=0,1,2$. Now consider the possible forms of some $\pi\in\Pi$:
  \begin{itemize}
    \item[(\bnfPrim{atom})] If $\pi$ is an atom, $a$, then the only term satisfying it is the same atom, $a$. As $\pi$ has edges to every other pattern in $\Pi$, $a$ must satisfy each of these other patterns too.
    \item[(\bnfPrim{var})] Suppose there is a term $t$ that satisfies all the patterns $\Pi\setminus\{\pi\}$. As $t$ satisfies the variable pattern $\pi$ vacuously, $t$ must satisfy all of $\Pi$. Without loss of generality then, we can ignore all variable patterns.
    \item[($\tau_1\cdots\tau_n$)] If $\pi$ is a composite term, then there can only be an edge to the other patterns if these are composite terms of the same length or are the variable pattern. By the previous case, we can ignore these variable patterns. Now, index the patterns by $i=1\ldots m$ and write $\pi_i=(\tau^{(i)}_1\cdots\tau^{(i)}_n)$. Then construct graphs $H_1\cdots H_n$ such that the nodes of $H_j$ are the patterns $\{\tau^{(i)}_j:i=1\ldots m\}$. If we draw edges according to the same rules as for $G$, then by the inductive definition of the edge condition for $G$, each of the graphs $H_i$ will be complete. By inductive assumption, a term $t_i$ exists satisfying all patterns in $H_i$ for each $i$, and therefore the term $(t_1\cdots t_n)$ must satisfy all the patterns $\Pi$.
  \end{itemize}
  The lemma follows.
\end{proof}

Using \Cref{lem:amb-subgraph}, we see that the different cases of unambiguous and ambiguous terms reduce to considering the structures of the complete subgraphs of $G$. Moreover, it suffices to only consider subgraphs formed from three nodes, i.e.\ \emph{triangles}: the first ambiguous case, of two white nodes and one or more black nodes, implies the existence of a triangle with two white nodes and one black node; the second case, of three or more white nodes and zero or more black nodes, implies the existence of a triangle with three white nodes. Conversely, a triangle with two or more white nodes implies one of these ambiguous cases. The only triangles present for the unambiguous cases contain two or more black nodes. Therefore we can enumerate all the triangles of $G$ containing two or more white nodes; if there are any, then there is an ambiguity and the patterns present in the triangles should be reported to the programmer. If there are none, then the program is unambiguous and deterministic.

In fact, there is one more potential source of ambiguity. Suppose that the two patterns of a sub-rule are not orthogonal (a common occurrence), then it is in principle possible that a halting term might match either side, and therefore lead to a 2-way ambiguity. The evaluation rules of our reference interpreter for \alethe\ do not suffer from this ambiguity, as it determines ahead of time the order and directionality of all sub-rules, but it is possible for a faithful implementation of \textAleph\ to suffer so. There are two ways to avoid this: the first is to require sub-rule patterns be orthogonal, though this also forbids many legal unambiguous programs---requiring additional levels of tedious and unnecessary indirection to circumvent the restriction---and so is undesirable; the second is to apply type inference (work towards which is presented in \Cref{app:typing}) to gain the extra knowledge needed to distinguish between legal and ambiguous programs. Namely, type inference allows the compiler to infer what sequences of patterns and terms occur in a program, and hence to determine which halting patterns are computationally connected. With this knowledge, and knowledge of the types of variables in a rule, it can determine what forms the halting terms of a sub-rule will take and therefore whether there is a unique mapping between terms and patterns or not.

\para{Serialisation Heuristics}

The Brownian semantics of \textAleph, particularly in the absence of coupling to a bias source, are not very performant. Specifically, sub-rule transitions execute a random walk in a phase space whose size---in the worst case---scales exponentially with the number of variables present in the current rule. Whilst explicit bias coupling circumvents this by essentially serving to annotate the preferred order and directionality of the sub-rules, it undermines the declarative nature of \textAleph. Instead, it is possible to algorithmically infer an appropriate execution path. Moreover, it can sometimes be the case (such as in the example of fractional addition) that the execution path is non-trivial and so granting the compiler the power to perform this routing automatically makes the job of the programmer easier: the programmer need only specify how all the variables relate to one another, and may even specify this excessively, and the compiler can figure out which relations are sensible to use.

\begingroup
 \def\fr#1#2{#1\!/\!#2}%
 \def\Fr#1#2{(\atom{Frac}~{#1}~{#2})}%
 \def\fwd#1{\overset\rightarrow{#1}}%
 \def\bwd#1{\overset\leftarrow{#1}}%
 \def\ruleLabel#1{\llap{\tiny$#1.$}}%
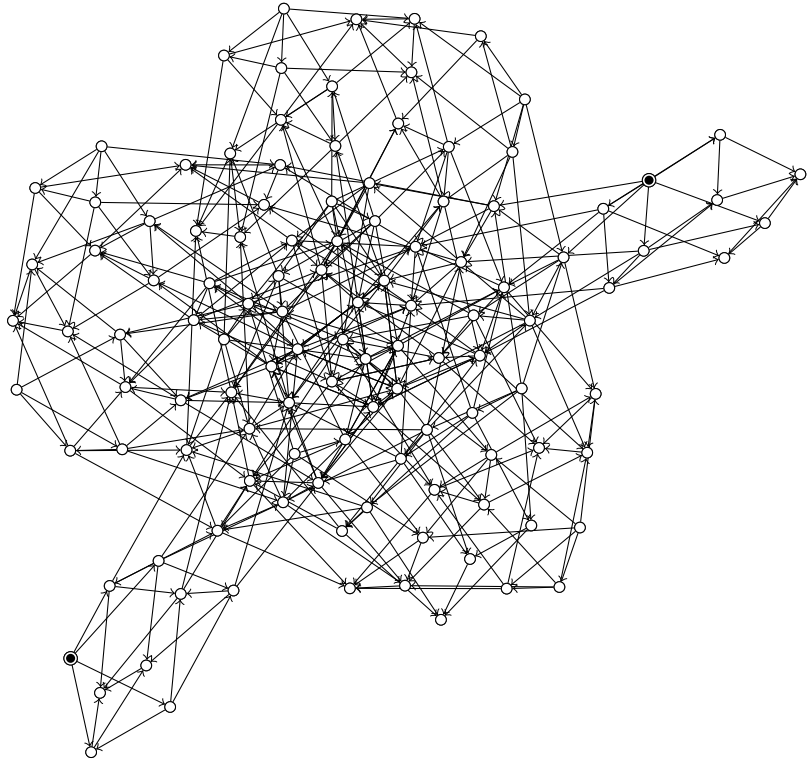
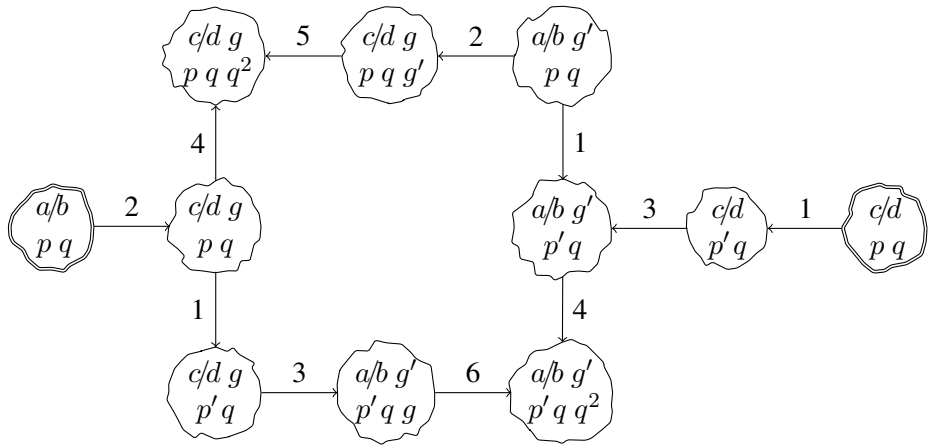
\begin{figure}
  \centering
  \begin{subfigure}{\textwidth}
    \centering
    \input{tg-frac.tikz}
    \vspace{0.5em}
    \caption{The full transition graph. Due to the size of the graph, labels have been suppressed. Halting states are marked black.}
    \label{fig:tg-frac-full}
  \end{subfigure}
  \begin{subfigure}{\textwidth}
  \centering
    \input{serial-frac}
    \vspace{0.5em}
    \caption{A more practical excerpt of the full transition graph.}
    \label{fig:tg-frac-excerpt}
  \end{subfigure}
  \def\fwd#1{$\vec{#1}$}%
  \def\bwd#1{\reflectbox{\ensuremath{\vec{\reflectbox{\ensuremath{#1}}}}}}%
  \def\en{$-$}%
  \caption{Two views of the transition graph for the fraction-addition routine. Nodes correspond to bags of variables in a `known' state. An edge with label $\ell$ is drawn from node $\alpha$ to node $\beta$ if sub-rule $\ell$ can map the knowledge state $\alpha$ to the knowledge state $\beta$ (with possible variable duplication/elision). It can be seen that there are two paths from $\{\fr ab,p,q\}$ to $\{\fr cd,p,q\}$ worth considering: \fwd2\en\fwd4\en\bwd5\en\bwd2\en\fwd1\en\bwd3\en\bwd1, and \fwd2\en\fwd1\en\fwd3\en\fwd6\en\bwd4\en\bwd3\en\bwd1. These are of roughly equal computational complexity, and so either is a viable choice.}
  \label{fig:tg-frac}
\end{figure}
To illustrate the problem, consider the following program excerpt which adds two simplified fractions:
  \begin{align*}
    & {\fr ab}~\infix{{+}~{\Fr pq}}~{\fr cd}{:} \\
    \ruleLabel{1}&\qquad {p}~{-}~{p'}{.} \\
    \ruleLabel{2}&\qquad {\fr ab}~\infix{\atomLocal{}~{\Fr pq}}~{\fr cd}~{g}{.} \\
    \ruleLabel{3}&\qquad {\fr cd}~\infix{\atomLocal{}~{\Fr{p'}q}}~{\fr ab}~{g'}{.} \\
    \ruleLabel{4}&\qquad {(\S q)}~{\square}~{q^2}{.} \\
    \ruleLabel{5}&\qquad {g'}~\infix{{\times}~{g}}~{q^2}{.} \\
    \ruleLabel{6}&\qquad {g}~\infix{{\times}~{g'}}~{q^2}{.}
  \end{align*}
where a fraction is represented by $\frac{p}{q+1} \equiv \Fr pq$ with $p\in\mathbb Z$ and $q\in\mathbb N$. Maintaining the invariant that fractions are in their simplest form is non-trivial, but involves finding the greatest common divisor of the numerators and denominators of two intermediate results, $g$ and $g'$, and then showing that their product $gg'=(q+1)^2$. It turns out that this fact can then be used to eliminate the intermediate garbage. In the unassisted semantics of \textAleph, the graph of all intermediate states generated by applying these rules transitively is given by \Cref{fig:tg-frac-full}. This graph has 117 nodes. It can therefore be seen that a computation is very likely to get stuck in one of the 115 intermediates, and will take a long time to find its way from the input variables, $\{\fr ab,p,q\}$, to the output variables, $\{p,q,\fr cd\}$. Many of these states are not interesting computational, and indeed only a small subgraph of 11 (or even 8) nodes is needed to perform the desired computation as shown in \Cref{fig:tg-frac-excerpt}.

From this transition subgraph, we can see that there are two possible paths to get from the input variables to the output variables, and the goal of the serialisation algorithm is to find these paths and to pick the most optimal. Here both of these paths are of equal complexity, and so either is valid. Having made a choice, the compiler can then direct the computer to perform the prescribed series of transitions, rather than executing a random walk. In a Brownian computational system, for example, this would be achieved by adding `fake' dependencies to the sub-rules to force a linear ordering (or possibly partially parallel, where this makes sense) and automatically coupling the sub-rules to the bias source in the correct direction. As we shall see, automatically determining the cost of a path is non-trivial and sometimes even impossible, and so a perfect solution to the serialisation problem is not generally possible. Nonetheless, for many programs it is possible to find the appropriate path, or to give the compiler enough information to make a better choice, and so further discussion of serialisation is warranted.
\endgroup

The simplest and least efficient serialisation strategy is as follows: suppose you start with the input variables $\{x,y,z\}$. Enumerate all the sub-rules that can be evaluated given this current `knowledge state', e.g.\ ${x}~\infix{\atom{Foo}}~{y}~\cursivewS$ and ${z}~\atomLocal{}~{x}$ would qualify but $\cursivewS~\infix{\atom{Map}~\atom{Bar}}~\cursivevS$ would not because neither the variable $\cursivewS$ nor the variable $\cursivevS$ are currently available. Evaluate all of these sub-rules in parallel, making sure to duplicate variables as necessary to avoid unlearning any. Here, we would end up with the knowledge state $\{\cursivewS,x,y,z\}$. Now ignore these sub-rules going forward and repeat this process until all sub-rules have been used once (and only once), and all variables are known. If some sub-rule has not been used or a variable remains unknown, then there is a logic error and the programmer should be appropriately chastised. Now, perform this entire process again, but starting from the output variables. If all is well, then we will have obtained two routes: one from the input variables to all variables, and one from the output variables to all variables. As these routes are reversible, we are immediately rewarded with a route from the input to the output variables, using each sub-rule twice. Whilst clearly inefficient, this strategy demonstrates that each sub-rule need be used at most twice (once in each direction).
\begin{listing}
  \centering
  \input{serial-polish}
  \caption{An example of serialisation as applied to the interconversion between trees and Polish notation.}
  \label{lst:serial-polish}
\end{listing}

A better strategy is to construct a `transition graph' as follows: take as nodes the powerset of all variables, i.e.\ $\varnothing$, $\{x\}$, $\{y\}$, $\{x,y\}$, etc. Mark the input and output nodes specially for later. Draw an edge between two nodes if a sub-rule (with possible variable duplication/elision) can be used to map between the two knowledge states, and label the edge with this sub-rule and in which direction it is to be applied. With the transition graph thus constructed, all possible routes between the input and output nodes can be enumerated\footnote{In fact there are infinitely many routes. It is therefore appropriate to restrict the routes under consideration to a sane (and finite) subset; namely, we avoid revisiting any node, and we also exclude routes that make use of a sub-rule more than twice.}. To proceed, a cost heuristic is needed in order to rank routes by preference. For example, if the goal is simply to minimise the number of sub-rules used then Dijkstra's algorithm will suffice. Note that, in the reference interpreter, transition graphs are constructed more restrictively. Namely, variable duplication/elision is not implemented as this always\footnote{The serialisation algorithm in \alethe\ is still subject to an exponential worst case complexity. The size of the transition graph constructed depends on the inter-dependency of the variables: the more dependency, the closer to linear in the number of variables the graph size is. If there is minimal dependency, such as in the automatically generated implemention of \atom{Dup} for $\dataDef{{,}~{a}~{b}~{c}~{d}~{e}}$, then the worst case will be realised.}\ leads to an exponentially large transition graph and, in practice, it is rarely needed. Where it is needed, the programmer must instead explicitly use \atom{Dup} (and create a new variable). This provides dramatically better performance at the expense of a slight inconvenience.

{\def\fwd#1{$\vec{#1}$}%
\def\bwd#1{\reflectbox{\ensuremath{\vec{\reflectbox{\ensuremath{#1}}}}}}%
\def\en{$-$}%
We saw an example of this algorithm earlier in \Cref{fig:tg-frac}. Another example, deserving of further comment, is given in \Cref{lst:serial-polish}. There are two routes of apparent equal cost, \fwd1\en\fwd2\en\bwd2\en\bwd3 and \fwd1\en\fwd2\en\fwd3\en\bwd3, but this is misleading. If the tree has $n$ nodes across $\ell\sim\log n$ levels, a single (non-recursive) step of \atom{Polish} has time complexity $\alpha$, and a single (non-recursive) step of \atom{PolishReads} has time complexity $\beta$, then the first route can be shown to have time complexity $\bigOO{2^\ell n(\alpha+\beta)}$ whilst the second has complexity $\bigOO{n(\alpha+2\ell\beta)}$. That is, if the serialisation algorithm picks a route which repeats a recursive step then consequently there will be an exponential overhead. Moreover, this becomes less obvious in the cases of corecursion, or higher order functions where a function may be passed into itself (recursively or corecursively)---compare, for example, the Y combinator. As such it is not generally possible to algorithmically discriminate this scenario (although the extra knowledge afforded by the type inference algorithm developed in \Cref{app:typing} may help). In fact, the reference interpreter for \alethe\ does not even try to do so. Instead, it exposes a method by which the programmer can adjust the heuristic cost function used in serialisation: each edge in the transition graph has a weight, given by the number of full-stops following the sub-rule in the source; therefore, in this case rule \ruleName{2} should be annotated with a cost of 2 (via two full-stops) in order to penalise its repeated use. To inspect the route chosen, the \texttt{:p} directive may be issued to the interpreter.}

\para{Directional Evaluation}

Consider an intermediate term along a reversible computation path (\Cref{fig:rev-classes}). In isolation, one cannot know in which direction computation should proceed. Bias coupling helps with this, but \alethe\ does not make use of this. Moreover, it would be nice to avoid explicit bias-coupling where practical. As such, we need some concept of computational `momentum' to maintain a consistent direction. This requires that the programs are deterministic such that phase space is branchless, which fortunately we have guaranteed through the ambiguity checker. Recall that, treating the forward and reverse directions of a rule as distinct, every term matches at most two rules; if we maintain a consistent direction of computation, then one of these rules will be the converse of the rule that was employed to reach the current state. More concretely, let the terms of the computation be labelled as some contiguous subset of $\{\ket{n}:n\in\mathbb Z\}$. By determinism, there is a unique rule $r(n\mapsto n+1)$ enacting each transition $\ket{n}\mapsto\ket{n+1}$, and by reversibility the unique rule enacting the transition $\ket{n+1}\mapsto\ket{n}$ is $\bar r$, the converse of $r$. Therefore, given a term $\ket{n+1}$, the rule $\bar r(n\mapsto n+1)$ can \emph{only} take us to $\ket{n}$; it can never also be the rule taking us to $\ket{n+2}$. It is, however, possible that the rule taking us to $\ket{n+2}$ is the same as $r(n\mapsto n+1)$. To summarise, all we need do is keep track of which rule we just applied, $r$; then, when considering the new term, we find any rules it matches and exclude $\bar r$ from this set. If the set is empty, we have entered an error state and should report it to the user. Otherwise, our ambiguity check has ensured that it either contains a single computational rule, which we duly apply, or contains one or more halting rules, whence we would halt the computation and report the result. This evaluation logic applies just as well to sub-rules. The one edge case is when evaluating a new term; as only halting terms can be constructed, it will have at most one computational rule to choose from and thus there is no ambiguity.

%% file: tg-frac.tikz
\begin{tikzpicture}[every state/.style={inner sep=0pt,minimum size=4pt}]
  \node[state,accepting,fill=black] (v0) at (3.4751484217817574,4.55031389406081) {};
  \node[state,accepting,fill=black] (v1) at (-4.135608504768771,-1.788929148574084) {};
  \node[state] (v2) at (2.5661983373642028,-0.054550676247050865) {};
  \node[state] (v3) at (2.2960064979553976,-0.8433926985344079) {};
  \node[state] (v4) at (0.16248274031146165,1.786648833053601) {};
  \node[state] (v5) at (-0.15403842532045073,1.5440941959252126) {};
  \node[state] (v6) at (1.304412630534438,0.24959210370513552) {};
  \node[state] (v7) at (-0.254561672247947,2.182077857105258) {};
  \node[state] (v8) at (0.17203732818832979,2.3555946149164053) {};
  \node[state] (v9) at (1.3996923141298372,0.9086729504972618) {};
  \node[state] (v10) at (1.6020290554355545,-0.8648324999393118) {};
  \node[state] (v11) at (-0.552154355514626,2.4391021190004674) {};
  \node[state] (v12) at (-0.6973961754798681,1.8807830126933425) {};
  \node[state] (v13) at (0.6520192928433094,0.44591629179692915) {};
  \node[state] (v14) at (1.9070232370458056,2.689304845402798) {};
  \node[state] (v15) at (1.2488960070372646,2.2247562734784028) {};
  \node[state] (v16) at (2.661815852528711,0.9380111542120468) {};
  \node[state] (v17) at (2.3505197522682204,3.5277773604112603) {};
  \node[state] (v18) at (1.5664526660483733,3.1303400311863356) {};
  \node[state] (v19) at (2.772958367441803,1.7178942713417453) {};
  \node[state] (v20) at (1.1667166458346292,2.758762333129569) {};
  \node[state] (v21) at (0.7100459088688089,2.1942193178406155) {};
  \node[state] (v22) at (2.0277899926416456,1.0011252426628683) {};
  \node[state] (v23) at (-1.339483077259822,0.27989229516804753) {};
  \node[state] (v24) at (-0.8772538864964434,0.536683076153312) {};
  \node[state] (v25) at (0.26442668464939034,-0.8251525960468068) {};
  \node[state] (v26) at (-1.1887503535337274,0.9229365672083918) {};
  \node[state] (v27) at (-0.5220645608338238,1.113583659319007) {};
  \node[state] (v28) at (0.49970769928371683,-0.18835218571858217) {};
  \node[state] (v29) at (-2.2033991623148776,-0.0965970082353895) {};
  \node[state] (v30) at (-1.7770409660693134,0.564319902527252) {};
  \node[state] (v31) at (-0.4609233675211441,-0.8611154275246178) {};
  \node[state] (v32) at (0.5532793282134818,1.241257847853684) {};
  \node[state] (v33) at (0.20963025253434922,0.8566729356953052) {};
  \node[state] (v34) at (1.120024992052391,-0.4719491452269722) {};
  \node[state] (v35) at (1.7993723129999715,1.7891708595447589) {};
  \node[state] (v36) at (1.1506358393486438,1.4657206941192418) {};
  \node[state] (v37) at (1.9259874103922248,-0.027611618119865754) {};
  \node[state] (v38) at (-0.23573789648295948,0.21327848485417994) {};
  \node[state] (v39) at (-0.564327191675239,-0.10099075591427249) {};
  \node[state] (v40) at (0.7372891135010878,-1.2775846259696395) {};
  \node[state] (v41) at (-1.3293196434637038,6.825497502439574) {};
  \node[state] (v42) at (-1.3630277721388893,6.0371026696706975) {};
  \node[state] (v43) at (-0.009805064791322347,3.2193195972644784) {};
  \node[state] (v44) at (-0.35504279113692555,2.9305782121524557) {};
  \node[state] (v45) at (-0.651669539959817,5.005092351126771) {};
  \node[state] (v46) at (-0.7002944510123187,4.269292058962024) {};
  \node[state] (v47) at (-0.13093864618508658,4.0082539930566) {};
  \node[state] (v48) at (-0.6949983650431387,5.794852862726046) {};
  \node[state] (v49) at (-2.117919726432933,6.20351206311362) {};
  \node[state] (v50) at (-0.6272690994902363,3.741847551248912) {};
  \node[state] (v51) at (-0.835456395648763,3.360834907797988) {};
  \node[state] (v52) at (-1.3690202376417993,5.351162126862515) {};
  \node[state] (v53) at (1.6779545179878366,4.925632693215907) {};
  \node[state] (v54) at (0.771231706589196,4.26871006333689) {};
  \node[state] (v55) at (0.3486980825664952,5.981038960355102) {};
  \node[state] (v56) at (1.8415659900291574,5.621816136969723) {};
  \node[state] (v57) at (0.84027403039665,4.993727240705263) {};
  \node[state] (v58) at (0.3905073363260326,6.690665445802169) {};
  \node[state] (v59) at (1.2619326187113395,6.459990553922406) {};
  \node[state] (v60) at (0.17543782511852904,5.302361524234345) {};
  \node[state] (v61) at (-0.37510981434760576,6.683699008727237) {};
  \node[state] (v62) at (-1.7817404196988733,1.2561509790322165) {};
  \node[state] (v63) at (-1.2612236672081574,1.6004460176638242) {};
  \node[state] (v64) at (-1.904406825150545,3.7967751364142064) {};
  \node[state] (v65) at (-2.11743550263423,2.4397611178958627) {};
  \node[state] (v66) at (-1.3999694216719738,3.280686504771533) {};
  \node[state] (v67) at (-2.0354296517343657,4.905352720781203) {};
  \node[state] (v68) at (-2.6101018130148055,0.9730542191814799) {};
  \node[state] (v69) at (-2.020355810288433,1.7370810263003533) {};
  \node[state] (v70) at (-2.486983921441511,3.8774224320592063) {};
  \node[state] (v71) at (3.4050519785962092,3.614477023449859) {};
  \node[state] (v72) at (2.951951841598388,3.122268021278065) {};
  \node[state] (v73) at (4.371534251111887,4.287617500246496) {};
  \node[state] (v74) at (2.873068166128606,4.1721933979584716) {};
  \node[state] (v75) at (4.410591385501697,5.152643917307501) {};
  \node[state] (v76) at (4.99976881745918,3.981625148104324) {};
  \node[state] (v77) at (4.466333710609381,3.519617162952977) {};
  \node[state] (v78) at (5.467030126379853,4.628829211366017) {};
  \node[state] (v79) at (-3.7279072761198657,5.000248792379955) {};
  \node[state] (v80) at (-3.811215928042087,4.256826731203449) {};
  \node[state] (v81) at (-1.1484351173644185,2.3102568936610397) {};
  \node[state] (v82) at (-1.49330932123178,2.0808475166081846) {};
  \node[state] (v83) at (-3.04196135531578,3.222060047967823) {};
  \node[state] (v84) at (-1.798227264463192,2.9128373636056173) {};
  \node[state] (v85) at (-1.3491116087899795,2.809316645968066) {};
  \node[state] (v86) at (-3.0951374519019654,4.007893143989082) {};
  \node[state] (v87) at (-4.600826872525891,4.447456218048963) {};
  \node[state] (v88) at (-2.307749022090527,3.183685047046739) {};
  \node[state] (v89) at (-2.515677050667505,2.694931513607954) {};
  \node[state] (v90) at (-3.8137665848881666,3.6172165657123405) {};
  \node[state] (v91) at (1.002213660425983,3.4655912741137294) {};
  \node[state] (v92) at (0.3443864049944915,2.8873559381409506) {};
  \node[state] (v93) at (-1.5906218982805276,4.224465426911279) {};
  \node[state] (v94) at (1.4344629465622716,4.2064242330833865) {};
  \node[state] (v95) at (0.40293336649773337,3.668404775626401) {};
  \node[state] (v96) at (-1.3778565222868036,4.752603364273528) {};
  \node[state] (v97) at (-0.2032210238671238,4.512173215707064) {};
  \node[state] (v98) at (-1.2245017601519739,3.7507375123253857) {};
  \node[state] (v99) at (-2.6222574501655442,4.755098060228369) {};
  \node[state] (v100) at (-3.456375950023912,0.9823133057465543) {};
  \node[state] (v101) at (-2.6861391224896876,1.6320204510600846) {};
  \node[state] (v102) at (-4.171686255867115,2.5452878603747413) {};
  \node[state] (v103) at (-4.846204597380111,1.7726356240289933) {};
  \node[state] (v104) at (-3.48423624860414,2.5038693111160577) {};
  \node[state] (v105) at (-4.6411671498241045,3.433416637619054) {};
  \node[state] (v106) at (-4.140740048676738,0.9642121411754452) {};
  \node[state] (v107) at (-3.4143045355450123,1.8075370736285366) {};
  \node[state] (v108) at (-4.894307625170197,2.686889645805903) {};
  \node[state] (v109) at (-2.978222856512475,-0.49678825604049753) {};
  \node[state] (v110) at (-1.9923798135757813,-0.8898958751013494) {};
  \node[state] (v111) at (-3.140303986211371,-1.883472969363521) {};
  \node[state] (v112) at (-2.825624707200261,-2.4337682388261377) {};
  \node[state] (v113) at (-3.865554251004385,-3.033553053509924) {};
  \node[state] (v114) at (-3.621750590672103,-0.8287875853122821) {};
  \node[state] (v115) at (-2.686550140289986,-0.9348164477976713) {};
  \node[state] (v116) at (-3.749129320562069,-2.246746115270515) {};
  \path (v0) edge[->] (v17);
  \path (v0) edge[->] (v71);
  \path (v0) edge[->] (v74);
  \path (v0) edge[->] (v75);
  \path (v0) edge[->] (v76);
  \path (v0) edge[->] (v94);
  \path (v1) edge[->] (v109);
  \path (v1) edge[->] (v112);
  \path (v1) edge[->] (v113);
  \path (v1) edge[->] (v114);
  \path (v2) edge[->] (v3);
  \path (v2) edge[->] (v9);
  \path (v2) edge[->] (v10);
  \path (v2) edge[->] (v19);
  \path (v2) edge[->] (v28);
  \path (v3) edge[->] (v6);
  \path (v3) edge[->] (v10);
  \path (v3) edge[->] (v16);
  \path (v3) edge[->] (v25);
  \path (v4) edge[->] (v5);
  \path (v4) edge[->] (v6);
  \path (v4) edge[->] (v11);
  \path (v4) edge[->] (v15);
  \path (v4) edge[->] (v24);
  \path (v4) edge[->] (v43);
  \path (v4) edge[->] (v81);
  \path (v5) edge[->] (v6);
  \path (v5) edge[->] (v12);
  \path (v5) edge[->] (v15);
  \path (v5) edge[->] (v24);
  \path (v5) edge[->] (v44);
  \path (v5) edge[->] (v82);
  \path (v6) edge[->] (v13);
  \path (v6) edge[->] (v16);
  \path (v6) edge[->] (v25);
  \path (v7) edge[->] (v4);
  \path (v7) edge[->] (v8);
  \path (v7) edge[->] (v9);
  \path (v7) edge[->] (v11);
  \path (v7) edge[->] (v18);
  \path (v7) edge[->] (v27);
  \path (v7) edge[->] (v84);
  \path (v8) edge[->] (v5);
  \path (v8) edge[->] (v9);
  \path (v8) edge[->] (v12);
  \path (v8) edge[->] (v18);
  \path (v8) edge[->] (v27);
  \path (v8) edge[->] (v85);
  \path (v9) edge[->] (v6);
  \path (v9) edge[->] (v13);
  \path (v9) edge[->] (v19);
  \path (v9) edge[->] (v28);
  \path (v10) edge[->] (v13);
  \path (v10) edge[->] (v22);
  \path (v10) edge[->] (v31);
  \path (v11) edge[->] (v12);
  \path (v11) edge[->] (v13);
  \path (v11) edge[->] (v21);
  \path (v11) edge[->] (v30);
  \path (v11) edge[->] (v50);
  \path (v12) edge[->] (v13);
  \path (v12) edge[->] (v21);
  \path (v12) edge[->] (v30);
  \path (v12) edge[->] (v51);
  \path (v13) edge[->] (v22);
  \path (v13) edge[->] (v31);
  \path (v14) edge[->] (v4);
  \path (v14) edge[->] (v15);
  \path (v14) edge[->] (v16);
  \path (v14) edge[->] (v20);
  \path (v14) edge[->] (v43);
  \path (v14) edge[->] (v91);
  \path (v15) edge[->] (v5);
  \path (v15) edge[->] (v16);
  \path (v15) edge[->] (v21);
  \path (v15) edge[->] (v44);
  \path (v15) edge[->] (v92);
  \path (v16) edge[->] (v22);
  \path (v17) edge[->] (v14);
  \path (v17) edge[->] (v18);
  \path (v17) edge[->] (v19);
  \path (v17) edge[->] (v20);
  \path (v17) edge[->] (v94);
  \path (v18) edge[->] (v15);
  \path (v18) edge[->] (v19);
  \path (v18) edge[->] (v21);
  \path (v18) edge[->] (v95);
  \path (v19) edge[->] (v16);
  \path (v19) edge[->] (v22);
  \path (v20) edge[->] (v11);
  \path (v20) edge[->] (v21);
  \path (v20) edge[->] (v22);
  \path (v20) edge[->] (v50);
  \path (v21) edge[->] (v12);
  \path (v21) edge[->] (v22);
  \path (v21) edge[->] (v51);
  \path (v23) edge[->] (v4);
  \path (v23) edge[->] (v24);
  \path (v23) edge[->] (v25);
  \path (v23) edge[->] (v29);
  \path (v23) edge[->] (v62);
  \path (v23) edge[->] (v81);
  \path (v24) edge[->] (v5);
  \path (v24) edge[->] (v25);
  \path (v24) edge[->] (v30);
  \path (v24) edge[->] (v63);
  \path (v24) edge[->] (v82);
  \path (v25) edge[->] (v31);
  \path (v26) edge[->] (v7);
  \path (v26) edge[->] (v23);
  \path (v26) edge[->] (v27);
  \path (v26) edge[->] (v28);
  \path (v26) edge[->] (v29);
  \path (v26) edge[->] (v84);
  \path (v27) edge[->] (v8);
  \path (v27) edge[->] (v24);
  \path (v27) edge[->] (v28);
  \path (v27) edge[->] (v30);
  \path (v27) edge[->] (v85);
  \path (v28) edge[->] (v25);
  \path (v28) edge[->] (v31);
  \path (v29) edge[->] (v30);
  \path (v29) edge[->] (v31);
  \path (v29) edge[->] (v68);
  \path (v30) edge[->] (v31);
  \path (v30) edge[->] (v69);
  \path (v32) edge[->] (v14);
  \path (v32) edge[->] (v23);
  \path (v32) edge[->] (v33);
  \path (v32) edge[->] (v34);
  \path (v32) edge[->] (v38);
  \path (v32) edge[->] (v62);
  \path (v32) edge[->] (v91);
  \path (v33) edge[->] (v15);
  \path (v33) edge[->] (v24);
  \path (v33) edge[->] (v34);
  \path (v33) edge[->] (v39);
  \path (v33) edge[->] (v63);
  \path (v33) edge[->] (v92);
  \path (v34) edge[->] (v40);
  \path (v35) edge[->] (v17);
  \path (v35) edge[->] (v32);
  \path (v35) edge[->] (v36);
  \path (v35) edge[->] (v37);
  \path (v35) edge[->] (v38);
  \path (v35) edge[->] (v94);
  \path (v36) edge[->] (v18);
  \path (v36) edge[->] (v33);
  \path (v36) edge[->] (v37);
  \path (v36) edge[->] (v39);
  \path (v36) edge[->] (v95);
  \path (v37) edge[->] (v34);
  \path (v37) edge[->] (v40);
  \path (v38) edge[->] (v29);
  \path (v38) edge[->] (v39);
  \path (v38) edge[->] (v40);
  \path (v38) edge[->] (v68);
  \path (v39) edge[->] (v30);
  \path (v39) edge[->] (v40);
  \path (v39) edge[->] (v69);
  \path (v41) edge[->] (v42);
  \path (v41) edge[->] (v48);
  \path (v41) edge[->] (v49);
  \path (v41) edge[->] (v58);
  \path (v41) edge[->] (v67);
  \path (v42) edge[->] (v45);
  \path (v42) edge[->] (v49);
  \path (v42) edge[->] (v55);
  \path (v42) edge[->] (v64);
  \path (v43) edge[->] (v4);
  \path (v43) edge[->] (v44);
  \path (v43) edge[->] (v45);
  \path (v43) edge[->] (v50);
  \path (v43) edge[->] (v54);
  \path (v43) edge[->] (v63);
  \path (v43) edge[->] (v81);
  \path (v44) edge[->] (v5);
  \path (v44) edge[->] (v45);
  \path (v44) edge[->] (v51);
  \path (v44) edge[->] (v54);
  \path (v44) edge[->] (v63);
  \path (v44) edge[->] (v82);
  \path (v45) edge[->] (v52);
  \path (v45) edge[->] (v55);
  \path (v45) edge[->] (v64);
  \path (v46) edge[->] (v7);
  \path (v46) edge[->] (v43);
  \path (v46) edge[->] (v47);
  \path (v46) edge[->] (v48);
  \path (v46) edge[->] (v50);
  \path (v46) edge[->] (v57);
  \path (v46) edge[->] (v66);
  \path (v46) edge[->] (v84);
  \path (v47) edge[->] (v8);
  \path (v47) edge[->] (v44);
  \path (v47) edge[->] (v48);
  \path (v47) edge[->] (v51);
  \path (v47) edge[->] (v57);
  \path (v47) edge[->] (v66);
  \path (v47) edge[->] (v85);
  \path (v48) edge[->] (v45);
  \path (v48) edge[->] (v52);
  \path (v48) edge[->] (v58);
  \path (v48) edge[->] (v67);
  \path (v49) edge[->] (v52);
  \path (v49) edge[->] (v61);
  \path (v49) edge[->] (v70);
  \path (v50) edge[->] (v51);
  \path (v50) edge[->] (v52);
  \path (v50) edge[->] (v60);
  \path (v50) edge[->] (v69);
  \path (v51) edge[->] (v52);
  \path (v51) edge[->] (v60);
  \path (v51) edge[->] (v69);
  \path (v52) edge[->] (v61);
  \path (v52) edge[->] (v70);
  \path (v53) edge[->] (v14);
  \path (v53) edge[->] (v54);
  \path (v53) edge[->] (v55);
  \path (v53) edge[->] (v59);
  \path (v53) edge[->] (v91);
  \path (v54) edge[->] (v15);
  \path (v54) edge[->] (v55);
  \path (v54) edge[->] (v60);
  \path (v54) edge[->] (v92);
  \path (v55) edge[->] (v61);
  \path (v56) edge[->] (v17);
  \path (v56) edge[->] (v53);
  \path (v56) edge[->] (v57);
  \path (v56) edge[->] (v58);
  \path (v56) edge[->] (v59);
  \path (v56) edge[->] (v94);
  \path (v57) edge[->] (v18);
  \path (v57) edge[->] (v54);
  \path (v57) edge[->] (v58);
  \path (v57) edge[->] (v60);
  \path (v57) edge[->] (v95);
  \path (v58) edge[->] (v55);
  \path (v58) edge[->] (v61);
  \path (v59) edge[->] (v60);
  \path (v59) edge[->] (v61);
  \path (v60) edge[->] (v61);
  \path (v62) edge[->] (v4);
  \path (v62) edge[->] (v63);
  \path (v62) edge[->] (v64);
  \path (v62) edge[->] (v68);
  \path (v62) edge[->] (v81);
  \path (v63) edge[->] (v5);
  \path (v63) edge[->] (v64);
  \path (v63) edge[->] (v69);
  \path (v63) edge[->] (v82);
  \path (v64) edge[->] (v70);
  \path (v65) edge[->] (v7);
  \path (v65) edge[->] (v62);
  \path (v65) edge[->] (v66);
  \path (v65) edge[->] (v67);
  \path (v65) edge[->] (v68);
  \path (v65) edge[->] (v84);
  \path (v66) edge[->] (v8);
  \path (v66) edge[->] (v63);
  \path (v66) edge[->] (v67);
  \path (v66) edge[->] (v69);
  \path (v66) edge[->] (v85);
  \path (v67) edge[->] (v64);
  \path (v67) edge[->] (v70);
  \path (v68) edge[->] (v69);
  \path (v68) edge[->] (v70);
  \path (v69) edge[->] (v70);
  \path (v71) edge[->] (v14);
  \path (v71) edge[->] (v72);
  \path (v71) edge[->] (v73);
  \path (v71) edge[->] (v76);
  \path (v71) edge[->] (v91);
  \path (v72) edge[->] (v15);
  \path (v72) edge[->] (v73);
  \path (v72) edge[->] (v77);
  \path (v72) edge[->] (v92);
  \path (v73) edge[->] (v78);
  \path (v74) edge[->] (v18);
  \path (v74) edge[->] (v72);
  \path (v74) edge[->] (v75);
  \path (v74) edge[->] (v77);
  \path (v74) edge[->] (v95);
  \path (v75) edge[->] (v73);
  \path (v75) edge[->] (v78);
  \path (v76) edge[->] (v77);
  \path (v76) edge[->] (v78);
  \path (v77) edge[->] (v78);
  \path (v79) edge[->] (v80);
  \path (v79) edge[->] (v86);
  \path (v79) edge[->] (v87);
  \path (v79) edge[->] (v96);
  \path (v79) edge[->] (v105);
  \path (v80) edge[->] (v83);
  \path (v80) edge[->] (v87);
  \path (v80) edge[->] (v93);
  \path (v80) edge[->] (v102);
  \path (v81) edge[->] (v4);
  \path (v81) edge[->] (v43);
  \path (v81) edge[->] (v82);
  \path (v81) edge[->] (v83);
  \path (v81) edge[->] (v88);
  \path (v81) edge[->] (v92);
  \path (v81) edge[->] (v101);
  \path (v82) edge[->] (v5);
  \path (v82) edge[->] (v44);
  \path (v82) edge[->] (v83);
  \path (v82) edge[->] (v89);
  \path (v82) edge[->] (v92);
  \path (v82) edge[->] (v101);
  \path (v83) edge[->] (v90);
  \path (v83) edge[->] (v93);
  \path (v83) edge[->] (v102);
  \path (v84) edge[->] (v81);
  \path (v84) edge[->] (v85);
  \path (v84) edge[->] (v86);
  \path (v84) edge[->] (v88);
  \path (v84) edge[->] (v95);
  \path (v84) edge[->] (v104);
  \path (v85) edge[->] (v82);
  \path (v85) edge[->] (v86);
  \path (v85) edge[->] (v89);
  \path (v85) edge[->] (v95);
  \path (v85) edge[->] (v104);
  \path (v86) edge[->] (v83);
  \path (v86) edge[->] (v90);
  \path (v86) edge[->] (v96);
  \path (v86) edge[->] (v105);
  \path (v87) edge[->] (v90);
  \path (v87) edge[->] (v99);
  \path (v87) edge[->] (v108);
  \path (v88) edge[->] (v11);
  \path (v88) edge[->] (v50);
  \path (v88) edge[->] (v89);
  \path (v88) edge[->] (v90);
  \path (v88) edge[->] (v98);
  \path (v88) edge[->] (v107);
  \path (v89) edge[->] (v12);
  \path (v89) edge[->] (v51);
  \path (v89) edge[->] (v90);
  \path (v89) edge[->] (v98);
  \path (v89) edge[->] (v107);
  \path (v90) edge[->] (v99);
  \path (v90) edge[->] (v108);
  \path (v91) edge[->] (v4);
  \path (v91) edge[->] (v43);
  \path (v91) edge[->] (v92);
  \path (v91) edge[->] (v93);
  \path (v91) edge[->] (v97);
  \path (v92) edge[->] (v5);
  \path (v92) edge[->] (v44);
  \path (v92) edge[->] (v93);
  \path (v92) edge[->] (v98);
  \path (v93) edge[->] (v99);
  \path (v94) edge[->] (v91);
  \path (v94) edge[->] (v95);
  \path (v94) edge[->] (v96);
  \path (v94) edge[->] (v97);
  \path (v95) edge[->] (v92);
  \path (v95) edge[->] (v96);
  \path (v95) edge[->] (v98);
  \path (v96) edge[->] (v93);
  \path (v96) edge[->] (v99);
  \path (v97) edge[->] (v11);
  \path (v97) edge[->] (v50);
  \path (v97) edge[->] (v98);
  \path (v97) edge[->] (v99);
  \path (v98) edge[->] (v12);
  \path (v98) edge[->] (v51);
  \path (v98) edge[->] (v99);
  \path (v100) edge[->] (v23);
  \path (v100) edge[->] (v62);
  \path (v100) edge[->] (v101);
  \path (v100) edge[->] (v102);
  \path (v100) edge[->] (v106);
  \path (v101) edge[->] (v24);
  \path (v101) edge[->] (v63);
  \path (v101) edge[->] (v102);
  \path (v101) edge[->] (v107);
  \path (v102) edge[->] (v108);
  \path (v103) edge[->] (v100);
  \path (v103) edge[->] (v104);
  \path (v103) edge[->] (v105);
  \path (v103) edge[->] (v106);
  \path (v104) edge[->] (v101);
  \path (v104) edge[->] (v105);
  \path (v104) edge[->] (v107);
  \path (v105) edge[->] (v102);
  \path (v105) edge[->] (v108);
  \path (v106) edge[->] (v29);
  \path (v106) edge[->] (v68);
  \path (v106) edge[->] (v107);
  \path (v106) edge[->] (v108);
  \path (v107) edge[->] (v30);
  \path (v107) edge[->] (v69);
  \path (v107) edge[->] (v108);
  \path (v109) edge[->] (v23);
  \path (v109) edge[->] (v62);
  \path (v109) edge[->] (v110);
  \path (v109) edge[->] (v111);
  \path (v109) edge[->] (v114);
  \path (v110) edge[->] (v24);
  \path (v110) edge[->] (v63);
  \path (v110) edge[->] (v111);
  \path (v110) edge[->] (v115);
  \path (v111) edge[->] (v116);
  \path (v112) edge[->] (v110);
  \path (v112) edge[->] (v113);
  \path (v112) edge[->] (v115);
  \path (v113) edge[->] (v111);
  \path (v113) edge[->] (v116);
  \path (v114) edge[->] (v29);
  \path (v114) edge[->] (v68);
  \path (v114) edge[->] (v115);
  \path (v114) edge[->] (v116);
  \path (v115) edge[->] (v30);
  \path (v115) edge[->] (v69);
  \path (v115) edge[->] (v116);
  \node[state,accepting,fill=black] at (3.4751484217817574,4.55031389406081) {};
  \node[state,accepting,fill=black] at (-4.135608504768771,-1.788929148574084) {};
  \node[state,fill=white] at (2.5661983373642028,-0.054550676247050865) {};
  \node[state,fill=white] at (2.2960064979553976,-0.8433926985344079) {};
  \node[state,fill=white] at (0.16248274031146165,1.786648833053601) {};
  \node[state,fill=white] at (-0.15403842532045073,1.5440941959252126) {};
  \node[state,fill=white] at (1.304412630534438,0.24959210370513552) {};
  \node[state,fill=white] at (-0.254561672247947,2.182077857105258) {};
  \node[state,fill=white] at (0.17203732818832979,2.3555946149164053) {};
  \node[state,fill=white] at (1.3996923141298372,0.9086729504972618) {};
  \node[state,fill=white] at (1.6020290554355545,-0.8648324999393118) {};
  \node[state,fill=white] at (-0.552154355514626,2.4391021190004674) {};
  \node[state,fill=white] at (-0.6973961754798681,1.8807830126933425) {};
  \node[state,fill=white] at (0.6520192928433094,0.44591629179692915) {};
  \node[state,fill=white] at (1.9070232370458056,2.689304845402798) {};
  \node[state,fill=white] at (1.2488960070372646,2.2247562734784028) {};
  \node[state,fill=white] at (2.661815852528711,0.9380111542120468) {};
  \node[state,fill=white] at (2.3505197522682204,3.5277773604112603) {};
  \node[state,fill=white] at (1.5664526660483733,3.1303400311863356) {};
  \node[state,fill=white] at (2.772958367441803,1.7178942713417453) {};
  \node[state,fill=white] at (1.1667166458346292,2.758762333129569) {};
  \node[state,fill=white] at (0.7100459088688089,2.1942193178406155) {};
  \node[state,fill=white] at (2.0277899926416456,1.0011252426628683) {};
  \node[state,fill=white] at (-1.339483077259822,0.27989229516804753) {};
  \node[state,fill=white] at (-0.8772538864964434,0.536683076153312) {};
  \node[state,fill=white] at (0.26442668464939034,-0.8251525960468068) {};
  \node[state,fill=white] at (-1.1887503535337274,0.9229365672083918) {};
  \node[state,fill=white] at (-0.5220645608338238,1.113583659319007) {};
  \node[state,fill=white] at (0.49970769928371683,-0.18835218571858217) {};
  \node[state,fill=white] at (-2.2033991623148776,-0.0965970082353895) {};
  \node[state,fill=white] at (-1.7770409660693134,0.564319902527252) {};
  \node[state,fill=white] at (-0.4609233675211441,-0.8611154275246178) {};
  \node[state,fill=white] at (0.5532793282134818,1.241257847853684) {};
  \node[state,fill=white] at (0.20963025253434922,0.8566729356953052) {};
  \node[state,fill=white] at (1.120024992052391,-0.4719491452269722) {};
  \node[state,fill=white] at (1.7993723129999715,1.7891708595447589) {};
  \node[state,fill=white] at (1.1506358393486438,1.4657206941192418) {};
  \node[state,fill=white] at (1.9259874103922248,-0.027611618119865754) {};
  \node[state,fill=white] at (-0.23573789648295948,0.21327848485417994) {};
  \node[state,fill=white] at (-0.564327191675239,-0.10099075591427249) {};
  \node[state,fill=white] at (0.7372891135010878,-1.2775846259696395) {};
  \node[state,fill=white] at (-1.3293196434637038,6.825497502439574) {};
  \node[state,fill=white] at (-1.3630277721388893,6.0371026696706975) {};
  \node[state,fill=white] at (-0.009805064791322347,3.2193195972644784) {};
  \node[state,fill=white] at (-0.35504279113692555,2.9305782121524557) {};
  \node[state,fill=white] at (-0.651669539959817,5.005092351126771) {};
  \node[state,fill=white] at (-0.7002944510123187,4.269292058962024) {};
  \node[state,fill=white] at (-0.13093864618508658,4.0082539930566) {};
  \node[state,fill=white] at (-0.6949983650431387,5.794852862726046) {};
  \node[state,fill=white] at (-2.117919726432933,6.20351206311362) {};
  \node[state,fill=white] at (-0.6272690994902363,3.741847551248912) {};
  \node[state,fill=white] at (-0.835456395648763,3.360834907797988) {};
  \node[state,fill=white] at (-1.3690202376417993,5.351162126862515) {};
  \node[state,fill=white] at (1.6779545179878366,4.925632693215907) {};
  \node[state,fill=white] at (0.771231706589196,4.26871006333689) {};
  \node[state,fill=white] at (0.3486980825664952,5.981038960355102) {};
  \node[state,fill=white] at (1.8415659900291574,5.621816136969723) {};
  \node[state,fill=white] at (0.84027403039665,4.993727240705263) {};
  \node[state,fill=white] at (0.3905073363260326,6.690665445802169) {};
  \node[state,fill=white] at (1.2619326187113395,6.459990553922406) {};
  \node[state,fill=white] at (0.17543782511852904,5.302361524234345) {};
  \node[state,fill=white] at (-0.37510981434760576,6.683699008727237) {};
  \node[state,fill=white] at (-1.7817404196988733,1.2561509790322165) {};
  \node[state,fill=white] at (-1.2612236672081574,1.6004460176638242) {};
  \node[state,fill=white] at (-1.904406825150545,3.7967751364142064) {};
  \node[state,fill=white] at (-2.11743550263423,2.4397611178958627) {};
  \node[state,fill=white] at (-1.3999694216719738,3.280686504771533) {};
  \node[state,fill=white] at (-2.0354296517343657,4.905352720781203) {};
  \node[state,fill=white] at (-2.6101018130148055,0.9730542191814799) {};
  \node[state,fill=white] at (-2.020355810288433,1.7370810263003533) {};
  \node[state,fill=white] at (-2.486983921441511,3.8774224320592063) {};
  \node[state,fill=white] at (3.4050519785962092,3.614477023449859) {};
  \node[state,fill=white] at (2.951951841598388,3.122268021278065) {};
  \node[state,fill=white] at (4.371534251111887,4.287617500246496) {};
  \node[state,fill=white] at (2.873068166128606,4.1721933979584716) {};
  \node[state,fill=white] at (4.410591385501697,5.152643917307501) {};
  \node[state,fill=white] at (4.99976881745918,3.981625148104324) {};
  \node[state,fill=white] at (4.466333710609381,3.519617162952977) {};
  \node[state,fill=white] at (5.467030126379853,4.628829211366017) {};
  \node[state,fill=white] at (-3.7279072761198657,5.000248792379955) {};
  \node[state,fill=white] at (-3.811215928042087,4.256826731203449) {};
  \node[state,fill=white] at (-1.1484351173644185,2.3102568936610397) {};
  \node[state,fill=white] at (-1.49330932123178,2.0808475166081846) {};
  \node[state,fill=white] at (-3.04196135531578,3.222060047967823) {};
  \node[state,fill=white] at (-1.798227264463192,2.9128373636056173) {};
  \node[state,fill=white] at (-1.3491116087899795,2.809316645968066) {};
  \node[state,fill=white] at (-3.0951374519019654,4.007893143989082) {};
  \node[state,fill=white] at (-4.600826872525891,4.447456218048963) {};
  \node[state,fill=white] at (-2.307749022090527,3.183685047046739) {};
  \node[state,fill=white] at (-2.515677050667505,2.694931513607954) {};
  \node[state,fill=white] at (-3.8137665848881666,3.6172165657123405) {};
  \node[state,fill=white] at (1.002213660425983,3.4655912741137294) {};
  \node[state,fill=white] at (0.3443864049944915,2.8873559381409506) {};
  \node[state,fill=white] at (-1.5906218982805276,4.224465426911279) {};
  \node[state,fill=white] at (1.4344629465622716,4.2064242330833865) {};
  \node[state,fill=white] at (0.40293336649773337,3.668404775626401) {};
  \node[state,fill=white] at (-1.3778565222868036,4.752603364273528) {};
  \node[state,fill=white] at (-0.2032210238671238,4.512173215707064) {};
  \node[state,fill=white] at (-1.2245017601519739,3.7507375123253857) {};
  \node[state,fill=white] at (-2.6222574501655442,4.755098060228369) {};
  \node[state,fill=white] at (-3.456375950023912,0.9823133057465543) {};
  \node[state,fill=white] at (-2.6861391224896876,1.6320204510600846) {};
  \node[state,fill=white] at (-4.171686255867115,2.5452878603747413) {};
  \node[state,fill=white] at (-4.846204597380111,1.7726356240289933) {};
  \node[state,fill=white] at (-3.48423624860414,2.5038693111160577) {};
  \node[state,fill=white] at (-4.6411671498241045,3.433416637619054) {};
  \node[state,fill=white] at (-4.140740048676738,0.9642121411754452) {};
  \node[state,fill=white] at (-3.4143045355450123,1.8075370736285366) {};
  \node[state,fill=white] at (-4.894307625170197,2.686889645805903) {};
  \node[state,fill=white] at (-2.978222856512475,-0.49678825604049753) {};
  \node[state,fill=white] at (-1.9923798135757813,-0.8898958751013494) {};
  \node[state,fill=white] at (-3.140303986211371,-1.883472969363521) {};
  \node[state,fill=white] at (-2.825624707200261,-2.4337682388261377) {};
  \node[state,fill=white] at (-3.865554251004385,-3.033553053509924) {};
  \node[state,fill=white] at (-3.621750590672103,-0.8287875853122821) {};
  \node[state,fill=white] at (-2.686550140289986,-0.9348164477976713) {};
  \node[state,fill=white] at (-3.749129320562069,-2.246746115270515) {};
\end{tikzpicture}

%% file: serial-frac.tex
\begingroup
 \def\fr#1#2{#1\!/\!#2}%
 \def\Fr#1#2{(\atom{Frac}~{#1}~{#2})}%
 \def\fwd#1{\overset\rightarrow{#1}}%
 \def\bwd#1{\overset\leftarrow{#1}}%
 \def\ruleLabel#1{\llap{\tiny$#1.$}}%
  \pgfmathsetseed{2736}
  \begin{tikzpicture}[
    every state/.style={inner sep=-0.5ex},
    decoration={random steps,amplitude=.375ex,segment length=.75ex,
                pre=lineto,pre length=1ex,post=lineto,post length=1ex},
    rounded corners=.15ex
  ]
    \long\def\hag(#1) [#2] #3{\node[state,accepting] (#1) [#2,decorate] {$\begin{array}{c}#3\end{array}$}}
    \long\def\bag(#1) [#2] #3{\node[state] (#1) [#2,decorate] {$\begin{array}{c}#3\end{array}$}}
    \hag (A) [] {\fr ab\\\,p~q};
    \bag (B) [right=of A] {\fr cd~g\\p~q};
    \bag (C1) [above=of B] {\fr cd~g\\p~q~q^2\!\!};
    \bag (D1) [right=of C1] {\fr cd~g\\p~q~g'\!\!};
    \bag (E1) [right=of D1] {\fr ab~g'\!\\p~q};
    \bag (C2) [below=of B] {\fr cd~g\\p'\!~q};
    \bag (D2) [right=of C2] {\fr ab~g'\!\\p'\!~q~g};
    \bag (E2) [right=of D2] {\fr ab~g'\!\\p'\!~q~q^2\!\!};
    \bag (F) [below=of E1] {\fr ab~g'\!\\p'\!~q};
    \bag (G) [right=of F] {\fr cd\\p'\!~q};
    \hag (H) [right=of G] {\fr cd\\\,p~q};
    \path (A) edge[->,above] node {2} (B)
          (B) edge[->,left] node {4} (C1)
              edge[->,left] node {1} (C2)
          (C1) edge[<-,above] node {5} (D1)
          (D1) edge[<-,above] node {2} (E1)
          (E1) edge[->,right] node {1} (F)
          (C2) edge[->,above] node {3} (D2)
          (D2) edge[->,above] node {6} (E2)
          (E2) edge[<-,right] node {4} (F)
          (F) edge[<-,above,shorten >=0.1pt] node {3} (G)
          (G) edge[<-,above,shorten >=0.8pt] node {1} (H);
  \end{tikzpicture}
\endgroup

%% file: serial-polish.tex
\begingroup%
\def\nl{\\[1.5em]}%
\def\nfx#1{\infix{\atom{#1}}}%
\def\infx#1#2{\infix{\atom{#1}~{#2}}}%
\def\ruleLabel#1{\llap{\tiny$#1.$}}%
\begin{sublisting}{\linewidth}
  \centering
  \begin{align*}
    & \dataDef{\atom{Tree}~{\alpha}} = \atom{Tree}~{\alpha}~{[\atom{Tree}~{\alpha}]} \\
    & \atom{polish}~{:\kern0.2ex:}~\atom{Tree}~{\alpha}~{\rightarrow}~{[(\atom{Int},{\alpha})]} \\
    & \atom{polish}~{(\atom{Tree}~{x}~\cursivexS)} = (\atom{length}~\cursivexS, x)~{:}~\atom{concatMap}~\atom{polish}~\cursivexS
  \end{align*}
  \caption{A snippet of \texttt{Haskell} code for converting an arbitrary tree into Polish notation. Polish notation is an isomorphic linear representation of tree-like structures that is well known for its use for arithmetic expressions. For example, the expression $(3 + 4) \times 7 - 2^5$ can be represented as ${-}~{\times}~{+}~{3}~{4}~{7}~{\raisebox{-0.5ex}{${}^\wedge$}}~{2}~{5}$.}
  \label{lst:serial-polish-hs}
\end{sublisting}
\begin{sublisting}{\linewidth}
  \centering
  \begin{minipage}[t]{.6\linewidth}
    \begin{align*}
      & (\atom{Tree}~{x}~\cursivetS)~\nfx{Polish}~{[({,}~{x}~{n})~{\cdot}~{p}]}{:} \\
      \ruleLabel{1}&\qquad \infx{Length}\cursivexS~{n}{.} \\
      \ruleLabel{2}&\qquad \cursivetS~\infx{ConcatMap}{\atom{Polish}}~{p}~\cursivelS{..} \\
      \ruleLabel{3}&\qquad {p}~\infx{PolishReads}{n}~\cursivetS~\cursivelS{.}
    \end{align*}\vspace{-1em}
  \end{minipage}%
  \begin{minipage}[t]{.3\linewidth}%
    \begin{equation*}\begin{array}{|rl}
      ~ & \dataDef{\atom{Tree}~{x}~\cursivetS}{;} \\
        & \dataDef{{,}~{a}~{b}}{;} \\
        &\quad \comment{\cmtSLopen~2-tuples}
    \end{array}\end{equation*}
  \end{minipage}
  \caption{An excerpt of a translation of the \texttt{Haskell} implementation into \alethe\ (see \Cref{lst:supp-polish} for the full program). Unfortunately, the implementation of \atom{ConcatMap} in \alethe\ has additional garbage in the form of a list of the lengths of the intermediate lists ($\cursivelS$). To eliminate this garbage, we write a function \atom{PolishReads} which takes a string of trees in Polish notation and converts them back into trees, also generating the same garbage value.}
  \label{lst:serial-polish-exc}
\end{sublisting}
\begin{sublisting}{\linewidth}
  \centering
  \vspace{2em}
  \pgfmathsetseed{18976}
  \begin{tikzpicture}[
    every state/.style={inner sep=-0.5ex},
    decoration={random steps,amplitude=.375ex,segment length=.75ex,
                pre=lineto,pre length=1ex,post=lineto,post length=1ex},
    rounded corners=.15ex
  ]
    \long\def\hag(#1) [#2] #3{\node[state,accepting] (#1) [#2,decorate] {$\begin{array}{c}#3\end{array}$}}
    \long\def\bag(#1) [#2] #3{\node[state] (#1) [#2,decorate] {$\begin{array}{c}#3\end{array}$}}
    \hag (A) [] {x\\\cursivetS};
    \bag (B) [right=of A] {x~n\\\cursivetS};
    \bag (C) [right=of B] {x~n\\p~\cursivelS};
    \bag (D) [right=of C] {x~n\\\cursivetS~\cursivelS};
    \hag (E) [right=of D] {x~n\\p};
    \path (A) edge[->,above,shorten >=0.3pt] node {1} (B)
          (B) edge[->,above,shorten >=0.5pt] node {2} (C)
          (C) edge[<-,above,bend left,shorten <=0.8pt,shorten >=0pt] node {2} (D)
              edge[->,below,bend right,shorten <=0.2pt] node {3} (D)
          (D) edge[<-,above,shorten >=0.8pt] node {3} (E);
  \end{tikzpicture}
  \vspace{1em}
  \def\fwd#1{$\vec{#1}$}%
  \def\bwd#1{\reflectbox{\ensuremath{\vec{\reflectbox{\ensuremath{#1}}}}}}%
  \def\en{$-$}%
  \caption{The motif employed above---generating the garbage in two different ways---can be used to eliminate it and achieve the desired (partial) bijection, as shown by its transition graph. There are two possible execution paths, \fwd1\en\fwd2\en\bwd2\en\bwd3 and \fwd1\en\fwd2\en\fwd3\en\bwd3. As explained in the text, the second route is \emph{strongly} preferred--- hence the cost annotation on \ruleName{2}.}
  \label{lst:serial-polish-tg}
\end{sublisting}
\begin{sublisting}{\linewidth}
  \centering
  \begin{minipage}[t]{.58\linewidth}
    \begin{align*}
      & {t}~\nfx{Polish}~{p}{:} \\
      &\qquad {p}~\nfx{PolishRead}~{[]}~{t}{.}
      \nl
      & {[{({,}~{x}~{n})}~{\cdot}~\cursivexS]}~\nfx{PolishRead}~{\cursivexS'}~{(\atom{Tree}~{x}~\cursivetS)}{:} \\
      &\qquad \cursivexS~{n}~\nfx{PolishReads}~{\cursivexS'}~\cursivetS{.}
    \end{align*}\vspace{-1em}
  \end{minipage}%
  \begin{minipage}[t]{.4\linewidth}
    \begin{align*}
      & \cursivexS~{\Z}~\nfx{PolishReads}~\cursivexS~{[]}{;} \\
      & \cursivexS~{(\S n)}~\nfx{PolishReads}~{\cursivexS''}~{[{t}~{\cdot}~\cursivetS]}{:} \\
      &\qquad \cursivexS~\nfx{PolishRead}~{\cursivexS'}~{t}{.} \\
      &\qquad {\cursivexS'}~{n}~\nfx{PolishReads}~{\cursivexS'}~\cursivetS{.}
    \end{align*}
  \end{minipage}
  \caption{In fact, programming in a `reversible-first' manner permits us to find the far more efficient and naturally bijective program above. This program was written by considering how to read a tree from its Polish representation: In our experience, picking the `harder' direction of a bijective computation to implement helps to suppress one's learned intuition for irreversible programming, and hence makes it less likely to fall into the traps of `shortcuts' with their accompanying garbage.}
  \label{lst:serial-polish-nice}
\end{sublisting}
\endgroup

%% file: part-alethe.tex
\section{A Spoonful of Sugar: \alethe}
\label{sec:alethe}

For clarity and convenience, we have introduced a number of syntactic shorthands in \Cref{sec:ex1,sec:ex2}. We now summarise and extend this sugar to construct a programming language, \alethe, the definition of which is given in \Cref{lst:alethe-dfn}. Implementing the additional measures discussed in \Cref{sec:impl}, a reference interpreter for \alethe\ is also made available\footnote{\url{https://github.com/hannah-earley/alethe-repl}} together with a `standard library' and select examples\footnote{\url{https://github.com/hannah-earley/alethe-examples}}. We proceed with the definition of \alethe\ by desugaring each form in turn.

\begin{listing}\centering
\fbox{\begin{minipage}{0.9\textwidth}\vspace{\baselineskip}\begin{equation*}\begin{aligned}
    \ruleName{pattern term} && \tau &::= \alpha ~|~ \bnfPrim{var} ~|~ (\,\tau^*\,) ~|~ \sigma \\
    \ruleName{atom} && \alpha &::= \bnfPrim{atom} ~|~ \sim\alpha ~|~ \#{}^{\backprime\backprime}\,\bnfPrim{char}^\ast\,'' ~|~ \#\alpha \\
    \ruleName{value} && \sigma &::= \mathbb{N} ~|~ {}^{\backprime\backprime}\, \bnfPrim{char}^\ast \,'' ~|~ [\,\tau^\ast\,] ~|~ [\, \tau^+ \,.\, \tau \,] ~|~ \blank \\
    \ruleName{party} && \pi &::= \tau:\tau^\ast ~|~ \bnfPrim{var}':\tau^\ast \\
    \ruleName{parties} && \Pi &::= \pi ~|~ \Pi\,;\,\pi \\
    \ruleName{relation} && \rho &::= \tau^\ast=\tau^\ast ~|~ \tau^\ast~{}^\backprime\tau^\ast{}^\prime~\tau^\ast \\
    \ruleName{definition head} && \delta_h &::= \rho ~|~ \{\,\Pi\,\}=\{\,\Pi\,\} \\
    \ruleName{definition rule} && \delta_r &::= \delta_h \,; ~|~ \delta_h\!: \Delta^+ \\
    \ruleName{definition halt} && \delta_t &::=~ !\,\tau^\ast; ~|~ !\,\rho\,; \\
    \ruleName{definition} && \delta &::= \delta_r ~|~ \delta_t \\
    \ruleName{cost annotation} && \xi &::= .^\ast \\
    \ruleName{declaration} && \Delta &::= \delta ~|~ \pi\,.\xi ~|~ \rho\,.\xi ~|~ !\,\rho\,.\xi \\
    \ruleName{statement} && \Sigma &::= \delta ~|~ \dataDef{\tau^\ast}; ~|~ \opn{import}~ {}^{\backprime\backprime}\bnfPrim{module path}''; \\
    \ruleName{program} && \mathcal P &::= \Sigma^\ast
\end{aligned}\end{equation*}\vspace{\baselineskip}\end{minipage}}
  \caption{Definition of \alethe\ syntax.}
  \label{lst:alethe-dfn}
\end{listing}

\begin{listing}
  \centering
  \input{alephex-sort2}
  \captionsetup{singlelinecheck=off}
  \def\cap{The \alethe\ implementation of insertion sort from the standard library, making use of nested locally scoped definitions. In contrast to the implementation in \Cref{lst:ex-sort}, this definition yields more useful `garbage' data in that $\cursivenS'$ contains the permutation which maps $\cursivexS$ to $\cursiveyS$, as shown in the boxed example. Specifically, this corresponds to the following permutation (in standard notation):}
  \caption[\cap]{\cap
    \begin{equation*}\begin{pmatrix}
      0 & 1 & 2 & 3 & 4 & 5 & 6 & 7 \\
      7 & 0 & 5 & 6 & 4 & 2 & 1 & 3
    \end{pmatrix}.\end{equation*}}
  \label{lst:ex-sort2}
\end{listing}

{\def\ir#1#2{\item[\ruleName{#1}] $#2$\\}\begin{itemize}
  \ir{patt.\ term}{\tau ::= \alpha ~|~ \bnfPrim{var} ~|~ (\,\tau^*\,) ~|~ \sigma}
    The definition of a pattern term only differs from its definition in \textAleph\ by the inclusion of sugared values, $\sigma$.
  \ir{atom}{\alpha ::= \bnfPrim{atom} ~|~ \sim\alpha ~|~ \#{}^{\backprime\backprime}\,\bnfPrim{char}^\ast\,'' ~|~ \#\alpha}
    Atoms are fundamentally the same as in \textAleph, but there are four ways of inputting an atom. If the name of an atom begins with an uppercase letter or a symbol, then it can be typed directly. In fact, any string of non-reserved and non-whitespace characters that does not begin with a lowercase letter (as defined by Unicode) qualifies as an atom; if it does begin with a lowercase letter, it qualifies as a variable. If one wishes to use an atom name that does not follow this rule, one can use the form $\#{}^{\backprime\backprime}\,\bnfPrim{char}^\ast\,''$, where $\bnfPrim{char}^\ast$ is any string (using \texttt{Haskell}-style character escapes if needed). Additionally you can prefix a symbol with $\#$ to suppress its interpretation as a relation (e.g.\ $\#+$). Finally, if an atom is prefixed with some number of tildes then it is locally scoped: that is, you can nest rule definitions and make these unavailable outside their scope, and you can refer to locally scoped atoms in outer scopes by using more tildes (there is no scope inheritance). You can even use an empty atom name if preceded by a tilde, which can be useful for introducing anonymous definitions. Example code using this sugar is given in \Cref{lst:ex-sort2}.
  \ir{value}{\sigma ::= \mathbb{N} ~|~ {}^{\backprime\backprime}\, \bnfPrim{char}^\ast \,'' ~|~ [\,\tau^\ast\,] ~|~ [\, \tau^+ \,.\, \tau \,] ~|~ \blank}
    The sugar here for natural numbers ($\mathbb N$) and lists is familiar from earlier, but there are two additional sugared value types: strings\footnote{Once again, using \texttt{Haskell}-style character escapes if needed}, which are lists of character atoms where the character atom for, e.g., \texttt{c} is \texttt{'c}, and units, which can be inserted with $\blank$ instead of $\unit$.
  \ir{party}{\pi ::= \tau:\tau^\ast ~|~ \bnfPrim{var}':\tau^\ast}
    Parties are the same as in raw \textAleph. Note, however, that the reference interpreter of \alethe\ has no separate representation for opaque variables: if the context is a variable, then it is assumed to be opaque rather than the variable pattern. When concurrency and contextual evaluation is supported in a future version, this deficiency will need to be addressed.
  \ir{parties}{\Pi ::= \pi ~|~ \Pi\,;\,\pi}
    Bags of parties are delimited by semicolons.
  \ir{relation}{\rho ::= \tau^\ast=\tau^\ast ~|~ \tau^\ast~{}^\backprime\tau^\ast{}^\prime~\tau^\ast}
    Relations are shorthand for where there is a single, variable context. These are familiar from the previous examples.
  \ir{def.\ head}{\delta_h ::= \rho ~|~ \{\,\Pi\,\}=\{\,\Pi\,\}}
    A rule is either a relation, or a mapping between party-bags which are enclosed in braces.
  \ir{def.\ rule}{\delta_r ::= \delta_h \,; ~|~ \delta_h\!: \Delta^+}
    If a rule has no sub-rules, then it is followed by a semicolon, otherwise it is followed by a colon and then its sub-rules/any locally scoped rules. These declarations must either be in the same line, or should be indented more than the rule head similar to \texttt{Haskell}'s off-side rule or \texttt{python}'s block syntax.
  \ir{def.\ halt}{\delta_t ::=~ !\,\tau^\ast; ~|~ !\,\rho\,;}
    As well as the canonical halting pattern form, there is also sugar for a relation wherein both sides of the relation each beget a halting pattern, as does the infix term (if any).
  \ir{def.}{\delta ::= \delta_r ~|~ \delta_t}
    This is the same as for \textAleph.
  \ir{cost ann.}{\xi ::= .^\ast}
    Sub-rules can be followed by more than one full-stop, with the number of full-stops used quantifying the cost of the sub-rule for use in the serialisation algorithm heuristics as explained towards the end of \Cref{sec:impl}.
  \ir{decl.}{\Delta ::= \delta ~|~ \pi\,.\xi ~|~ \rho\,.\xi ~|~ !\,\rho\,.\xi}
    In addition to the party sub-rules present in \textAleph, there are three sugared forms. If a declaration is a definition then, after introducing and resolving fresh anonymous names for any locally scoped atoms, the behaviour is the same as if the definition was written in the global scope. If a declaration takes the form of a relation, then it is the same as if we introduced a fresh context variable and bound each side of the relation to it. If the relation is preceded by $!$, then it defines both a sub-rule and halting pattern, e.g.
    \begin{align*}
      {!}~\atomLocal{Go}~{p}~\cursivexS~{[]}~{[]} =
          \atomLocal{Go}~{p}~\cursivexS~{[{n}~{\cdot}~\cursivenS]}~{\cursiveyS'}{.}
    \end{align*}
    becomes
    \begin{align*}
      & {!}~\atomLocal{Go}~{p}~{\cursivexS}~{[]}~{[]}{;} \\
      & {!}~\atomLocal{Go}~{p}~{\cursivexS}~{[{n}~{\cdot}~{\cursivenS}]}~{\cursiveyS'}{;} \\
      & \phantom{!}~\atomLocal{Go}~{p}~\cursivexS~{[]}~{[]} =
              \atomLocal{Go}~{p}~\cursivexS~{[{n}~{\cdot}~\cursivenS]}~{\cursiveyS'}{.}
    \end{align*}
  \ir{statement}{\Sigma ::= \delta ~|~ \dataDef{\tau^\ast}; ~|~ \opn{import}~ {}^{\backprime\backprime}\bnfPrim{module path}'';}
    A statement is a definition, a data definition, or an import statement. A data definition $\dataDef{\S~{n}}{;}$ is simply sugar for
    \begin{align*}
      & {!}~{\S}~{n}{;} \\
      & \infix{\atom{Dup}~({\S}~{n})}~{({\S}~{n'})}{:} \\
      &\qquad \infix{\atom{Dup}~{n}}~{n'}{.}
    \end{align*}
    i.e.\ it marks the pattern as halting and automatically writes a definition of \atom{Dup}. Future versions may automatically write other definitions, such as comparison functions, or include a derivation syntax akin to \texttt{Haskell}'s. The import statement imports all the definitions of the referenced file, as if they were one file, and supports mutual dependency. Future versions may support partial, qualified and renaming import variants.
  \ir{program}{\mathcal P ::= \Sigma^\ast}
    A program is a series of statements. Additionally, comments can be added in \texttt{Haskell}-style:
    \begin{align*}
      & \comment{\cmtSLopen~this is a comment.} \\
      & {!}~{\S}~\comment{\cmtMLopen~so is this...~\cmtMLclose}~{n}{;}~\comment{\cmtSLopen~...and this.}\\
      & \comment{\cmtMLopen~and this is}\\
      & \comment{\cmtMLcont~a multiline comment~\cmtMLclose}
    \end{align*}
\end{itemize}}

{\def\midtilde{\raisebox{-0.23em}{\textasciitilde}}
\noindent We also note that the interpreter for \texttt{alethe} treats the atoms \Z, \S, \Cons, \Nil, \atom{Garbage}, and character atoms specially when printing to the terminal. Specifically, they are automatically recognised and printed according to the sugared representations defined above. That is, except for the \atom{Garbage} atom which, when in the first position of a term, hides its contents---rendering as \texttt{\{\midtilde GARBAGE\midtilde\}}. This is useful when working with reversible simulations of irreversible functions that generate copious amounts of garbage data. If one assigns the garbage to a variable, then the special \texttt{:g} directive can be used to inspect its contents. Other directives include \texttt{:q} to quit, \texttt{:l file1 ...} to load the specified files as the current program, \texttt{:r} to reload the current program, \texttt{:v} to list the currently assigned variables after the course of the interpreter session, and \texttt{:p} to show all the loaded rules of the current program (and the derived serialisation strategy for each, see later). Computations can be performed in one of two ways; the first, \texttt{| (+ 3) 4 \raisebox{-0.3em}{-}}, takes as input a halting term, attempts to evaluate it to completion, and returns the output if successful, i.e.\ \texttt{() 7 (+ 3)}. The second, \texttt{> 4 `+ 3` y}, takes a relation as input and attempts to get from one side to the other. For example, here we evaluate \texttt{(+ 3) 4 ()}, obtaining \texttt{() 7 (+ 3)} which we then unify against \texttt{() y (+ 3)}, finally resulting in the variable assignment $y\mapsto7$. To run in the opposite direction, use \texttt{<} instead. Note that, whilst the interpreter does understand concurrent rules, it is not yet able to evaluate them.}

\para{Unification Efficiency}

Beyond the implementation concerns expressed in \Cref{sec:impl}, another issue of practical importance is the efficiency of identifying patterns matching a given term: to date, the standard library alone contains just shy of two thousand patterns, and so searching the known patterns linearly is impractical. Our interpreter constructs a trie-like structure for this purpose such that the time complexity for unification scales as $\bigOO{m+\log n}$ where $m$ is the number of patterns which match and $n$ the total number of patterns. Though asymptotically this is the best complexity possible, there are certainly improvements that could be made to the unification algorithm. In particular, the data structures used are primarily binary trees; hash tables with fixed lookup time for reasonable values of $n$ would yield significant improvements in distinguishing atoms, and random access arrays would yield significant improvements for distinguishing between composite terms of differing length. Notwithstanding these possible improvements, the reference interpreter has proven sufficiently fast for demonstrating non-trivial computations in \alethe. For example, computing $8! = \num{40320}$ takes only a few seconds despite the unary representation of natural numbers.

%% file: alephex-sort2.tex
\begingroup%
\def\nl{\\[1.5em]}%
\begin{align*}
  & \cursivexS~\infix{\atom{InsertionSort}~{p}}~{\cursivenS'}~\cursiveyS{:} \\
  &\qquad \phantom{!}~\cursivenS~\infix{\atom{Reverse}}~{\cursivenS'}{.} \\
  &\qquad {!}~\atomLocal{Go}~{p}~\cursivexS~{[]}~{[]} =
              \atomLocal{Go}~{p}~\cursivexS~{[{n}~{\cdot}~\cursivenS]}~{\cursiveyS'}{.} \\
  &\qquad \phantom{!}~
          \atomLocal{Go}~{p}~{[{x}~{\cdot}~\cursivexS]}~\cursivenS~\cursiveyS = 
          \atomLocal{Go}~{p}~\cursivexS~{[{n}~{\cdot}~{\cursivenS'}]}~{\cursiveyS'}{:} \\
  &\qquad\qquad {x}~\cursiveyS~\infix{\atom{Insert}~{p}}~{n}~{\cursiveyS'}{.} \\
  &\qquad\qquad \cursivenS~\infix{\atom{Map}~{(\atomLocal{Go}~{n})}}~{\cursivenS'}{.} \\
  &\qquad\qquad {m}~\infix{\atomLocal{Go}~{n}}~{m'}{:} \\
  &\qquad\qquad\qquad \infix{{<}~{m}~{n}}~{b}{.} \\
  &\qquad\qquad\qquad \infix{{<}~{m'}~{n}}~{b}{.} \\
  &\qquad\qquad\qquad {m}~\infix{\atomLocal[2]{}~{b}}~{m'}{.} \\
  &\qquad\qquad {m}~\infix{\atomLocal{}~\atom{True}}~{m}{;} \\
  &\qquad\qquad {m}~\infix{\atomLocal{}~\atom{False}}~{(\S{m})}{;}
  \nl
  & \cursivexS~\infix{\atom{Reverse}}~\cursiveyS{:} \\
  &\qquad {!}~\atomLocal{Go}~\cursivexS~{[]} = \atomLocal{Go}~{[]}~\cursiveyS{.} \\
  &\qquad\phantom{!}~\atomLocal{Go}~{[{x}~{\cdot}~\cursivexS]}~\cursiveyS =
                     \atomLocal{Go}~\cursivexS~{[{x}~{\cdot}~\cursiveyS]}{;}
\end{align*}
\fbox{\begin{minipage}{.95\textwidth}\vspace{-0.2em}\begin{align*}
  & {[77~2~42~68~41~36~8~36]}~\infix{\atom{InsertionSort}~{<}}~{[7~0~5~6~4~2~1~3]}~{[2~8~36~36~41~42~68~77]}{.}
\end{align*}\vspace{-0.8em}\end{minipage}}%
\endgroup

%% file: part-type.tex
\begingroup
\def\isot{\mathrel{\ensuremath{:\kern0.11ex:}}}
\section{Towards a Type System}
\label{app:typing}

Type systems are very useful, in that they allow a compiler to statically analyse a program before it is run and identify whether it is ever possible for a value to be passed to a function that does not understand such an input. This is extremely powerful, and leads to the \texttt{Haskell} maxim `if it compiles, it runs': that is, although there may be logical errors, the program should not crash.

We may be tempted to try to start with a simple polymorphic type system, such as Hindley-Milner, but unfortunately Hindley-Milner is unsuitable for \textAleph\ being that it does not have objects that can be uniquely assigned a computational role. Consider attempting to type the symbol $\square$ from \Cref{lst:ex-square-def}, for example, which implements squaring/square rooting for natural numbers. One is tempted to say something like $\square\isot\mathbb{N}~\unit\leftrightarrow\unit~\mathbb{N}\isot\square$, but this does not work. Inspecting the definition more closely, we see that $\square$ appears in a number of contexts: ${\square}~{n}~{\unit}$, ${\square}~{n}~{m}~{\square}$, and ${\unit}~{m}~{\square}$. It is the second context where our attempt to type $\square$ breaks down: not only does its apparent arity change, but it appears twice in the same expression! Similar problems are encountered when trying to type other declarative languages, such as \texttt{Prolog}. To proceed, we shall require that any candidate type system treats terms \emph{holistically}, not ascribing any special importance to any particular part of a term. Such a holistic type system will most likely also be declarative in nature.

A fruitful way of thinking of types of terms is by considering their endpoints---their halting states. A well-formed term will be able to eventually transition to one or more halting terms, and the properties of these halting terms are what we are most interested in. For example, we would like to know that if we construct a term from $+$, two natural numbers, and $\unit$, then we will be able to extract two natural numbers---rather than, say, a boolean and a string---once computation finishes. But, the `end' state is free to wander back to the `beginning', and so it is more appropriate to consider the type of the term as an unordered pair, consisting of the two possible halting states. Checking this type will then consist of enumerating which rules may be encountered along a computational path, verifying that all these rules are consistent with the initial type and that the other state is reachable in principle. If the ambiguity checker is turned off, then we see that it is possible to have a term with more than two halting states and so its type should be a set of types corresponding to each possible halting state. We then notice that conventional data---such as a natural number---is of the same ilk, save that it only has one halting state. We call these type-sets \emph{isotypes}, and their inhabitants \emph{isovalues}, as each isovalue has one or more isomorphically equivalent representations.

How might these isotypes look? We begin with conventional data; let us call the isotype of a natural number \atom{Nat}, with inhabitants $\Z$ and $\S n$ where $n$ inhabits \atom{Nat}. Let us now attempt to type addition. As a first approximation, we write $\isot{+}~\atom{Nat}~\atom{Nat}~{\unit} = {\unit}~\atom{Nat}~\atom{Nat}~{+}{;}$ where $\isot$ at the beginning marks this as an isotype. This is in the right spirit, but doesn't readily admit a consistent type theory. In particular, how would we represent the isotype of a continuation in the term $\atom{Walk}~{\beta}~{(\atom{Charlie}'~x)}$, as introduced towards the end of \Cref{sec:ex2}? Being more careful, we can say that an addition isovalue is the unordered pair consisting of ${+}~{a}~{b}~{\unit}$ and ${\unit}~{c}~{d}~{+}$ where $a$, $b$, $c$ and $d$ are all isovalues of isotype \atom{Nat}. We need this explication, because the isotypes of the arguments could be far more complicated. Where it is unambiguous, however, it will be reasonable to write this isotype as $\isot{+}~\atom{Nat}~\atom{Nat}~{\unit} = {\unit}~\atom{Nat}~\atom{Nat}~{+}{;}$. Notice that, as this isotype is holistic, it is completely orthogonal to the isotype $\isot{+}~\atom{Rat}~\atom{Rat}~{\unit} = {\unit}~\atom{Rat}~\atom{Rat}~{+}{;}$ where \atom{Rat} is the isotype of rationals. This reveals that we get overloading for free in this type system. It is nevertheless desirable to introduce a notion similar to \texttt{Haskell}'s typeclasses, however, or else (partial) isomorphisms employing addition, but polymorphic in the isotype they are adding, will have much too generic an isotype (in particular, the type inference algorithm would not be able to conclude that an isotype $\isot{+}~{a}~{b}~{\unit} = {\unit}~{c}~{d}~{+}{;}$ does not exist, and so we would have four isotype variables rather than the 1 expected). Typeclasses permit the programmer to specify that an isotype of the form $\isot{+}~{a}~{b}~{\unit}={\unit}~{c}~{d}~{+}$ is only legal if it is of the more restrictive form $\isot{+}~{a}~{a}~{\unit}={\unit}~{a}~{a}~{+}$ for some isotype $a$. Furthermore, the introduction of typeclasses permits the abbreviation of complex (ad-hoc) polymorphic isotypes.

What, then, might a polymorphic isotype look like? A good example is given by the map function, which would have as isotype the unordered pair of $(\atom{Map}~{f})~{\cursivexS}~{\unit}$ and ${\unit}~{\cursiveyS}~(\atom{Map}~{f})$ where $\cursivexS$ is an inhabitant of the isotype $[a]$, i.e.\ the isotype of lists of elements of isotype $a$, and where $\cursiveyS$ is an inhabitant of $[b]$. $f$ is required to be some term which, when used to construct the halting states ${f}~{x}~{\unit}$ and ${\unit}~{y}~{f}$, yields the isotype $\isot{f}~{a}~{\unit} = {\unit}~{b}~{f}$. More concisely, we write this as
\begin{align*}
  \isot{}& {[a]}~\infix{\atom{Map}~{f}}~{[b]}{:} \\
  \isot{}&\qquad {a}~\infix{f}~{b}{.} \\
  & []~\infix{\atom{Map}~{f}}~[]{;} \\
  & [{x}~{\cdot}~\cursivexS]~\infix{\atom{Map}~{f}}~[{y}~{\cdot}~\cursiveyS]{:} \\
  &\qquad {x}~\infix{f}~{y}{.} \\
  &\qquad \cursivexS~\infix{\atom{Map}~{f}}~\cursiveyS{.}
\end{align*}
and we note the similarity of the isotype to the implementation of \atom{Map} at the isovalue level. The composite term $\atom{Map}~{f}$ deserves further attention; this is itself halting, and so can be considered to be an implicitly declared anonymous isotype with one representation.

As briefly mentioned earlier, a type checker (and a type inference algorithm, if it exists) will perform a reachability analysis---starting from one halting pattern, it will enumerate all the accessible rules, finding the most general type consistent with these. The information gleaned this way is very valuable for all sorts of static analyses. We can determine what other halting states are accessible, as well as whether there are any missing cases along the way that could lead to a halting state. We can also resolve the issue of sub-rule ambiguity raised earlier, helping strengthen the phase space restrictions and avoid unintentional generation of entropy. It may also be possible to use this information to refine the cost heuristic used for serialisation, improving runtimes. Furthermore, knowledge of the execution path gives opportunities for identifying code motifs and performing structural transformations, a necessary ingredient for performing the kind of code optimisations a compiler is wont to do.

There is, however, a potential hiccup in the reachability analysis that occurs when a continuation-passing-style is used. In this case, there are rules which are shared between different computations. For example, in $\atom{Walk}~{\beta}~{c}\leftrightarrow\atom{Walk'}~{\alpha}~{c}$, $c$ is a continuation. There will be a number of computations which each pass control to such a \atom{Walk} term, and later accept control back from such a \atom{Walk'} term. Consequently, no definitive isotype can be assigned to these intermediate terms. Moreover, consider these two potential computations:
\begin{align*}
  & \atom{Foo} = \atom{Walk}~{\varphi}~\atom{Foo'}{;} \\
  & \atom{Bar} = \atom{Walk}~{\beta}~\atom{Bar'}{;}
\end{align*}
It is conceivable that the reachability analyser would, starting from \atom{Foo}, (correctly) determine that the pattern $\atom{Walk}~{\beta}~{c}$ is reachable, and then (incorrectly) presume the term \atom{Bar} is reachable from \atom{Foo}. If so, the type checker would then (rightly) complain of a type error, perhaps that the isotypes of $\atom{Foo'}$ and $\atom{Bar'}$ do not unify. The remedy to the first issue, of indefinite isotyping of intermediate terms, is to only assign definitive isotypes to sets of halting terms; intermediate terms are assigned definite isotypes only in the course of inferring isotypes for particular halting sets, and these intermediate isotypes are appropriately scoped such that the parallel process of isotype inference may proceed in spite of conflicting intermediate term isotypes. The second issue is remedied by specialising types as early as possible: in order to pass the ambiguity checker, any computation making use of a shared rule as above must pass control to the shared rule in a way orthogonal to any other computations passing control. In the above example, this holds because \atom{Foo'} and \atom{Bar'} are orthogonal patterns such that no term can be constructed satisfying them both simultaneously. Therefore, by immediately specialising the type of the continuation $c$ in \atom{Walk}, there is no risk of reachability `leaking' into inaccessible rules. One may be concerned that perhaps the reachability graph is much more complicated, engendering other reachability leaks, but each such confluence point between distinct computational paths \emph{must} have this same explicit orthogonality property in order to satisfy the ambiguity checker, and so there is in fact no risk. However, this remedy may be somewhat overzealous; suppose that the example above is part of a larger program,
\begin{align*}
  & \atom{Qux}~\atom{False} = \atom{Foo}{;} \\
  & \atom{Qux}~\atom{True} = \atom{Bar}{;}
\end{align*}
then the eager specialisation of the continuation isotype will yet result in a type-checking error, similar to \texttt{Haskell}'s (optional) monomorphism restriction. To circumvent such monomorphism, it would be advantageous to allow the type checker to introduce additional scopes for intermediate types where such branches occur, providing that these alternate paths eventually converge to a single type. If this is not possible, then the solution may simply be to expect the programmer to annotate a `partial isotype' of any shared polymorphic rules in order to give the type checker sufficient information.

We are optimistic that, in programs without shared rules, it is possible to perform automatic type inference without any type annotations as in the Hindley-Milner type system. As for programs with shared rules, we are less certain; nonetheless, the power and generality of type inference algorithms in type systems as sophisticated as \texttt{Haskell}'s most recent iterations suggests that it is attainable. We also note that the above is just a sketch of a possible type system, and will require further refinement and formalisation. Finally, we have neglected as yet to consider concurrency. Concurrency severely undermines the power of the type system, as there is no longer any continuity between terms: at any point, a term could be consumed, produced, or merged with another. It may be possible to extend the notion of isotype to include all the terms that might interact with one another, but it is unclear how useful this would be. Alternatively, one could simply ascribe types to single halting patterns, instead of seeking isotypes and an enumeration of all possible isomorphic representations of an isovalue. Again, this severely weakens the utility of the type system. A desirable compromise would be to identify the non-concurrent parts of a program, and apply type checking and inference to just these parts; a more rudimentary level of type checking could then be applied to the concurrent parts to at least verify the consistency of any non-concurrent sub-rules employed.

\endgroup

%% file: part-sigma.tex
\section{The \texorpdfstring{\boldmath$\eulerSigma$}{Sigma} Calculus}
\label{app:sigma}

An earlier prototype, $\Sigma$ was more functional in nature. Its terms were nested applications of `general permutations', which could be named for convenience (and for achieving recursion without a fixed point combinator). Its definition was
\begin{equation*}\begin{aligned}
  \ruleName{term} && \tau &::= \langle\pi \tau^\ast : \tau^\ast \pi\rangle ~|~ \bnfPrim{var} ~|~ \bnfPrim{ref} ~|~ \tau* \\
\end{aligned}\end{equation*}
where \bnfPrim{ref} is a reference to a named pattern. A generalised permutation is given by $\langle\pi \tau^\ast : \tau^\ast \pi\rangle$, where $\pi$ is a special variable that matches the permutation itself. That is, the term $(\langle\pi~x~y : y~x~\pi\rangle~1~2)$ would transition to $(2~1~\langle\pi~x~y : y~x~\pi\rangle)$. Notice that this is the convention taken in \textAleph, that the atom or term in the first position tends to take the `active' role, and tends to place itself at the end of the term after rule execution. In $\Sigma$, however, this convention is enshrined in the language itself because all the information concerning the transition rules is held within the permutation terms themselves, and so it is required that there be two special locations within a term corresponding to the current and previous permutation. The first position is used for the current permutation in analogy to the $\lambda$ calculus, and the last position is used for the previous permutation for symmetry. Permutations are `general' in the sense that they can alter tree structure, and can copy/elide variables.

The $\Sigma$ calculus can also be proven to be microscopically reversible and Reverse-Turing complete, but the absence of sub-rules renders composition more unwieldy. Additionally, discriminating different cases is tricky. As with the $\lambda$ calculus, where data can be represented with functions via Church encoding, we can represent data in $\Sigma$ with permutations and these permutations will perform case statements. Taking lists as an example, we have
\begin{align*}
  \atom{cons} &\equiv \langle \pi \{nc\}zf : c\{nf\}z\pi \rangle \\
  \atom{nil} &\equiv \langle \pi \{nc\}zf : n\{fc\}z\pi \rangle
\end{align*}
where $\{xyz\}$ is sugar for $(\bot xyz\top)$ used to represent an inert expression ($\top\equiv\bot\equiv\varepsilon\equiv\unit$, not being a permutation, is an inert object that does nothing and is used to represent halting states in $\Sigma$). Using these definitions, list reversal can be implemented thus,
\begin{align*}
  \atom{rev} &\equiv \langle\pi l \top : \atom{nil} \{\lambda\nu\}\{\{\varepsilon\varepsilon\}l\} \pi\rangle \\
  \lambda &\equiv \langle\pi \{\atom{rev}\,\nu\}\{\{r'r''\}\{\tilde ll'l''\}\} \tilde r : \tilde l \{\atom{rev}'\,\nu\}\{\{\tilde rr'r''\}\{l'l''\}\} \pi\rangle \\
  \nu &\equiv \langle\pi \{\atom{rev}'\,\lambda\}\{r\{l'l''\}\}\atom{cons} : \atom{cons}\{\atom{rev}\,\lambda\}\{\{l'r\}l''\} \pi\rangle \\
  \atom{rev}' &\equiv \langle\pi \{\lambda\nu\}\{r\{\varepsilon\varepsilon\}\}\atom{nil} : \bot r\pi\rangle
\end{align*}
which is certainly not the most edifying program in the world! It would be used as so:
\begin{align*}
  (\atom{rev}[1\,4\,6\,2]\top) &\overset\ast\longleftrightarrow (\bot[2\,6\,4\,1]\atom{rev}')
\end{align*}

\para{\texorpdfstring{\boldmath$\eulermu$}{Mu}-Recursive Functions}

To give a deeper flavour of $\Sigma$, we implement the $\mu$-recursive functions (recall their definition and \alethe\ implementation from \Cref{lst:aleph-murec}). In order to simulate a notion of function composition, we first adopt a `$\Sigma$-function' motif: a function is represented by the triple $F=\{f~z~f'\}$ where $f$ and $f'$ are the initial and terminal permutations, and $z$ is any data we wish to bind to $F$. The permutations $f$ and $f'$ must be such that $(f~z~f'~x~\top)\overset\ast\longleftrightarrow(\bot~y~g~f~z~f')$, where $x$ is the input, $y$ the output and $g$ any garbage data.

The base $\mu$-recursive functions can then be implemented thus:
\begin{align*}
  Z&\equiv \{z\varepsilon z\} & z&\equiv\langle z\varepsilon z~n~\top : \bot~0~n~z\varepsilon z \rangle \\
  S&\equiv \{s\varepsilon s\} & s&\equiv\langle s\varepsilon s~n~\top : \bot~\{\atom{succ} n\}~\varepsilon~s\varepsilon s\rangle \\
  \Pi_i^k&\equiv \{\pi_i^k\varepsilon\pi_i^k\} & \pi_i^k&\equiv\langle\pi_i^k\varepsilon\pi_i^k~\{n_1,\dots,n_k\}~\top : \\
  &&&\qquad\bot~n_i~\{n_1,\dots,n_{i-1},n_{i+1},\dots,n_k\}~\pi_i^k\varepsilon\pi_i^k\rangle
\end{align*}
Note that we have written $\langle z\varepsilon z\cdots$ instead of $\langle\pi\varepsilon z\cdots$ or $\langle\pi\varepsilon\pi\cdots$ for clarity.

The composition operator for $F$ of arity $k$ over functions $G_i$ is given by $\{c_k\{F\vec G\}c_k'\}$, where
\begin{align*}
  c_k &\equiv \langle c_k\{F\vec G\}c_k'~\vec n~\top : c_k''~F~(G_1\wedge\vec n)~\cdots~(G_k\wedge\vec n)~c_k\rangle \\
  c_k'' &\equiv \langle c_k''~F~(m_1~h_1\vee G_1)~\cdots~(m_k~h_k\vee G_k)~c_k : c_k'~\vec h~\vec G~(F\wedge\vec m)~c_k'' \rangle \\
  c_k' &\equiv \langle c_k'~\vec h~\vec G~(y~h_f\vee F)~c_k'' : \bot~y~\{h_f\,\vec h\}~c_k\{F\vec G\}c_k' \rangle
\end{align*}
where we have introduced sugared forms $(F\wedge x)\equiv(fzf'~x~\top)$ and $(y~h\vee F)\equiv(\bot~y~h~fzf')$.

The primitive recursion operator is a little more complicated. To understand the implementation below, it is important to know that the number 0 is represented as $\{\atom{zero}~\varepsilon\}$ and the number $m+1$ is given by $\{\atom{succ}~m\}$. When the form is unknown, we can write a unified representation as $\{\tilde mm\}$ where $\tilde m$ is the constructor and $m$ is either $\varepsilon$ or the predecessor. Furthermore, the constructors are defined as follows:
\begin{align*}
  \atom{zero} &\equiv \langle\pi \{pq\} z r : p\{rq\} z \pi\rangle \\
  \atom{succ} &\equiv \langle\pi \{pq\} z r : q\{pr\} z \pi\rangle
\end{align*}
We write the primitive recursion of $F$ and $G$ (arities $k$ and $k+2$) as $R_{FG}=\{\rho_k\{FG\}\rho_k'\}$, defined by:
\begin{align*}
  \rho_k &= \langle\rho_k\{FG\}\rho_k'~\{\vec n\{\tilde mm\}\}~\top : \tilde m~\{\zeta_k\,\sigma_k\}~\{FG\vec nm\}~\rho_k \rangle \\
  \rho_k' &= \langle \rho_k'~\{\zeta_k'\,\sigma_k''\}~\{yhFG\}~\tilde m : \bot~y~\{\tilde mh\}~\rho_k\{FG\}\rho_k' \rangle %\\
  \\[0.4\baselineskip]
  \zeta_k &= \langle\zeta_k~\{\rho_k\,\sigma_k\}~\{FG\vec n\varepsilon\}~\atom{zero} : \zeta_k'~(F\wedge\vec n)~G~\zeta_k\rangle \\
  \zeta_k' &= \langle\zeta_k'~(y~h\vee F)~G~\zeta_k : \atom{zero}~\{\rho_k'\,\sigma_k''\}~\{yhFG\}~\zeta_k'\rangle
  \\[0.4\baselineskip]
  \sigma_k &= \langle\sigma_k~\{\zeta_k\,\rho_k\}~\{FG\vec nm\}~\atom{succ} : \sigma_k'~(R_{FG}\wedge\vec nm)~\vec nm~\sigma_k \rangle \\
  \sigma_k' &= \langle\sigma_k'~(y~h\vee R_{FG})~\vec nm~\sigma_k : \sigma_k''~(G\wedge\vec nmy)~F~h~\sigma_k'\rangle \\
  \sigma_k'' &= \langle\sigma_k''~(z~h'\vee G)~F~h~\sigma_k' : \atom{succ}~\{\zeta_k'\,\rho_k'\}~\{z\{hh'\}FG\}~\sigma_k''\rangle
\end{align*}
In the above, the $\rho$ permutation switches between the $\zeta$ and $\sigma$ permutations depending on if $\{\tilde mm\}=0$ or $>0$; the $\zeta$ route computes $y=F(\vec n)$, whilst the $\sigma$ route recurses down to compute $y=R_{FG}(\vec n,m)$ and then $z=G(\vec n,m,y)$. Both these routes then put their results into a common representation and merge the control flow via $\tilde m$.

Finally, we implement the minimalisation operator as $M_F^m=\{\mu_k\{Fm\}\mu_k'\}$; $M_F^m$ is a generalisation of the more canonical minimalisation operator (which is equivalent to $M_F^0$),
\[ M_F^m(\vec n) = \min\{ \ell\ge m : f(\vec n; \ell)=0 \} \]
Interestingly, this is much simpler to implement than primitive recursion---only requiring four permutations---yet is the necessary ingredient for universality:
\begin{align*}
  \mu_k &= \langle\mu_k\{Fm\}\mu_k'~\vec n~\top : \mu_k''~m~(F\wedge \vec nm)~(M_F^{\{\atom{succ}\,m\}}\wedge\vec n)~\mu_k\rangle \\
  \mu_k'' &= \langle\mu_k''~m~(\{\tilde yy\}h\vee F)~q~\mu_k : \tilde y~\{\mu_k'\mu_k'''\}~m~(\{\tilde yy\}h\vee F)~q~\mu_k''\rangle \\
  \mu_k''' &= \langle\mu_k'''~\{\mu_k'\mu_k'''\}~m~p~(m'h\vee M_F^{\{\atom{succ}\,m\}})~\atom{succ} : \\
  &\qquad \atom{succ}~\{\mu_k''\mu_k'\}~m'~p~(m'h\vee M_F^{\{\atom{succ}\,m\}})~\mu_k'''\rangle \\
  \mu_k' &= \langle\mu_k'~\{\mu_k''\mu_k'''\}~m'~(F\wedge\vec nm)~(M_F^{\{\atom{succ}\,m\}}\wedge\vec n)~\tilde\mu : \bot~m'~\{\tilde\mu\vec n\}~\mu_k\{Fm\}\mu_k'\rangle
\end{align*}
In the above, we compute $F(\vec n;m)$ and check if it's 0, if so we return $m$ as $m'$, otherwise we compute $M_F^{m+1}(\vec n)$ in a recursive manner and return its result as $m'$. Clearly, if there is no $m$ such that $F(\vec n;m)=0$ then $M_F^m(\vec n)=\bot~\forall m$. For concision, we compute both functions simultaneously, making use of $\Sigma$'s parallelism to approximate laziness. We then conveniently reverse both computations in a Bennett-like manner to produce only a single bit of garbage (the history data we return is $\tilde\mu$, representing whether or not $m=m'$, and the input vector $\vec n$).

\para{The Interpreter} Whilst the implementation of the $\mu$-recursive functions demonstrates that realising meaningful programs in $\Sigma$ is certainly possible, it is inelegant and clunky. Rather than programming in $\Sigma$ feeling like the $\lambda$ calculus, as was intended, it feels much more like programming in assembly. Nevertheless, we release the source code of the interpreter\footnote{\url{https://github.com/hannah-earley/sigma-repl}}, example code\footnote{\url{https://github.com/hannah-earley/sigma-examples}}, and syntax highlighters\footnote{\url{https://github.com/hannah-earley/sigma-syntax}}\ for completeness.

%% file: polish.tex
\begingroup%
\def\nl{\\[1.5em]}%
\def\nfx#1{\infix{\atom{#1}}}%
\def\infx#1#2{\infix{\atom{#1}~{#2}}}%
\begin{minipage}[t]{0.45\linewidth}\begin{align*}
  & \infx{Length}{[]}~{\Z}{;} \\
  & \infx{Length}{[{x}~{\cdot}~\cursivexS]}~{(\S\ell)}{:} \\
  &\qquad \infx{Length}\cursivexS~{\ell}{.}
  \nl
  & \infx{Sum}{[]}~{\Z}{;} \\
  & \infx{Sum}{[\Z~{\cdot}~\cursivenS]}~{\sigma}{:} \\
  &\qquad \infx{Sum}\cursivenS~{\sigma}{.} \\
  & \infx{Sum}{[{(\S n)}~{\cdot}~\cursivenS]}~{(\S\sigma)}{:} \\
  &\qquad \infx{Sum}{[{n}~{\cdot}~\cursivenS]}~{\sigma}{.}
  \nl
  & {[]}~\infx{Map}{f}~{[]}{;} \\
  & {[{x}~{\cdot}~\cursivexS]}~\infx{Map}{f}~{[{y}~{\cdot}~\cursiveyS]}{;} \\
  &\qquad {x}~\infix{f}~{y}{.} \\
  &\qquad \cursivexS~\infx{Map}{f}~\cursiveyS{.}
  \nl
  & \infx{Map}{{f}~{[]}}~{[]}{;} \\
  & \infx{Map}{{f}~{[{x}~{\cdot}~\cursivexS]}}~{[{y}~{\cdot}~\cursiveyS]}{;} \\
  &\qquad \infix{f~x}~{y}{.} \\
  &\qquad \infx{Map}{f~\cursivexS}~\cursiveyS{.}
  \nl
  & {[]}~\nfx{Concat}~{[]}~{[]}{;} \\
  & {[{[]}~{\cdot}~\cursivexSS]}~\nfx{Concat}~\cursiveyS~{[{\Z}~{\cdot}~\cursivelS]}{:} \\
  &\qquad \cursivexSS~\nfx{Concat}~\cursiveyS~\cursivelS{.} \\
  & {[{[x~{\cdot}~\cursivexS]}~{\cdot}~\cursivexSS]}~\nfx{Concat}~{{x}~{\cdot}~\cursiveyS}~{[{(\S\ell)}~{\cdot}~\cursivelS]}{:} \\
  &\qquad {[\cursivexS~{\cdot}~\cursivexSS]}~\nfx{Concat}~\cursiveyS~{[{\ell}~{\cdot}~\cursivelS]}{.}
  \nl
  & \cursivexS~\infx{ConcatMap}{f}~\cursiveyS~\cursivelS{:} \\
  &\qquad \cursivexS~\infx{Map}{f}~\cursiveySS{.} \\
  &\qquad \cursiveySS~\nfx{Concat}~\cursiveyS~\cursivelS{.}
\end{align*}\vspace{-1em}\end{minipage}%
\quad\def\nl{\\[1.805em]}% stretch the right column to match the height of the left...
\begin{minipage}[t]{.45\linewidth}\begin{align*}
  & \dataDef{\atom{Tree}~{x}~\cursivetS}{;} \\
  & \dataDef{{,}~{a}~{b}}{;}
  \nl
  &{n}~\infx{TreeSize'}{(\atom{Tree}~{x}~\cursivetS)}~{(\S n')}{:} \\
  &\qquad {n}~\infix{\atomLocal{Go}~\cursivetS}~{n'}{.} \\
  &\qquad {n}~\infix{\atomLocal{Go}~{[]}}~{n}{;} \\
  &\qquad {n}~\infix{\atomLocal{Go}~{[{t}~{\cdot}~\cursivetS]}}~{n''}{:} \\
  &\qquad\qquad {n}~\infx{TreeSize'}{t}~{n'}{.} \\
  &\qquad\qquad {n'}~\infix{\atomLocal[2]{Go}~\cursivetS}~{n''}{.}
  \nl
  & \infx{TreeSize}{t}~{n}{:} \\
  &\qquad \Z~\infx{TreeSize'}{t}~{n}{.}
  \nl
  & (\atom{Tree}~{x}~\cursivetS)~\nfx{Polish}~{[({,}~{x}~{n})~{\cdot}~{p}]}{:} \\
  &\qquad \infx{Length}\cursivexS~{n}{.} \\
  &\qquad \cursivetS~\infx{ConcatMap}{\atom{Polish}}~{p}~\cursivelS{.} \\
  &\qquad {p}~\infx{PolishReads}{n}~\cursivetS~\cursivelS{.}
  \nl
  & {[{({,}~{x}~{n})}~{\cdot}~\cursivexS]}~\nfx{PolishRead}~{\cursivexS'}~{(\atom{Tree}~{x}~\cursivetS)}~{(\S\ell)}{:} \\
  &\qquad \cursivexS~\infx{PolishReads}{n}~{\cursivexS'}~\cursivetS~\cursivelS{.} \\
  &\qquad \infx{Length}\cursivelS~{n}{.} \\
  &\qquad \infx{Sum}\cursivelS~{\ell}{.} \\
  &\qquad \infx{Map}{\atom{TreeSize}~\cursivetS}~\cursivelS{.}
  \nl
  & \cursivexS~\infx{PolishReads}{\Z}~\cursivexS~{[]}~{[]}{;} \\
  & \cursivexS~\infx{PolishReads}{(\S n)}~{\cursivexS''}~{[{t}~{\cdot}~\cursivetS]}~{[{\ell}~{\cdot}~\cursivelS]}{:} \\
  &\qquad \cursivexS~\nfx{PolishRead}~{\cursivexS'}~{t}~{\ell}{.} \\
  &\qquad {\cursivexS'}~\infx{PolishReads}{n}~{\cursivexS'}~\cursivetS~\cursivelS{.}
\end{align*}\end{minipage}%
\endgroup